\documentclass[aps,superscriptaddress,twocolumn,twoside,floatfix,pra,nofootinbib]{revtex4}
\usepackage{times}
\usepackage{epsfig}
\usepackage{amsfonts}
\usepackage{amsmath}
\usepackage{amssymb}
\usepackage{color}
\usepackage{multirow}
\usepackage[normalem]{ulem}
\usepackage{latexsym}
\usepackage{amsfonts}
\usepackage{mathrsfs}
\usepackage{natbib}
\usepackage{verbatim}
\usepackage[T1]{fontenc}
\usepackage{amsthm}

\usepackage[colorlinks=true,linkcolor=blue,citecolor=magenta,urlcolor=blue]{hyperref}

\DeclareMathOperator{\Tr}{tr}
\DeclareMathOperator{\spec}{spec}

\newcommand{\ket}[1]{|#1\rangle}
\newcommand{\bra}[1]{\langle#1|}
\newcommand{\braket}[2]{\langle#1|#2\rangle}
\newcommand{\bracket}[3]{\langle#1|#2|#3\rangle}
\newcommand{\ketbra}[2]{|#1\rangle\langle#2|}
\newcommand{\ketbraq}[1]{\ketbra{#1}{#1}}

\newcommand{\norm}[2][]{#1\lVert #2 #1\rVert}
\newcommand{\abs}[2][]{#1| #2 #1|}
\newcommand{\ave}[2][]{#1\langle #2 #1\rangle}
\newcommand{\nbox}[2][9]{\hspace{#1pt} \mbox{#2} \hspace{#1pt}}
\newcommand{\hc}{^{\dagger}}
\newcommand{\tran}{^\text{T}}
\newcommand{\I}{\openone}
\newcommand{\etad}{\eta^{\textnormal{d}}}
\newcommand{\etar}{\eta^{\textnormal{r}}}
\newcommand{\etap}{\eta^{\textnormal{p}}}
\newcommand{\etajm}{\eta^{\textnormal{jm}}}
\newcommand{\etag}{\eta^{\textnormal{g}}}

\newcommand{\gjm}{g^{\textnormal{jm}}}
\newcommand{\maxval}[1]{\stackrel{\textnormal{#1}}{\leq}}
\newcommand{\Bfunc}[2]{\mathcal{#1}_{d}^{\textnormal{#2}}}
\newcommand{\maxent}[1]{\psi_{#1}^\text{max}}
\newcommand{\lmax}{\lambda_\textnormal{max}}
\newcommand{\lmin}{\lambda_\textnormal{min}}
\newtheorem{theorem}{Theorem}[section]

\newtheorem{corollary}[theorem]{Corollary}
\newtheorem{remark}[theorem]{Remark}
\newtheorem{lemma}[theorem]{Lemma}
\newtheorem{definition}[theorem]{Definition}
\newtheorem{proposition}[theorem]{Proposition}

\newcommand{\ba}{\begin{eqnarray}}
\newcommand{\be}{\begin{equation}}
\newcommand{\ee}{\end{equation}}

\newcommand{\ea}{\end{eqnarray}}
\newcommand{\ban}{\begin{eqnarray*}}
	\newcommand{\ean}{\end{eqnarray*}}

\newcommand{\etal}{{\it{et al.}}}

\newcommand{\cH}{\mathcal{H}}
\newcommand{\cB}{\mathcal{B}}
\newcommand{\cL}{\mathcal{L}}
\newcommand{\cP}{\mathcal{P}}
\newcommand{\bC}{\mathbb{C}}

\definecolor{orange}{rgb}{1,0.5,0}

\begin{document}

\title{Mutually unbiased bases and symmetric informationally complete measurements in Bell experiments}

\author{Armin Tavakoli}
\affiliation{Department of Applied Physics, University of Geneva, 1211 Geneva, Switzerland}

\author{M\'at\'e Farkas}
\affiliation{Institute of Theoretical Physics and Astrophysics,
	National Quantum Information Centre, Faculty of Mathematics,
	Physics and Informatics, University of Gdansk, 80-952 Gdansk, Poland}
	\affiliation{International Centre for Theory of Quantum Technologies, University of Gdansk, 80-308 Gdansk, Poland}

\author{Denis Rosset}
\affiliation{Perimeter Institute for Theoretical Physics, 31 Caroline St. N, Waterloo, Ontario, N2L 2Y5, Canada}

\author{Jean-Daniel Bancal}
\affiliation{Department of Applied Physics, University of Geneva, 1211 Geneva, Switzerland}

\author{J\k{e}drzej Kaniewski}
\affiliation{Faculty of Physics, University of Warsaw, Pasteura 5, 02-093 Warsaw, Poland}

\begin{abstract}
Mutually unbiased bases (MUBs) and symmetric informationally complete projectors (SICs) are crucial to many conceptual and practical aspects of quantum theory. Here, we develop their role in quantum nonlocality by: i) introducing families of Bell inequalities that are maximally violated by $d$-dimensional MUBs and SICs respectively, ii) proving device-independent certification of natural operational notions of MUBs and SICs, and iii) using MUBs and SICs to develop optimal-rate and nearly optimal-rate protocols for device independent quantum key distribution and device-independent quantum random number generation respectively. Moreover, we also present the first example of an extremal point of the quantum set of correlations which admits physically inequivalent quantum realisations. Our results elaborately demonstrate the foundational and practical relevance of the two most important discrete Hilbert space structures to the field of quantum nonlocality.
\end{abstract}

\date{\today}

\maketitle

\textbf{One-sentence summary --- Quantum nonlocality is developed based on the two most celebrated discrete structures in quantum theory. } 
\section{Introduction}
Measurements are crucial and compelling processes at the heart of quantum physics. Quantum measurements, in their diverse shapes and forms, constitute the bridge between the abstract formulation of quantum theory and concrete data produced in laboratories. Crucially, the quantum formalism of measurement processes gives rise to experimental statistics that elude classical models. Therefore, appropriate measurements are indispensible for harvesting and revealing quantum phenomena. Sophisticated manipulation of quantum measurements is both at the heart of the most well-known features of quantum theory such as contextuality \cite{KS} and the violation of Bell inequalities \cite{Bell} as well as its most groundbreaking applications such as quantum cryptography \cite{Qcrypt} and quantum computation \cite{Nielsen}. In the broad landscape of quantum measurements \cite{Qmeas}, certain classes of measurements are outstanding due to their breadth of relevance in foundations of quantum theory and applications in quantum information processing.

Two widely celebrated, intensively studied and broadly useful  classes of measurements are known as mutually unbiased bases (MUBs) and symmetric informationally complete measurements (SICs). Two measurements are said to be mutually unbiased if by preparing any eigenstate of the first measurement and then performing the second measurement, one finds that all outcomes are equally likely \cite{Schwinger}. A typical example of MUBs corresponds to measuring two perpendicular components of the polarisation of a photon. A SIC is a quantum measurement with the largest number of possible outcomes such that all measurement operators have equal magnitude overlaps \cite{Zauner, Renes}. Thus, the former is a relationship between two different measurements whereas the latter is a relationship within a single measurement. Since MUBs and SICs are both conceptually natural, elegant and (as it turns out) practically important classes of measurements, they are often studied in the same context \cite{Wooters1, Grassl1, Beneduci, Bengtsson1, BengtssonCabello, Rastegin}. Let us briefly review their importance to foundational and applied aspects of quantum theory.

MUBs are central to the concept of complementarity in quantum theory i.e.~how the knowledge of one quantity limits (or erases) the knowledge of another quantity (see e.g.~Ref.~\cite{MUBreview} for a review of MUBs). This is often highlighted through variants of the famous Stern--Gerlach experiment in which different Pauli observables are applied to a qubit. For instance, after first measuring (say) $\sigma_x$, we know whether our system points up or down the $x$-axis. If we then measure $\sigma_z$, our knowledge of the outcome of yet another $\sigma_x$ measurement is entirely erased since $\sigma_z$ and $\sigma_x$ are MUBs.  This phenomenon leads to a inherent uncertainty for the outcomes of MUB measurements on all quantum states, which can be formalised in terms of entropic quantities, leading to so-called entropic uncertainty relations.  It is then natural that MUBs give rise to the strongest entropic uncertainties in quantum theory   \cite{MaassenUffink}. Moreover, MUBs play a prominent role in quantum cryptography, where they are employed in many of the most well-known quantum key distribution protocols \cite{BB84, E91, 6state, Cerf02, SARG} as well as in secret sharing protocols \cite{QSS0, QSS1, QSS2}. Their appeal to cryptography stems from the idea that eavesdroppers who measure an eigenstate of one basis in another basis unbiased to it  obtain no useful information, while they also induce a large disturbance in the state which allows their presence to be detected.  Furthermore, complete (i.e.~largest possible in a given dimension) sets of MUBs are tomographically complete and their symmetric properties make them pivotal for quantum state tomography \cite{Wooters, Adamson}. In addition, MUBs are useful for a range of other problems such as quantum random access coding \cite{Ambainis, Aguilar, Tavakoli15, Tavakoli18, Farkas19}, quantum error correction \cite{Gottesman, Calderbank} and entanglement detection \cite{Spengler}. This broad scope of relevance has motivated much efforts towards determining the largest number of MUBs that exist in general Hilbert space dimensions \cite{MUBreview}.

The motivations behind the study of SICs are quite similar to the ones discussed for MUBs. It has been shown that SICs are natural measurements for quantum state tomography \cite{Caves}, which has also prompted several experimental realisations of SICs \cite{Medendorp, Pimenta, Bent}. Also, some  protocols for quantum key distribution derive their success directly from the defining properties of SICs \cite{Renes2, Singapore}, which have also been experimentally demonstrated \cite{Bouchard}. Furthermore, a key property of SICs is that  they have the largest number of outcomes possible while still being extremal measurements i.e.~they cannot be simulated by stochastically implementing other measurements. This gives SICs a central role in a range of applications which include random number generation from entangled qubits \cite{Acin}, certification of non-projective measurements \cite{TavakoliNonProj, Piotr, TavakoliSIC, Massi}, semi-device-independent self-testing \cite{TavakoliNonProj} and entanglement detection  \cite{Shang, Bae}. Moreover, SICs have a key role  in quantum Bayesianism \cite{FuchsRev} and they exhibit interesting connections to several areas of mathematics, for instance Lie and Jordan algebras \cite{Appelby1} and algebraic number theory \cite{Appelby}. Due to their broad interest, much research effort has been directed towards proving the existence of SICs in all Hilbert space dimensions (presently known, at least, up to dimension 121) \cite{Zauner, Renes, ScottGrassl, Scott}. See e.g.~Ref.~\cite{FuchsReview} for a recent review of SICs.

In this work, we broadly investigate MUBs and SICs in the context of Bell nonlocality experiments. In these experiments, two separated observers perform measurements on entangled quantum systems which can produce nonlocal correlations that elude any local hidden variable model \cite{Brunner}. In recent years, Bell inequalities have played a key role in the rise of device-independent quantum information processing where they are used to certify  properties of quantum systems. Naturally, certification of a physical property can be achieved under different assumptions of varying strength. Device-independent approaches offer the strongest form of certification since the only assumptions made are space-like separation and the validity of quantum theory. The advent of device-independent quantum information processing has revived interest in Bell inequalities as these can now be tailored to the purpose of certifying useful resources for quantum information processing. The primary focus of such certification has been on various types of entangled states \cite{Supic}. However, quantum measurements are equally important building blocks for quantum information processing. Nevertheless, our understanding of which arrangements of high-dimensional measurements can be certified in a device-independent manner is highly limited.\footnote{We speak of arrangements of measurements because for a single measurement (acting on a quantum system with no internal structure) no interesting property can be certified. The task becomes non-trivial when at least two measurements are present and we can certify the relation between them.} The simplest approach relies on combining known self-testing results for two-qubit systems, which allows us to certify high-dimensional measurements constructed out of qubit building blocks~\cite{Vidick, Coladangelo}. Alternatively, device-independent certification of high-dimensional structures can be proven from scratch, but to the best of our knowledge only two results of this type have been proven: (a) a triple of MUBs in dimension three~\cite{Jed} and (b) the measurements conjectured to be optimal for the Collins--Gisin--Linden--Massar--Popescu Bell inequality~\cite{Sarkar} (the former is a single result, while the latter is a family parameterised by the dimension $d \geq 2$). None of these results can be used to certify MUBs in dimension $d \geq 4$.


Since mutual unbiasedness and symmetric informational completeness are natural and broadly important concepts in quantum theory, they are prime candidates of interest for such certification in general Hilbert space dimensions. This challenge is increasingly relevant due to the broader experimental advances towards high-dimensional systems along the frontier of quantum information theory. This is also reflected in the fact that recent experimental implementations of MUBs and SICs can go well beyond the few lowest Hilbert space dimensions \cite{Bent, Bouchard, Romero}.

Focusing on mutual unbiasedness and symmetric informational completeness, we solve the above challenges. To this end, we first construct Bell inequalities that are maximally violated using a maximally entangled state of local dimension $d$ and, respectively, a pair of $d$-dimensional MUBs and a $d$-dimensional SIC. In the case of MUBs, we show that the maximal quantum violation of the proposed Bell inequality device-independently certifies that the measurements satisfy an operational definition of mutual unbiasedness, and also that the shared state is essentially a maximally entangled state of local dimension $d$. Similarly, in the case of SICs, we find that the maximal quantum violation device-independently certifies that the measurements satisfy an analogous operational definition of symmetric informational completeness. Moreover, we also show that our Bell inequalities are useful in two practically relevant tasks. For the case of MUBs, we consider a scheme for device-independent quantum key distribution and prove a key rate of $\log d$ bits, which is  optimal for any protocol that extracts key from a $d$-outcome measurement. For SICs, we construct a scheme for device-independent random number generation. For two-dimensional SICs, we obtain the largest amount of randomness possible for any protocol based on qubits. For three-dimensional SICs, we obtain more randomness than can be obtained in any protocol based on projective measurements and quantum systems of dimension up to seven. For low dimensions, we numerically show that both protocols are robust to noise, which is imperative to any experiment. The implementation of these two protocols involves performing a Bell-type experiment,  estimating the outcome statistics and computing the resulting Bell inequality violation. The efficiency and security of the protocol is then deduced only from the observed Bell inequality violation, i.e.~it does not require a complete characteristion of the devices. Device-independent protocols can in principle be implemented on any experimental platform suitable for Bell nonlocality experiments, such as entangled spins \cite{Hensen}, entangled photons \cite{Shalm, Giustina} and entangled atoms \cite{Harald}.

\section{Bell inequalities for mutually unbiased bases}
The task of finding Bell inequalities which are maximally violated by mutually unbiased bases for $d \geq 3$ has been attempted several times~\cite{Bechmann, ji08a, liang09a, lim10a} but with limited success. The only convincing candidate is the inequality corresponding to $d = 3$ studied in Ref.~\cite{ji08a} and even then there is only numerical evidence (no analytical proof is known). Some progress has been made in Ref.~\cite{Jed}, which considers the case of prime $d$ and proposes a family of Bell inequalities maximally violated by a specific set of $d$ MUBs in dimension $d$. These inequalities, however, have two drawbacks: (a) there is no generalisation to the case of non-prime $d$ and (b) even for the case of prime $d$ we have no characterisation of the quantum realisations that achieve the maximal violation.

In this work we present a family of Bell inequalities in which the maximal quantum  violation is achieved with a maximally entangled state and any pair of $d$-dimensional MUBs. These Bell inequalities have been constructed so that their maximal quantum violation can be computed analytically which then enables us to obtain a detailed characterisation of the optimal realisations. As a result we discover a new, intermediate form of device-independent certification.

We formally define a pair of MUBs as two orthonormal bases on a $d$-dimensional Hilbert space $\bC^{d}$, namely $\{\ket{e_j}\}_{j=1}^d$ and $\{\ket{f_k}\}_{k=1}^d$, with the property that 
\begin{equation}\label{MUB}
	|\braket{e_j}{f_k}|^2=\frac{1}{d}
\end{equation}
for all $j$ and $k$. The constant on the right-hand-side is merely a consequence of the two bases being normalised. To this end, consider a bipartite Bell scenario parameterised by an integer $d\geq 2$. Alice randomly receives one of $d^2$ possible inputs labelled by $x\equiv x_1x_2\in[d]^2$ (where $[s]\equiv \{1,\ldots, s\}$) and produces a ternary output labelled by $a\in\{1,2, \perp\}$. Bob receives a random binary input labelled by $y\in\{1,2\}$ and produces a $d$-valued output labelled by $b\in[d]$. The joint probability distribution in the Bell scenario is denoted by $p(a,b|x,y)$ and the scenario is illustrated in Figure~\ref{FigMUBscenario}.

\begin{figure}
	\centering
	\includegraphics[width=\columnwidth]{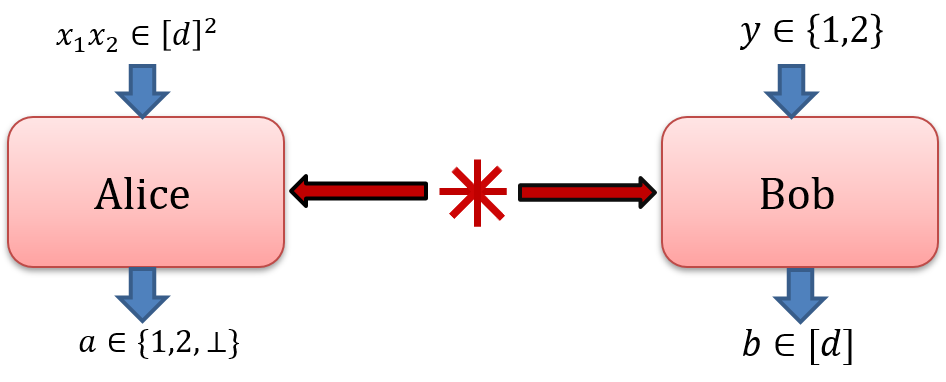}
	\caption{Bell scenario for two MUBs of dimension $d$. Alice receives one of $d^2$ inputs and produces a ternary output  while Bob receives a binary input and produces a $d$-valued output.}\label{FigMUBscenario}
\end{figure}

To make our choice of Bell functional transparent, we will phrase it as a game in which Alice and Bob collectively win or lose points.  If Alice outputs $a = \perp$, no points will be won or lost. If she outputs $a \in \{1, 2\}$, points will be won or lost if $b = x_{y}$. More specifically, Alice and Bob win a point if $a = y$ and lose a point if $a=\bar{y}$, where the bar-sign flips the value of $y\in\{1,2\}$. This leads to the score
\begin{equation}\label{Md}
\Bfunc{R}{MUB} \equiv \sum_{x, y} p(a=y,b=x_y|x,y)-p(a=\bar{y},b=x_y|x,y),
\end{equation}
where the sum goes over $x = x_{1}x_{2} \in [d]^{2}$ and $y \in \{1, 2\}$.

At this point the outcome $a=\perp$ might seem artificial, so let us show why it plays a crucial role in the construction of the game. To this end, we use intuition based on the hypothetical case in which Alice and Bob share a maximally entangled state 
\begin{equation}\label{maxent}
\ket{\maxent{d}}=\frac{1}{\sqrt{d}}\sum_{k=1}^{d}\ket{k,k}.
\end{equation}
The reason we consider the maximally entangled state is that we aim to tailor the Bell inequalities so that this state is optimal. Then, we would like to ensure that Alice, via her measurement and for her outcomes $a\in\{1,2\}$, remotely prepares Bob in a pure state. This would allow Bob to create stronger correlations as compared to the case of Alice remotely preparing his system is a mixed state. Hence, this corresponds to Alice's outcomes $a\in\{1,2\}$ being represented by rank-one projectors. Since the subsystems of $\ket{\maxent{d}}$ are maximally mixed, it follows that $p(a=1|x)=p(a=2|x)=1/d$ $\forall x$. Thus, we want to motivate Alice to employ a strategy in which she outputs $a=\perp$ with probability $p(a=\perp|x)=1-2/d$. Our tool for this purpose is to introduce a penalty. Specifically, whenever Alice decides to output $a\in\{1,2\}$, she is penalised by losing $\gamma_d$ points. Thus, the total score  (the Bell functional) reads
\begin{equation}\label{Fd}
\Bfunc{S}{MUB}\equiv \Bfunc{R}{MUB}-\gamma_d\sum_x \big( p(a=1|x)+p(a=2|x)\big).
\end{equation}
Now, outputting $a\in\{1,2\}$ contributes towards  $\Bfunc{R}{MUB}$ but also causes a penalty $\gamma_d$. Therefore, we expect to see a trade-off between $\gamma_d$ and the rate at which Alice outputs $a=\perp$. We must suitably choose $\gamma_d$ such that Alice's best strategy is to output $a=\perp$ with (on average over $x$) the desired probability $p(a=\perp|x)=1-2/d$. This accounts for the intuition that leads us to the following Bell inequalities for MUBs.

\begin{theorem}[Bell inequalities for MUBs]\label{thmBellForMUBs}
	The Bell functional $\Bfunc{S}{MUB}$ in Eq.~\eqref{Fd} with 
	\begin{equation}
	\gamma_d=\frac{1}{2}\sqrt{\frac{d-1}{d}},
	\end{equation}
	obeys the tight local bound 
	\begin{equation}\label{CMUB}
	\Bfunc{S}{MUB} \maxval{LHV} 2\left(d-1\right)\left(1-\frac{1}{2}\sqrt{\frac{d-1}{d}}\right),
	\end{equation}
	and the quantum bound
	\begin{equation}\label{QMUB}
	\Bfunc{S}{MUB} \maxval{Q} \sqrt{d\left(d-1\right)} .
	\end{equation}
	Moreover, the quantum bound can be saturated by sharing a maximally entangled state of local dimension $d$ and Bob performing measurements in any two mutually unbiased bases. 
\end{theorem}
\begin{proof} A complete proof is presented in Supplementary Material (SM, section I.A). The essential ingredient to obtain the bound in Eq.~\eqref{QMUB} is the Cauchy--Schwarz inequality. Furthermore, for local models, by inspecting the symmetries of the Bell functional $\Bfunc{S}{MUB}$, one finds that the local bound can be attained by Bob always outputting $b=1$. This greatly simplifies the evaluation of the bound in Eq.~\eqref{CMUB}.

To see that the bound in Eq.~\eqref{QMUB} can be saturated in quantum theory, let us evaluate the Bell functional for a particular quantum realisation. Let $\ket{\psi}$ be the shared state, $\{P_{x_1}\}_{x_1=1}^d$ and $\{Q_{x_2}\}_{x_2=1}^d$ be the measurement operators of Bob corresponding to $y=1$ and $y=2$ respectively and $A_{x}$ be the observable of Alice defined as the difference between Alice's outcome-one and outcome-two measurement operators, i.e.~$A_x=A_x^1-A_x^2$. Then, the Bell functional reads
\begin{align}\label{eq:MUBfuncop}
\Bfunc{S}{MUB} =\sum_{x} \bracket{\psi}{A_x\otimes \left(P_{x_1}-Q_{x_2}\right)-\gamma_d\left(A_x^1+A_x^2\right)\otimes \I}{\psi}.
\end{align}
Now, we choose the maximally entangled state of local dimension $d$, i.e.~$\ket{\psi}=\ket{\maxent{d}}$, and define Bob's measurements as rank-one projectors  $P_{x_1}=\ketbra{\phi_{x_1}}{\phi_{x_1}}$ and $Q_{x_2}=\ketbra{\varphi_{x_2}}{\varphi_{x_2}}$ which correspond to MUBs, i.e.~$|\braket{\phi_{x_1}}{\varphi_{x_2}}|^2=1/d$. Finally, we choose Alice's observables as $A_x=\sqrt{d/(d-1)}(P_{x_1}-Q_{x_2})\tran$, where the pre-factor ensures the correct normalisation and $\tran$ denotes the transpose in the standard basis. Note that $A_{x}$ is a rank-two operator; the corresponding measurement operator $A_x^1$ ($A_x^2$) is a rank-one projector onto the eigenvector of $A_x$ associated to the positive (negative) eigenvalue. Since the subsystems of $\ket{\psi_d^\text{max}}$ are maximally mixed, this implies  $\bracket{\psi_d^\text{max}}{(A_x^1+A_x^2)\otimes \openone}{\psi_d^\text{max}}=2/d$. Inserting all this into the above quantum model and exploiting the fact that for any linear operator $O$ we have $O \otimes \I \ket{\maxent{d}} = \I \otimes O\tran \ket{\maxent{d}}$, we straightforwardly saturate the bound in Eq.~\eqref{QMUB}.
\end{proof}

We remark that for the case of $d=2$ one could also choose $\gamma_2=0$ and retain the property that qubit MUBs are optimal. In this case the marginal term is not necessary, because in the optimal realisation Alice never outputs $\perp$. Then, the quantum bound becomes $2\sqrt{2}$ and the local bound becomes $2$. The resulting Bell inequality resembles the Clauser--Horne--Shimony--Holt (CHSH) inequality \cite{CHSH}, not just because it gives the same local and quantum values, but also because the optimal realisations coincide. More specifically, the measurements of Bob are precisely the optimal CHSH measurements, whereas the four measurements of Alice correspond to two pairs of optimal CHSH measurements.

\section{Device-independent certification of  mutual unbiasedness}

Theorem~\ref{thmBellForMUBs} establishes that a pair of MUBs of any dimension can generate a maximal quantum violation in a Bell inequality test. We now turn to the converse matter, namely that of device-independent certification. Specifically, given that we observe the maximal quantum violation, i.e.~equality in Eq.~\eqref{QMUB}, what can be said about the shared state and the measurements? Since the measurement operators can only be characterised on the support of the state, to simplify the notation let us assume that the marginal states of Alice and Bob are full-rank.\footnote{Note that this is not a physical assumption but a mathematical convention which simplifies the notation in the rest of this work. Whenever the marginal state is not full-rank, the local Hilbert space naturally decomposes as a direct sum of two terms, where the state is only supported on one of them. Clearly, the measurement operators can only be characterised on the support of the state and that is precisely what we achieve. This convention allows us to only write out the part that can be characterised and leave out the rest.}

\begin{theorem}[Device-independent certification]\label{thmDIMUB}
	The maximal quantum value of the Bell functional $\Bfunc{S}{MUB}$ in Eq.~\eqref{Fd} implies that
	\begin{itemize}
		\item There exist local isometries which allow Alice and Bob to extract a maximally entangled state of local dimension $d$.
		\item If the marginal state of Bob is full-rank, the two $d$-outcome measurements he performs satisfy the relations
		\begin{equation}\label{sandwich}
		P_{a}=dP_{a} Q_{b} P_{a} \qquad \text{and}  \qquad  Q_{b}=dQ_{b} P_{a} Q_{b},
		\end{equation}
		for all $a$ and $b$.
	\end{itemize}
\end{theorem}
\begin{proof}
The proof is detailed in SM (section I.A). Here, we briefly summarise the part concerning Bob's measurements. Since the Cauchy--Schwarz inequality is the main tool for proving the quantum bound in Eq.~\eqref{QMUB}, saturating it implies that also the Cauchy--Schwarz inequality is saturated. This allows us to deduce that the measurements of Bob are projective and moreover we obtain the following optimality condition:  
\begin{equation}\label{crossrelation}
A_x\otimes \I \ket{\psi}=\I \otimes \sqrt{\frac{d}{d-1}}\left(P_{x_1}-Q_{x_2}\right)\ket{\psi},
\end{equation}
for all $x_{1}, x_{2} \in [d]$ where the factor $\sqrt{d/(d-1)}$ can be regarded as a normalisation. Since we do not attempt to certify the measurements of Alice, we can without loss of generality assume that they are projective. This implies that the spectrum of $A_x$ only contains $\{+1,-1, 0\}$ and therefore $(A_x)^3=A_x$. This allows us to obtain a relation that only contains Bob's operators. Tracing out Alice's system and subsequently eliminating the marginal state of Bob (it is assumed to be full-rank) leads to
\begin{equation}
P_{x_1}-Q_{x_2}=\frac{d}{d-1}\left(P_{x_1}-Q_{x_2}\right)^3.
\end{equation}
Expanding this relation and then using projectivity and the completeness of measurements, one recovers the result in Eq.~\eqref{sandwich}.
\end{proof}
We have shown that observing the maximal quantum value of $\mathcal{S}_d^\text{MUB}$ implies that the measurements of Bob satisfy the relations given in Eq.~\eqref{sandwich}. It is natural to ask whether a stronger conclusion can be derived, but the answer turns out to be negative. In SM (section I.B) we show that any pair of $d$-outcome measurements (acting on a finite-dimensional Hilbert space) satisfying the relations in Eq.~\eqref{sandwich} is capable of generating the maximal Bell inequality violation. For $d = 2, 3$ the relations given in Eq.~\eqref{sandwich} imply that the unknown measurements correspond to a direct sum of MUBs (see SM section II.C) and since in these dimension there exists only a single pair of MUBs (up to unitaries and complex conjugation), our results imply a self-testing statement of the usual kind. However, since in higher dimensions not all pairs of MUBs are equivalent~\cite{Brierley}, our certification statement is less informative than the usual formulation of self-testing. In other words, our inequalities allow us to self-test the quantum state, but we cannot completely determine the measurements (see Refs.~\cite{Jeba, Kaniewski2} for related results). Note that we could also conduct a device-independent characterisation of the measurements of Alice. In fact, Eq.~(61) from the SM enables us to relate the measurements of Alice to the measurements of Bob, which we have already characterised. However, since we do not expect the observables of Alice to satisfy any simple algebraic relations and since they are not directly relevant for the scope of this work (namely MUBs and SICs), we do not pursue this direction.

The certification provided in Theorem~\ref{thmDIMUB} turns out to be sufficient to determine all the probabilities $p(a,b|x,y)$ that arise in the Bell experiment (see SM section I.C), which means that the maximal quantum value of $\Bfunc{S}{MUB}$ is achieved by a single probability distribution. Due to the existence of inequivalent pairs of MUBs in certain dimensions (e.g.~for $d = 4$), this constitutes the first example of an extremal point of the quantum set which admits inequivalent quantum realisations.\footnote{Recall that the notion of equivalence we employ is precisely the one that appears in the context of self-testing, i.e.~we allow for additional degrees of freedom, local isometries and a transposition.}

It is important to understand the relation between the condition given in Eq.~\eqref{sandwich} and the concept of MUBs. Naturally, if $\{P_{a}\}_{a=1}^d$ and $\{Q_{b}\}_{b=1}^d$ are $d$-dimensional MUBs, the relations \eqref{sandwich} are satisfied. Interestingly, however, there exist solutions to Eq.~\eqref{sandwich} which are neither MUBs nor direct sums thereof. While, as mentioned above, for $d = 2, 3$ one can show that any measurements satisfying the relations \eqref{sandwich} must correspond to a direct sum of MUBs, this is not true in general. For $d=4,5$ we have found explicit examples of measurement operators satisfying Eq.~\eqref{sandwich} which cannot be written as a direct sum of MUBs. In fact, they cannot even be transformed into a pair of MUBs via a completely positive unital map (see SM section II for details). These results beg the crucial question: how should one interpret the condition given in Eq.~\eqref{sandwich}?

To answer this question we resort to an operational formulation of what it means for two measurements to be mutually unbiased. An operational approach must rely on observable quantities (i.e.~probabilities), as opposed to algebraic relations between vectors or operators. This notion, which we refer to as mutually unbiased measurements (MUMs), was recently formalised by Tasca~\etal~\cite{Tasca}. Note that in what follows we use the term ``eigenvector'' to refer to eigenvectors corresponding to non-zero eigenvalues.
\begin{definition}[Mutually unbiased measurements]\label{defOp}
We say that two $n$-outcome measurements $\{P_a\}_{a=1}^n$ and $\{Q_b\}_{b=1}^n$ are mutually unbiased if they are projective and the following implications hold:
\begin{align}\label{opMUB}\nonumber
& \bracket{\psi}{P_a}{\psi}=1 \Rightarrow \bracket{\psi}{Q_b}{\psi}=\frac{1}{n}\\
&  \bracket{\psi}{Q_b}{\psi}=1 \Rightarrow \bracket{\psi}{P_a}{\psi}=\frac{1}{n},
\end{align} 
for all $a$ and $b$. That is, two projective measurements are mutually unbiased if the eigenvectors of one measurement give rise to a uniform outcome distribution for the other measurement.
\end{definition}
Note that this definition captures precisely the intuition behind MUBs without the need to specify the dimension of the underlying Hilbert space. Interestingly enough, MUMs admit a simple algebraic characterisation.

\begin{theorem}\label{thmOP}
Two $n$-outcome measurements $\{P_a\}_{a=1}^n$ and $\{Q_b\}_{b=1}^n$ are mutually unbiased if and only if
\begin{equation}
P_{a}=nP_{a} Q_{b} P_{a} \qquad \text{and}  \qquad  Q_{b}=nQ_{b} P_{a} Q_{b},
\end{equation}
for all $a$ and $b$.
\end{theorem}
\begin{proof}
Let us first assume that the algebraic relations hold. By summing over the middle index, one finds that both measurements are projective. Moreover, if $\ket{\psi}$ is an eigenvector of $P_a$, then $\bracket{\psi}{Q_b}{\psi} = \bracket{\psi}{P_a Q_b P_a}{\psi} = \frac{1}{n} \bracket{\psi}{P_a}{\psi} = \frac{1}{n}$.
By symmetry, the analogous property holds if $\ket{\psi}$ is an eigenvector of $Q_b$.

Conversely, let us show that MUMs must satisfy the above algebraic relations. Since $\sum_{a} P_{a} = \I$ we can choose an orthonormal basis of the Hilbert space composed only of the eigenvectors of the measurement operators. Let $\{\ket{e_j^a}\}_{a,j}$ be an orthonormal basis, where $a \in [n]$ tells us which projector the eigenvector corresponds to and $j$ labels the eigenvectors within a fixed projector (if $P_{a}$ has finite rank, then $j \in [ \Tr P_{a} ]$, otherwise $j \in \mathbb{N}$). By construction for such a basis we have
$P_a\ket{e_j^{a'}}=\delta_{aa'}\ket{e_j^a}$. To show that $P_{a} = nP_{a} Q_{b} P_{a}$ it suffices to show that the two operators have the same coefficients in this basis. Since
\begin{align}
\bracket{e_j^{a'}}{nP_aQ_bP_a}{e_k^{a''}}&=n\delta_{aa'}\delta_{aa''}\bracket{e_j^a}{Q_b}{e_k^a},\\
 \bracket{e_j^{a'}}{P_a}{e_k^{a''}}&=\delta_{aa'}\delta_{aa''}\delta_{jk}
\end{align}
it suffices to show that $n\bracket{e_j^a}{Q_b}{e_k^a}=\delta_{jk}$. For $j=k$ this is a direct consequence of the definition in Eq.~\eqref{opMUB}. To prove the other case, define $\ket{\phi_{\theta}}=\left(\ket{e_j^a}+\mathrm{e}^{\mathrm{i} \theta} \ket{e_k^a}\right)/\sqrt{2}$, for $\theta \in [ 0, 2\pi )$. Since $P_a \ket{\phi_{\theta}}=\ket{\phi_{\theta}}$, we have $\bracket{\phi_{\theta}}{Q_b}{\phi_{\theta}}=1/n$. Writing this equality out gives
\begin{equation}
\frac{1}{n}=\frac{1}{2}\left(\frac{2}{n}+\mathrm{e}^{\mathrm{i} \theta} \bracket{e_j^a}{Q_b}{e_k^a} + \mathrm{e}^{-\mathrm{i} \theta} \bracket{e_k^a}{Q_b}{e_j^a}\right).
\end{equation}
Choosing $\theta = 0$ implies that the real part of $\bracket{e_j^a}{Q_b}{e_k^a}$ vanishes, while $\theta = \pi/2$ implies that the imaginary part vanishes. Proving the relation $Q_{b} = n Q_{b} P_{a} Q_{b}$ proceeds in an analogous fashion.
\end{proof}

Theorem~\ref{thmOP} implies that the maximal violation of the Bell inequality for MUBs certifies precisely the fact the Bob's measurements are mutually unbiased. To provide further evidence that MUMs constitute the correct device-independent generalisation of MUBs, we give two specific situations in which the two objects behave in the same manner.

Maassen and Uffink considered a scenario in which two measurements (with a finite number of outcomes) are performed on an unknown state. Their famous uncertainty relation provides a state-independent lower bound on the sum of the Shannon entropies of the resulting distributions~\cite{MaassenUffink}. While the original result only applies to rank-one projective measurements, a generalisation to non-projective measurements reads \cite{Krishna}
\begin{equation}
H(P)+H(Q)\geq -\log c, 
\end{equation}
where $H$ denotes the Shannon entropy and $c=\max_{a,b} \norm{\sqrt{P_a}\sqrt{Q_b}}^2$ where $\norm{\cdot}$ is the operator norm. If we restrict ourselves to rank-one projective measurements on a Hilbert space of dimension $d$, one finds that the largest uncertainty, corresponding to $c=1/d$, is obtained only by MUBs. It turns out that precisely the same value is achieved by any pair of MUMs with $d$ outcomes regardless of the dimension of the Hilbert space:
\begin{multline}
c=\max_{a,b} \norm{\sqrt{P_a} \sqrt{Q_b}}^2=\max_{a,b} \norm{P_aQ_b}^2\\
=\max_{a,b} \norm{P_aQ_bP_a}=\max_{a} \norm{P_a/d}=\frac{1}{d}.
\end{multline}

A closely related concept is that of measurement incompatibility, which captures the phenomenon that two measurements cannot be performed simultaneously on a single copy of a system. The extent to which two measurements are incompatible can be quantified e.g.~by so-called incompatibility robustness measures \cite{HMZ16}. In SM (section II.D), we show that according to these measures MUMs are exactly as incompatible as MUBs. Moreover, we can show that for the so-called generalised incompatibility robustness \cite{Haapasalo}, MUMs are among the most incompatible pairs of $d$-outcome measurements.

\section{Application: Device-independent Quantum Key Distribution}\label{sec:QKD}
The fact that the maximal quantum violation of the Bell inequalities introduced above requires a maximally entangled state and MUMs, and moreover that it is achieved by a unique probability distribution, suggests that these inequalities might be useful for device-independent quantum information processing. In the task of quantum key distribution~\cite{BB84, E91, QKD} Alice and Bob aim to establish a shared data set (a key) that is secure against a malicious eavesdropper. Such a task requires the use of incompatible measurements, and MUBs in dimension $d = 2$ constitute the most popular choice. Since in the ideal case the measurement outcomes of Alice and Bob that contribute to the key should be perfectly correlated, most protocols are based on maximally entangled states. In the device-independent approach to quantum key distribution, the amount of key and its security is deduced from the observed Bell inequality violation.

We present a proof-of-principle application to device-independent quantum key distribution based on the quantum nonlocality witnessed through the Bell functional in Eq.~\eqref{Fd}. In the ideal case, Alice and Bob follow the strategy that gives them the maximal violation, i.e.~they share a maximally entangled state of local dimension $d$ and Bob measures two MUBs. To generate the key we provide Alice with an extra setting that produces outcomes which are perfectly correlated with the outcomes of the first setting of Bob. This will be the only pair of settings from which the raw key will be extracted and let us denote them by $x=x^*$ and $y=y^*=1$. In most rounds of the experiment, Alice and Bob choose these settings and therefore contribute towards the raw key. However, to ensure security, a small number of rounds is used to evaluate the Bell functional. In these rounds, which are chosen at random, Alice and Bob randomly choose their measurement settings. Once the experiment is complete, the resulting value of the Bell functional is used to infer the amount of secure raw key shared between Alice and Bob. The raw key can then be turned into the final key by standard classical post-processing. For simplicity, we consider only individual attacks and moreover we focus on the limit of asymptotically many rounds in which fluctuations due to finite statistics can be neglected.

The key rate, $K$, can be lower bounded by \cite{Masanes}
\begin{equation}\label{keyrate}
K\geq -\log \left( P_g^\beta\right)-H(B_{y^*}|A_{x^*}),
\end{equation}
where $P_g^\beta$ denotes the highest probability that the eavesdropper can correctly guess Bob's outcome when his setting is $y^*$ given that the Bell inequality value $\beta$ was observed, and $H(\cdot|\cdot)$ denotes the conditional Shannon entropy. The guessing probability $P_g^\beta$ is defined as
\begin{equation}\label{guess}
P_g^\beta \equiv \sup \bigg\{ \sum_{c=1}^d \bracket{\psi_\text{ABE}}{\I \otimes P_c \otimes E_c}{\psi_\text{ABE}} \bigg\},
\end{equation}
where $\{ E_{c} \}_{c = 1}^{d}$ is the measurement employed by the eavesdropper to produce her guess, the expression inside the curly braces is the probability that her outcome is the same as Bob's for a particular realisation and the supremum is taken over all quantum realisations (the tripartite state and measurements of all three parties) compatible with the observed Bell inequality value $\beta$.

Let us first focus on the key rate in a noise-free scenario, i.e.~in a scenario in which $\Bfunc{S}{MUB}$ attains its maximal value. Then, one straightforwardly arrives at the following result.
\begin{theorem}[Device-independent key rate]
In the noiseless case the quantum key distribution protocol based on $\Bfunc{S}{MUB}$ achieves the key rate of
	\begin{equation}\label{qkey}
	K=\log d
	\end{equation}
	for any integer $d\geq 2$. 
\end{theorem}
\begin{proof}
In the noiseless case, Alice and Bob observe exactly the correlations predicted by the ideal setup. In this case the outcomes for settings $(x^{*}, y^{*})$ are perfectly correlated which implies that $H(B_{y^*}|A_{x^*}) = 0$. Therefore, the only non-trivial task is to bound the guessing probability.

Since the actions of the eavesdropper commute with the actions of Alice and Bob, we can assume that she performs her measurement first. If the probability of the eavesdropper observing outcome $c \in [d]$, which we denote by $p(c)$, is non-zero, then the (normalised) state of Alice and Bob conditioned on the eavesdropper observing that outcome is given by:
\begin{equation}
\rho_\text{AB}^{(c)} = \frac{1}{p(c)} \Tr_{C} \big[ (\I \otimes \I \otimes E_{c}) \ketbra{\psi_\text{ABE}}{\psi_\text{ABE}} \big].
\end{equation}
Now Alice and Bob share one of the post-measurement states $\rho_\text{AB}^{(c)}$ and when they perform their Bell inequality test, they will obtain different distributions depending on $c$, which we write as $p_c(a,b|x,y)$. However, since the statistics achieve the maximal quantum value of $\Bfunc{S}{MUB}$ and we have previously shown that the maximal quantum value is achieved by a single probability point, all the probability distributions $p_c(a,b|x,y)$ must be the same. Moreover, we have shown that for this probability point, the marginal distribution of outcomes on Bob's side is uniform over $[d]$ for both inputs. This implies that
\begin{equation}
P_g = \sum_{c=1}^{d} p(c) p_{c}(b = c | y = 1) =\frac{1}{d},
\end{equation}
because $p_{c}(b = c | y = 1) = p(b = c | y = 1) = \frac{1}{d}$ for all $c$.
\end{proof}

We remark that the argument above is a direct consequence of a more general result which states that if a bipartite probability distribution is a nonlocal extremal point of the quantum set, then no external party can be correlated with the outcomes~\cite{Franz}. 

It is interesting to note that the obtained key rate is the largest possible for general setups in which the key is generated from a $d$-outcome measurement. Also, the key rate is optimal for all protocols based on a pair of entangled $d$-dimensional systems subject to projective measurements. This follows from the fact that projective measurements in $\bC^d$ cannot have more than $d$ outcomes. It has recently been shown that the same amount of randomness can be generated using a modified version of the Collins--Gisin--Linden--Massar--Popescu inequalities~\cite{Sarkar}, but note that the measurements used there do not correspond to mutually unbiased bases (except for the special case of $d = 2$).

Let us now depart from the noise-free case and estimate the key rate in the presence of noise. To ensure that both the guessing probability and the conditional Shannon entropy can be computed in terms of a single noise parameter, we have to introduce an explicit noise model. We employ the standard approach in which the measurements remain unchanged, while the maximally entangled state is replaced with an isotropic state given by
\begin{equation}\label{isotropic}
\rho_v=v \ketbra{\maxent{d}}{\maxent{d}}+\frac{1-v}{d^2}\I,
\end{equation}
where $v\in[0,1]$ is the visibility of the state. Using this state and the ideal measurements for Alice and Bob, the relation between $v$ and $\Bfunc{S}{MUB}$ can be easily derived from \eqref{eq:MUBfuncop}, namely,
\begin{equation}
v=\frac{1}{2}\left(1+\frac{\Bfunc{S}{MUB}}{\sqrt{d(d-1)}}\right).
\end{equation}
Utilising this formula, we also obtain the value of $H(B_{y^*}|A_{x^*})$ as a function of the Bell violation. The remaining part of \eqref{keyrate} is the guessing probability \eqref{guess}. In the case of $d=3$, we proceed to bound this quantity through semidefinite programming.

Concretely, we implement the three-party semidefinite relaxation \cite{NPA} of the set of quantum correlations at local level one\footnote{We attribute one operator to each outcome of Bob and the eavesdropper, but only take into account  the first two outcomes of Alice.}. This results in a moment matrix of size $532\times532$ with $15617$ variables. The guessing probability is directly given by the sum of three elements of the moment matrix. It can then be maximised under the constraints that the value of the Bell functional $\mathcal{S}_3^\text{MUB}$ is fixed and the moment matrix is positive semidefinite. However, we notice that this problem is invariant under the following relabelling: $b\to \pi(b)$ for $y=1$, $c\to \pi(c)$, $x_1\to\pi(x_1)$, where $\pi\in S_3$ is a permutation of three elements. Therefore, it is possible to simplify this semidefinite program by requiring the matrix to be invariant under the group action of $S_{3}$ on the moment matrix (i.e.~it is a Reynolds matrix)~\cite{Cai, TavakoliSIC, Rosset}. This reduces the number of free variables in the moment matrix to $2823$. With the SeDuMi~\cite{sedumi} solver, this lowers the precision ($1.1\times 10^{-6}$ instead of $8.4\times 10^{-8}$), but speeds up the computation ($155$s  instead of $8928$s) and requires less memory ($0.1$GB  instead of $5.5$GB). For the maximal value of $\mathcal{S}_d^\text{MUB}$, we recover the noise-free result of $K=\log 3$ up to the fifth digit. Also, we have a key rate of at least one bit when $\Bfunc{S}{MUB} \gtrsim 2.432$ and a non-zero key rate when $\Bfunc{S}{MUB} \gtrsim 2.375$. The latter is close to the local bound, which is $\Bfunc{S}{MUB}\approx 2.367$. The resulting lower bound on the key rate as a function of the Bell inequality violation is plotted in Fig.~\ref{FigQKD}.

\begin{figure}
	\centering
	\includegraphics[width=\columnwidth]{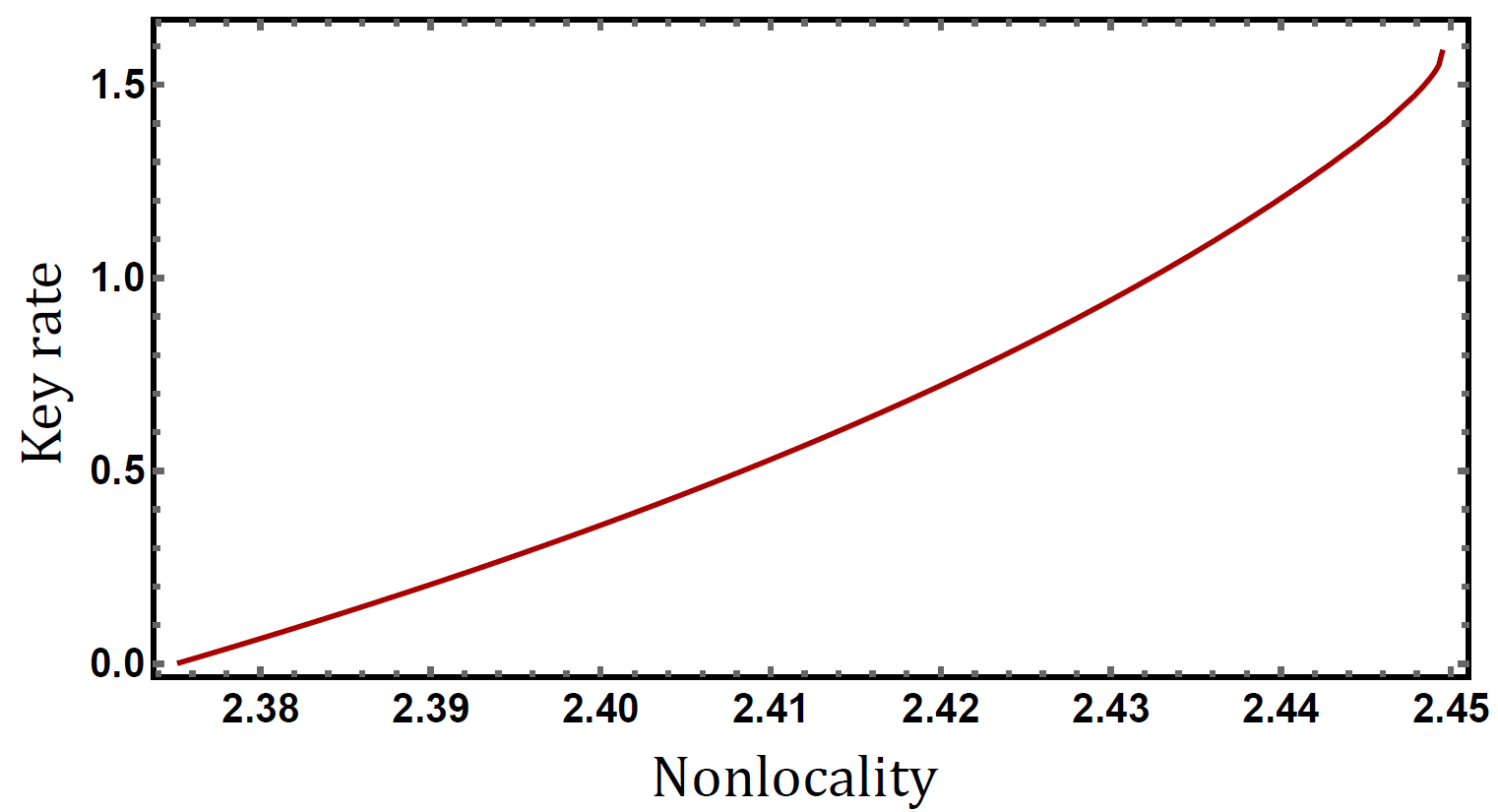}
	\caption{Lower bound on the key rate $K$ in the asymptotic limit versus the value of the Bell functional $\mathcal{S}_3^\text{MUB}$.}\label{FigQKD}
\end{figure}


\section{Nonlocality for symmetric informational completeness}
We now shift our focus from MUBs to SICs. We construct Bell inequalities whose maximal quantum violations are achieved with SICs. We formally define a SIC as a set of $d^2$ unit vectors in  $\bC^d$, namely $\{\ket{r_j}\}_{j=1}^{d^2}$, with the property that 
\begin{equation}\label{sicqq}
	|\braket{r_j}{r_k}|^2=\frac{1}{d+1}
\end{equation} 
for all $j\neq k$, where the constant on the right-hand-side is fixed by normalisation. The reason for there being precisely $d^2$ elements in a SIC\footnote{By SIC (in singular) we refer to one  set of symmetric informationally complete projectors. By SICs (in plural), we refer to all such sets.} is that this is the largest number of unit vectors in $\mathbb{C}^d$ that could possibly admit the uniform overlap property \eqref{sicqq}. Moreover, we formally distinguish between a SIC as the presented set of rank-one projectors and a SIC-POVM (positive operator-valued measure) which is the generalised quantum measurement with $d^2$ possible outcomes corresponding to the sub-normalised projectors $\{\frac{1}{d}\ketbra{r_k}{r_k}\}_{k=1}^{d^2}$.

Since the treatment of SICs in Bell nonlocality turns out to be more challenging than for the case of MUBs, we first establish the relevance of SICs in a simplified Bell scenario subject to additional constraints. This serves as a stepping stone to a subsequent relaxation which gives a standard (unconstrained) Bell inequality for SICs. We then focus on the device-independent certification power of these inequalities, which leads us to an operational notion of symmetric informational completeness. Finally, we extend the Bell inequalities so that their maximal quantum violations are achieved with both projectors forming SICs and a single generalised measurement corresponding to a SIC-POVM.

\subsection{Stepping stone: quantum correlations for SICs}
Consider a Bell scenario, parameterised by an integer $d\geq 2$, involving two parties Alice and Bob who share a physical system.  Alice receives an input labelled by a tuple $(x_{1}, x_{2})$ representing one of $\binom{d^2}{2}$ possible inputs, which we collectively refer to as $x = x_{1} x_{2}$. The tuple is randomly taken from the set $\text{Pairs}(d^2) \equiv \{ x | x_{1},x_{2} \in [d^2] \text{ and } x_{1} < x_{2} \}$. Alice performs a measurement on her part of the shared system and produces a ternary output labelled by $a\in\{1, 2, \perp\}$. Bob receives an input labelled by $y\in[d^2]$ and the associated measurement produces a binary outcome labelled by $b\in\{1, \perp\}$. The joint probability distribution is denoted by $p(a,b|x, y)$, and the Bell scenario is illustrated in Fig.~\ref{FigSICmarginal}.

\begin{figure}
	\centering
	\includegraphics[width=\columnwidth]{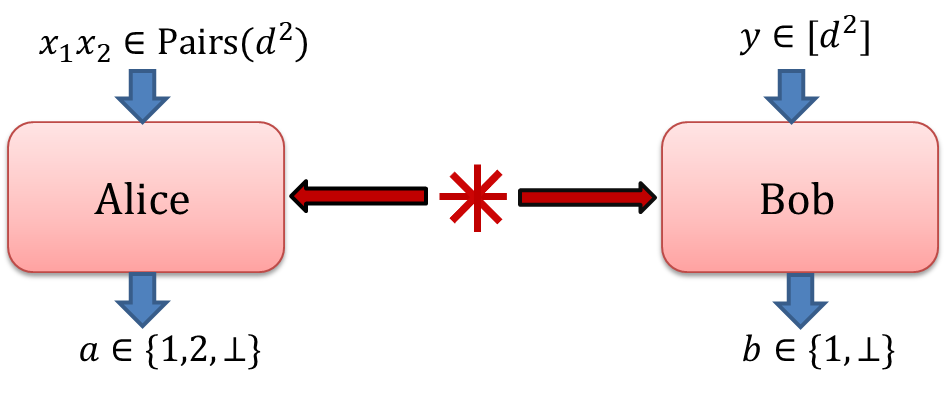}
	\caption{Bell scenario for SICs of dimension $d$. Alice receives one of $\binom{d^2}{2}$ inputs and returns a ternary outcome while Bob receives one of $d^2$ inputs and returns a binary outcome.}\label{FigSICmarginal}
\end{figure}

Similar to the case of MUBs, in order to make our choice of Bell functional transparent, we phrase it as a game played by Alice and Bob. We imagine that their inputs are supplied by a referee, who promises to provide $x = x_{1} x_{2}$ and $y$ such that either $y=x_{1}$ or $y=x_{2}$. Similar to the previous game Alice can output $a = \perp$ to ensure that no points are won or lost. However, in this game also Bob can ensure that no points are won or lost by outputting $b=\perp$. If neither of them outputs $\perp$, a point is either won or lost. Specifically, when $a = 1$ a point is won if $y = x_{1}$ (and lost otherwise), whereas if $a = 2$ then a point is won if $y = x_{2}$ (and lost otherwise). Let us remark that in this game Bob's only role is to decide whether in a given round points can be won/lost or not. For this game the total number of points (the Bell functional) reads
\begin{equation}\label{sicBell}
\begin{aligned}
	\mathcal{R}_{d}^\text{SIC}\equiv &\sum_{x_{1} < x_{2}} \Big(p(1,1|x, x_{1})-p(1,1|x, x_{2})\\
	&+p(2,1|x, x_{2})-p(2,1|x, x_{1})\Big),
\end{aligned}
\end{equation}
where the sum is taken over all $x\in \text{Pairs}(d^2)$.

Let us now impose additional constraints on the marginal distributions of the outputs. More specifically, we require that
\begin{equation}\label{marg}
\begin{aligned}
	& \forall x: && p\left(a=1|x\right)+p\left(a=2|x\right)=\frac{2}{d},\\
	& \forall y: && p(b=1|y)=\frac{1}{d}.
\end{aligned}
\end{equation}
The intuition behind these constraints is analogous to that discussed for the case of MUBs. Namely, we imagine that Alice and Bob perform measurements on a maximally entangled state of local dimension $d$. Then, we wish to fix the marginals such that the measurements of Alice (Bob) for the outcomes $a\in\{1,2\}$ ($b=1$) remotely prepare Bob's (Alice's) subsystem in a pure state. This corresponds to the marginals $p\left(a=1|x\right)=p\left(a=2|x\right)=p\left(b=1|x\right)=1/d$ which is reflected in the marginal constraints in Eq.~\eqref{marg}. We remark that imposing these constraints simplifies both the intuitive understanding of the game and the derivation of the results below. However, it merely serves as a stepping stone to a more general subsequent treatment in which the constraints \eqref{marg} will be removed.

To write the value of the Bell functional of a quantum realisation, let us introduce two simplifications. The measurement operators of Alice are denoted by $\{A_{x}^a\}$ and as before it is convenient to work with the observables defined as $A_{x}=A_{x}^{1} - A_{x}^{2}$. The measurements of Bob are denoted by $\{B_{y}^{b}\}$, but since they only have two outcomes, all the expressions can be written in terms of a single operator from each input $y$. In our case it is convenient to use the outcome-one operator and for convenience we will skip the superscript, i.e.~we will write $B_{y} \equiv B_{y}^{1}$ for all $y$. Then, the Bell functional evaluated on a specific quantum realisation reads
\begin{equation}
	\Bfunc{R}{SIC}=\sum_{x_{1} < x_{2}} \bracket{\psi}{A_{x}\otimes \left(B_{x_{1}}-B_{x_{2}}\right)}{\psi}.
\end{equation}
Note that the Bell functional, in particular when written in a quantum model, is much reminiscent of the expression $\Bfunc{R}{MUB}$ \eqref{Md} encountered for MUBs, with the key difference that the roles of the inputs and outputs of Bob are swapped. 

Let us consider a quantum strategy in which Alice and Bob share  a maximally entangled state $\ket{\maxent{d}}$. Moreover, Bob's measurements are defined as $B_{y} = \ketbra{\phi_y}{\phi_y}$, where
$\{ \ket{\phi_{y}} \}_{y = 1}^{d^{2}}$ is a set of unit vectors forming a SIC (assuming it exists in dimension $d$), i.e.~$|\braket{\phi_y}{\phi_{y'}}|^2=1/(d+1)$ for all $y\neq y'$. Also, we define Alice's observables as $A_{x}=\sqrt{(d+1)/d}\left(B_{x_{1}}-B_{x_{2}}\right)\tran$, where the pre-factor ensures normalisation. Firstly, since the subsystems of Alice and Bob are maximally mixed, and the outcomes $a\in\{1,2\}$ and $b=1$ each correspond to rank-one projectors, the marginal constraints in Eq.~\eqref{marg} are satisfied. Using the fact that for any linear operator $O$ we have $O \otimes \I \ket{\maxent{d}} = \I \otimes O\tran \ket{\maxent{d}}$, we find that
\begin{multline} \label{strat}
	\Bfunc{R}{SIC}=\\
	\sqrt{\frac{d+1}{d}}\sum_{x_{1} < x_{2}} \bracket{\maxent{d}}{\I\otimes \left(\ketbra{\phi_{x_{1}}}{\phi_{x_{1}}}-\ketbra{\phi_{x_{2}}}{\phi_{x_{2}}}\right)^2}{\maxent{d}}\\
	=\sqrt{\frac{d+1}{d}}\sum_{x_{1} < x_{2}}\left( \frac{2}{d}-\frac{2}{d(d+1)}\right)=d(d-1)\sqrt{d(d+1)}.
\end{multline}

In fact, this strategy relying on a maximally entangled state and a SIC achieves the maximal quantum value of $\Bfunc{R}{SIC}$ under the constraints of Eq.~\eqref{marg}. In SM (section III.A) we prove that under these constraints the tight quantum and no-signaling bounds on $\Bfunc{R}{SIC}$ read
\begin{align}\label{sicTsirelson}
	\mathcal{R}_{d}^\text{SIC} &\maxval{Q} d(d-1)\sqrt{d\left(d+1\right)}\\
	\Bfunc{R}{SIC} &\maxval{NS} d\left(d^2-1\right).
\end{align}

We remark that SICs are not known to exist in all Hilbert space dimensions. However, their existence in all dimensions is strongly conjectured and explicit SICs have been found in all dimensions up to $121$  \cite{Scott}. 

\subsection{Bell inequalities for SICs}
The marginal constraints in Eq.~\eqref{marg} allowed us to prove that the quantum realisation based on SICs achieves the maximal quantum value of $\Bfunc{R}{SIC}$. Our goal now is to remove these constraints to obtain a standard Bell functional. Analogously to the case of MUBs we add marginal terms to the original functional $\Bfunc{R}{SIC}$.

To this end, we introduce penalties for both Alice and Bob. Specifically, if Alice outputs $a\in\{1, 2\}$ they lose $\alpha_d$ points, whereas if Bob outputs $b=1$, they lose $\beta_d$ points. The total number of points in the modified game constitutes our final Bell functional 
\begin{multline}\label{sicBellfin}
	\Bfunc{S}{SIC}\equiv \Bfunc{R}{SIC}-\alpha_d \sum_{x_{1} < x_{2}}\left(p\left(a=1|x\right)+p\left(a=2|x\right)\right)\\
	-\beta_d \sum_{y}p\left(b=1|y\right).
\end{multline}
Hence, our aim is to suitably choose the penalties $\alpha_d$ and $\beta_d$ so that the maximal quantum value of $\Bfunc{S}{SIC}$ is achieved with a strategy that closely mimics the marginal constraints \eqref{marg} and thus maintains the optimality of Bob performing a SIC.

\begin{theorem}[Bell inequalities for SICs]\label{sictheorem}
	The Bell functional $ \Bfunc{S}{SIC}$  in Eq.~\eqref{sicBellfin} with 
	\begin{equation}
	\begin{aligned}
		& \alpha_d=\frac{1-\delta_{d,2}}{2}\sqrt{\frac{d}{d+1}}\\\label{coefs}
		& \beta_d=\frac{d-2}{2}\sqrt{d(d+1)},
	\end{aligned}
	\end{equation}
	obeys the tight local bound
	\begin{equation}\label{SIClocal}
	\Bfunc{S}{SIC} \maxval{LHV}
	\begin{cases}
	4 &\mbox{for} \quad d = 2,\\
	d^{2}(d - 1) - d(d^{2} - d - 1)\sqrt{\frac{d}{d+1}} &\mbox{for} \quad d \geq 3,
	\end{cases}
	\end{equation}
	and the quantum bound
	\begin{equation}\label{qboundsic}
		\Bfunc{S}{SIC} \maxval{Q} \frac{d+2\delta_{d,2}}{2}\sqrt{d\left(d+1\right)}.
	\end{equation}
	Moreover, the quantum bound is tight and can be saturated by sharing a maximally entangled state of local dimension $d$ and choosing Bob's outcome-one projectors to form a SIC.
\end{theorem}
\begin{proof}
The proof is presented in SM (section III.B ). In order to obtain the quantum bound in Eq.~\eqref{qboundsic}, the key ingredients are the Cauchy--Schwarz inequality and semidefinite relaxations of polynomial optimisation problems. To derive the local bound in Eq.~\eqref{SIClocal}, the key observation is that the symmetries of the Bell functional allow us to significantly simplify the problem. 

The fact that the quantum bound is saturated by a maximally entangled state and Bob performing a SIC can be seen immediately from the previous discussion that led to Eq.~\eqref{strat}. With that strategy, we find $\Bfunc{R}{SIC}=d(d-1)\sqrt{d(d+1)}$. Since it also respects $p(a=1|x)+p(a=2|x)=2/d$ $\forall x$, as well as $p(b=1|y)=1/d$ $\forall y$, a direct insertion into Eq.~\eqref{sicBellfin} saturates the bound in Eq.~\eqref{qboundsic}.
\end{proof}

Note that in the limit of $d \to \infty$ the local bound  grows linearly in $d$ while the quantum bound grows quadratically in $d$.

We remark that for the special case of $d=2$, no penalties are needed to maintain the optimality of SICs (which is why the delta function appears in Eq.~\eqref{coefs}). The derived Bell inequality for a qubit SIC (which corresponds to a tetrahedron configuration on the Bloch sphere) can be compared to the so-called elegant Bell inequality \cite{ElegantBell} whose maximal violation is also achieved using the tetrahedron configuration. While we require six settings of Alice and four settings of Bob, the elegant Bell inequality requires only four settings of Alice and three settings of Bob. However, the additional complexity in our setup carries an advantage when considering the critical visibility of the shared state; i.e.~the smallest value of $v$ in Eq.~\eqref{isotropic} (defining an isotropic state) for which the Bell inequality is violated. The critical visibility for violating the elegant Bell inequality is $86.6\%$, whereas for our Bell inequality it is lowered to $81.6\%$. We remark that on the Bloch sphere, the anti-podal points corresponding to the four measurements of Alice and the six measurements of Bob form a cube and a cuboctahedron respectively, which constitutes an instance of the type of Bell inequalities proposed in Ref.~\cite{Plato}.

\subsection{Device-independent certification}
Theorem~\ref{sictheorem} shows that for any dimension $d \geq 2$ we can construct a Bell inequality which is maximally violated by a SIC in that dimension (provided a SIC exists). Let us now consider the converse question, namely that of device-independent certification. In analogy with the case of MUBs (Eq.~\eqref{sandwich}), we find a simple description of Bob's measurements.

\begin{theorem}[Device-independent certification]\label{thmDIsic}
The maximal quantum value of the Bell functional $\Bfunc{S}{SIC}$, provided the marginal state of Bob is full-rank, implies that his measurement operators $\{ B_{y} \}_{y = 1}^{d^{2}}$ are projective and satisfy 
\begin{equation}\label{sicsum}
\sum_{y} B_{y} = d \, \I
\end{equation}
and
\begin{equation}\label{sicsandwich}
B_{y}=(d+1) B_y B_{y'} B_y
\end{equation}
for all $y \neq y'$.
\end{theorem}
A complete proof, which is similar in spirit to the proof of Theorem~\ref{thmDIMUB}, can be found in SM (section III.C).
	
For the special case of $d=2$, the conclusion can be made even more accurate: the maximal quantum violation of $\mathcal{S}_2^\text{SIC}$ implies that Bob's outcome-one projectors are rank-one projectors acting on a qubit whose Bloch vectors form a regular tetrahedron (up to the three standard equivalences used in self-testing).

Similar to the case of MUBs, we face the key question of interpreting the condition in Eq.~\eqref{sicsandwich} and its relation to SICs. Again in analogy with the case of MUBs, we note that the concept of a SIC references the dimension of the Hilbert space, which should not appear explicitly in a device-independent scenario. Hence we consider an operational approach to SICs, which must rely on observable quantities (i.e.~probabilities). This leads us to the following natural definition of a set of projectors being \textit{operationally symmetric informationally complete} (OP-SIC). 

\begin{definition}[Operational SIC]
We say that a set of projectors $\{B_a\}_{a=1}^{n^2}$ is operationally symmetric informationally complete (OP-SIC) if 
\begin{equation}
\sum_a B_a=n \, \I
\end{equation} 
and 
\begin{equation}\label{OPSIC}
\bracket{\psi}{B_a}{\psi}=1\Rightarrow \bracket{\psi}{B_b}{\psi}=\frac{1}{n+1},
\end{equation}
for all $a\neq b$.
\end{definition}
This definition trivially encompasses SICs as special instances of OP-SICs. More interestingly, an argument analogous to the proof of Theorem~\ref{thmOP} shows that this definition is in fact equivalent to the relations given in Eqs.~\eqref{sicsum} and \eqref{sicsandwich}.
Hence, in analogy with the case of MUBs, the property of Bob's measurements certified by the maximal violation of our Bell inequality is precisely the notion of OP-SICs.
\subsection{Adding a SIC-POVM}
The Bell inequalities proposed above (Bell functional $\Bfunc{S}{SIC}$) are tailored to sets of rank-one projectors forming a SIC. However, it is also interesting to consider a closely related entity, namely a SIC-POVM, which is obtained simply by normalising these projectors, so that they can be collectively interpreted as arising from a single measurement. That is, a SIC-POVM on $\bC^d$ is a measurement $\{E_a\}_{a=1}^{d^2}$ in which every measurement operator can be written as $E_a=\frac{1}{d}\ketbra{\phi_a}{\phi_a}$, where the set of rank-one projectors $\{\ketbra{\phi_a}{\phi_a}\}_a$ forms a SIC. Due to the simple relation between SICs and SIC-POVMs, we can extend the Bell inequalities for SICs proposed above such that they are optimally implemented with both a SIC (as before) and a SIC-POVM.

\begin{figure}
	\centering
	\includegraphics[width=\columnwidth]{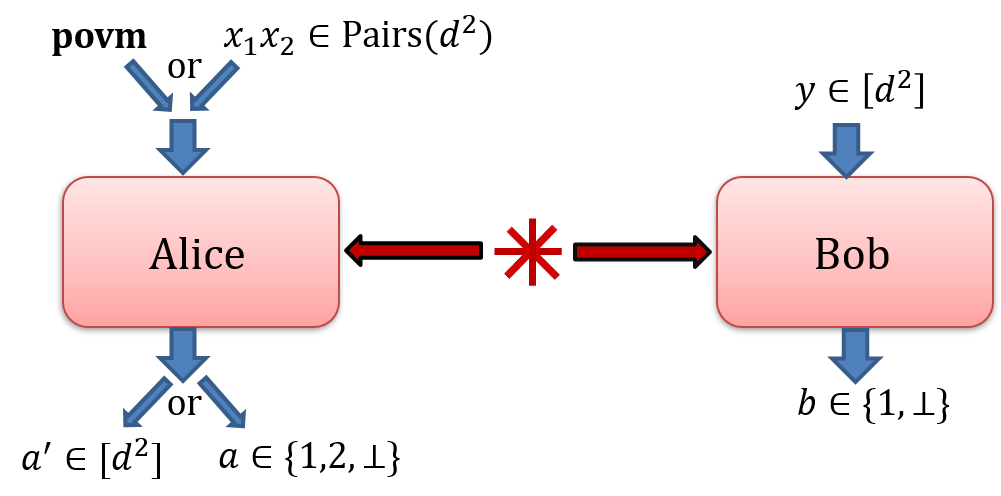}
	\caption{Bell scenario for SICs and SIC-POVMs of dimension $d$. This scenario modifies the original Bell scenario for SICs (see Figure~\ref{FigSICmarginal}) by supplying Alice with an extra setting labelled by $\mathbf{povm}$ which has $d^2$ possible outcomes.}\label{SICPOVMscenario}
\end{figure}

It is clear that in order to make SIC-POVMs relevant to the Bell experiment, it must involve at least one setting which corresponds to a $d^2$-outcome measurement. For the Bell scenario previously considered for SICs (see Figure~\ref{FigSICmarginal}), no such measurement is present. Therefore, we supplement the original Bell scenario by introducing a single additional measurement setting of Alice, labelled by $\mathbf{povm}$, which has $d^2$ outcomes labelled by $a'\in[d^2]$. The modified Bell scenario is illustrated in Figure~\ref{SICPOVMscenario}. We construct the Bell functional $\Bfunc{T}{SIC}$ for this scenario by modifying the previously considered Bell functional  $\Bfunc{S}{SIC}$:
\begin{equation}\label{Sd'}
\mathcal{T}_d^\text{SIC}=\Bfunc{S}{SIC}-\sum_{y=1}^{d^2}p(a'=y,b=\perp|\mathbf{povm},y).
\end{equation}
Hence, whenever Bob outputs ``$\perp$`` and the outcome associated to the setting $\mathbf{povm}$ coincides with the input of Bob, a point is lost.

Evidently, the largest quantum value of $\mathcal{T}_d^\text{SIC}$ is no greater than the largest quantum value of $\Bfunc{S}{SIC}$. In order for the former to equal the latter, we require that: i) $\Bfunc{S}{SIC}$ reaches its maximal quantum value (which is given in Eq.~\eqref{qboundsic}) and ii) that $p(a'=y,b=\perp|\mathbf{povm},y)=0$ $\forall y$. We have already seen that by sharing a maximally entangled state and Bob's outcome-one projectors $\{B_y\}_y$ forming a SIC, the condition i) can be satisfied. By normalisation, we have that Bob's outcome-$\perp$ projectors are $B_y^{\perp}=\I-B_y$.
Again noting that for any linear operator $O$ we have $O \otimes \I \ket{\maxent{d}} = \I \otimes O\tran \ket{\maxent{d}}$, observe that if Bob applies $B_y^{\perp}$, then Alice's local state is orthogonal to $B_y$.  Hence, if Alice chooses her POVM $\{E_{a'}\}$, corresponding to the setting $\mathbf{povm}$, as the SIC-POVM defined by $E_{a'}=\frac{1}{d}B_{a'}^\text{T}$, the probability of finding $a'=y$ vanishes. This satisfies condition ii). Hence, we conclude that in a general quantum model
\begin{equation}
\mathcal{T}_d^\text{SIC} \maxval{Q} \frac{d+2\delta_{d,2}}{2}\sqrt{d\left(d+1\right)},
\end{equation}
and that the bound can be saturated by supplementing the previous optimal realisation with a SIC-POVM on Alice's side.

\section{Application: Device-independent quantum random number generation}
The fact that the Bell functionals $\Bfunc{S}{SIC}$ and $\mathcal{T}_d^\text{SIC}$ achieve their maximal quantum values with a SIC and a SIC-POVM respectively, opens up the possibility for device-independent quantum information protocols for tasks in which SICs and SIC-POVMs are desirable. We focus on one such application, namely that of device-independent quantum random number generation \cite{Pironio2}. This is the task of certifying that the data generated by a party cannot be predicted by a malicious eavesdropper. In the device-independent setting, both the amount of randomness and its security is derived from the violation of a Bell inequality.

Non-projective measurements, such as SIC-POVMs, are useful for this task. The reason is that a Bell experiment implemented with entangled systems of local dimension $d$ and standard projective measurements cannot have more than $d$ outcomes. Consequently, one cannot hope to certify more than $\log d$ bits of local randomness. However, Bell experiment relying on $d$-dimensional entanglement implemented with (extremal) non-projective measurements can have up to $d^2$ outcomes \cite{Ariano}. This opens  the possibility of generating up to $2\log d$ bits of local randomness without increasing the dimension of the shared entangled state. Notably, for the case of $d=2$, such optimal quantum random number generation has been shown using a qubit SIC-POVM \cite{Acin}.

Here, we employ our Bell inequalities for SIC-POVMs to significantly outperform standard protocols relying on projective measurements on $d$-dimensional entangled states. To this end, we briefly summarise the scenario for randomness generation. Alice and Bob perform many rounds of the Bell experiment illustrated in Figure~\ref{SICPOVMscenario}. Alice will attempt to generate local randomness from the outcomes of her setting labelled by $\mathbf{povm}$. In most rounds of the Bell experiment, Alice performs $\mathbf{povm}$ and records the outcome $a'$. In a smaller number of rounds, she randomly chooses her measurement setting and the data is used towards estimating the value of the Bell functional $\mathcal{T}_d^\text{SIC}$ defined in Eq.~\eqref{Sd'}. A malicious eavesdropper may attempt to guess Alice's relevant outcome $a'$. To this end, the eavesdropper may entangle her system with that of Alice and Bob, and perform a well-chosen POVM $\{E_c\}_c$ to enhance her guess. In analogy to Eq.~\eqref{guess}, the eavesdropper's guessing probability reads
\begin{equation}\label{guessing}
P_g^\beta \equiv \sup \bigg\{ \sum_{c=1}^{d^2} \bracket{\psi_\text{ABE}}{A_{\mathbf{povm}}^c \otimes \I \otimes E_c}{\psi_\text{ABE}} \bigg\},
\end{equation}
where $\{ E_{c} \}_{c = 1}^{d^{2}}$ is the measurement employed by the eavesdropper to produce her guess, the expression inside the curly braces is the probability that her outcome is the same as Alice's outcome for the setting $\mathbf{povm}$ for a particular realisation and the supremum is taken over all quantum realisations (the tripartite state and measurements of all three parties) compatible with the observed Bell inequality violation $\beta=\mathcal{T}_d^\text{SIC}$.

We quantify the randomness generated by Alice using the conditional min-entropy $H_\text{min}(A_\mathbf{povm}|E)=-\log\left(P_g^\beta\right)$. To obtain a device-independent lower bound on the randomness, we must evaluate an upper bound on $P_g^\beta$ for a given observed value of the Bell functional. We saw in Section~\ref{sec:QKD} that if the eavesdropper is only trying to guess the outcome of a single measurement setting, we can without loss of generality assume that they are only classically correlated with the systems of Alice and Bob.
As before, we restrict ourselves to the asymptotic limit of many rounds, in which fluctuations due to finite statistics can be neglected.

\begin{figure}
	\centering
	\includegraphics[width=\columnwidth]{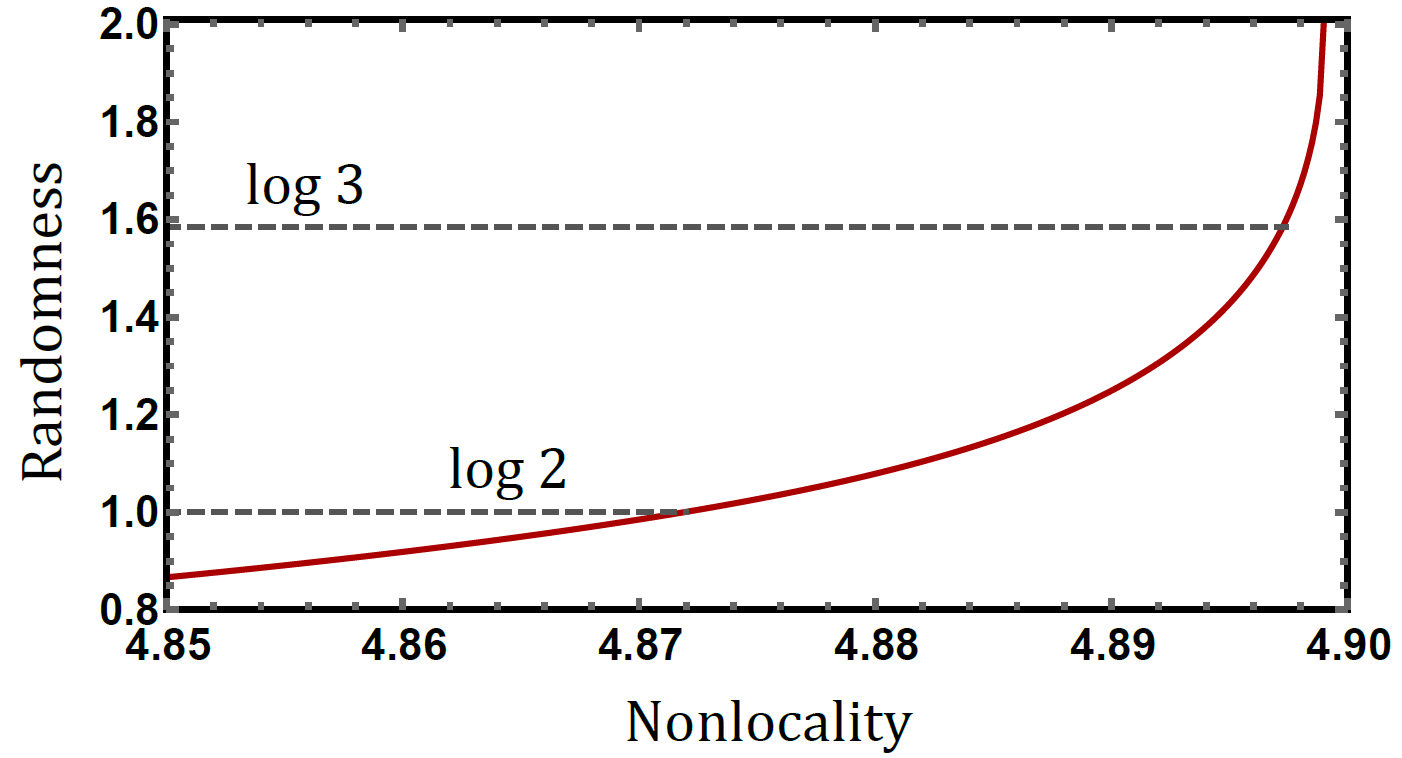}
	\caption{Lower bound on the amount of device-independent randomness versus the value of $\mathcal{T}_2^\text{SIC}$.}\label{FigRand}
\end{figure}

In order to bound the randomness for some given value of $\mathcal{T}_d^\text{SIC}$, we use the hierarchy of quantum correlations~\cite{NPA}. We restrict ourselves to the cases of $d=2$ and $d=3$. For the case of $d=2$, we construct a moment matrix with the operators $\left\{\left(\I, A_x\right)\otimes\left(\I, B_y\right)\otimes\left(\I, E\right)\right\} \cup \left\{A_\mathbf{povm}\otimes\left(\I, B_y,E\right)\right\}$, neglecting the $\perp$ outcome. The matrix is of size $361\times 361$ with $10116$ variables. Again, we can make use of symmetry to simplify the semidefinite program. In this case, the following permutation leaves the problem invariant: $x_1 \to \pi(x_1)$, $x_2 \to \pi(x_2)$, $a \to f_\pi(a,x_1,x_2)$, $a'\to\pi(a')$, $y\to\pi(y)$, $c\to\pi(c)$, where
\begin{equation}
f_\pi(a,x_1,x_2)=
\begin{cases}
a & \pi(x_1)<\pi(x_2)\\
2 & \pi(x_1)\geq\pi(x_2) \text{ and } a=1\\
1 & \pi(x_1)\geq\pi(x_2) \text{ and } a=2\\
\perp & \pi(x_1)\geq\pi(x_2) \text{ and } a=\perp\\
\end{cases}
\end{equation}
and $\pi\in S_4$. Using this symmetry reduces the number of free variables to $477$. The trade-off between the amount of certified randomness and the nonlocality is illustrated in Figure~\ref{FigRand}. We find that for sufficiently large values of $T_2^\text{SIC}$ (roughly $\mathcal{T}_2^\text{SIC}\geq 4.8718$), we outperform the one-bit limitation associated to projective measurements on entangled qubits. Notably, for even larger values of $T_2^\text{SIC}$, we also outperform the restriction of $\log 3$ bits associated to projective measurements on entangled systems of local dimension three. For the optimal value of $T_2^\text{SIC}$ we find $H_\text{min}(A_\mathbf{povm}|E)\gtrsim 1.999$, which is compatible up to numerical precision with the largest possible amount of randomness obtainable from qubit systems under general measurements, namely  two bits\footnote{It is impossible to obtain more than two bits of device-independent local randomness from qubits in a Bell scenario since every qubit measurement with more than four outcomes can be stochastically simulated with measurements of at most four outcomes \cite{Ariano}. Naturally, a four-outcome measurement cannot generate more than two bits of randomness.}.

For the case of $d=3$ we bound the guessing probability following the method of Ref~\cite{Pironio2}. This has the advantage of requiring only a bipartite, and hence smaller, moment matrix than the tripartite formulation. However, the amount of symmetry leaving the problem invariant is reduced, because the objective function only involves one outcome. Concretely, we construct a moment matrix of size $820\times 820$ with $263549$ variables. We then write the guessing probability as $P(a'=1|\mathbf{povm})$ and identify the following group of permutations leaving the problem invariant: $x_1 \to \pi(x_1)$, $x_2 \to \pi(x_2)$, $a \to f_\pi(a,x_1,x_2)$, $a'\to\pi(a')$, $y\to\pi(y)$, where $\pi\in S_9$ leaves element 1 invariant and permutes elements $2,\ldots,9$ in all possible ways. Taking this symmetry into account reduces the number of free variables to $460$. In order to further simplify the problem we make use of RepLAB, a recently developed tool which decomposes representations of finite groups into irreducible representations~\cite{replabWebsite, replab}. This allows us to write the moment matrix in a preferred basis in which it is block diagonal. The semidefinite constraint can then be imposed on each block independently, with the largest block of size $28\times28$ instead of $820\times 820$. Solving one semidefinite program with SeDuMi~\cite{sedumi} then takes $0.7$s with $<0.1$GB of memory instead of $162$s/$0.2$GB without block-diagonalisation, and fails due to lack of memory without any symmetrisation ($>400$GB required).

Using entangled states of dimension three and corresponding SIC-POVMs, one can attain the full range of values for $\mathcal{T}_3^\text{SIC}$. Importantly, the guessing probability is independent of the outcome guessed by the eavesdropper, and we can verify that the bound we obtain is convex, hence guaranteeing that no mixture of strategy by the eavesdropper must be considered~\cite{Pironio2}. The randomness is then given in Figure~\ref{FigRand3}, which indicates that by increasing the value of  $\mathcal{T}_3^\text{SIC}$, we can obtain more randomness than the best possible schemes relying on standard projective measurements and entangled systems of dimensions $3,4,5,6,7$. Especially, in the case of $\mathcal{T}_3^\text{SIC}$ being maximal, we find that $H_\text{min}(A_\mathbf{povm}|E)\approx 3.03$ bits. This is larger than what can be obtained by performing projective measurements on eight dimensional systems (since $\log 8=3$ bits). It is, however, worth noting that this last value is obtained at the boundary of the set of quantum correlations where the precision of the solver is significantly reduced\footnote{In particular, the DIMACS errors at this point are of the order of $10^{-4}$.}. It is not straightforward to estimate the extent to which this reduced precision may influence the guessing probability, so it would be interesting to reproduce this computation with a more precise solver such as SDPA~\cite{sdpa}.

\begin{figure}[t!]
	\centering
	\includegraphics[width=\columnwidth]{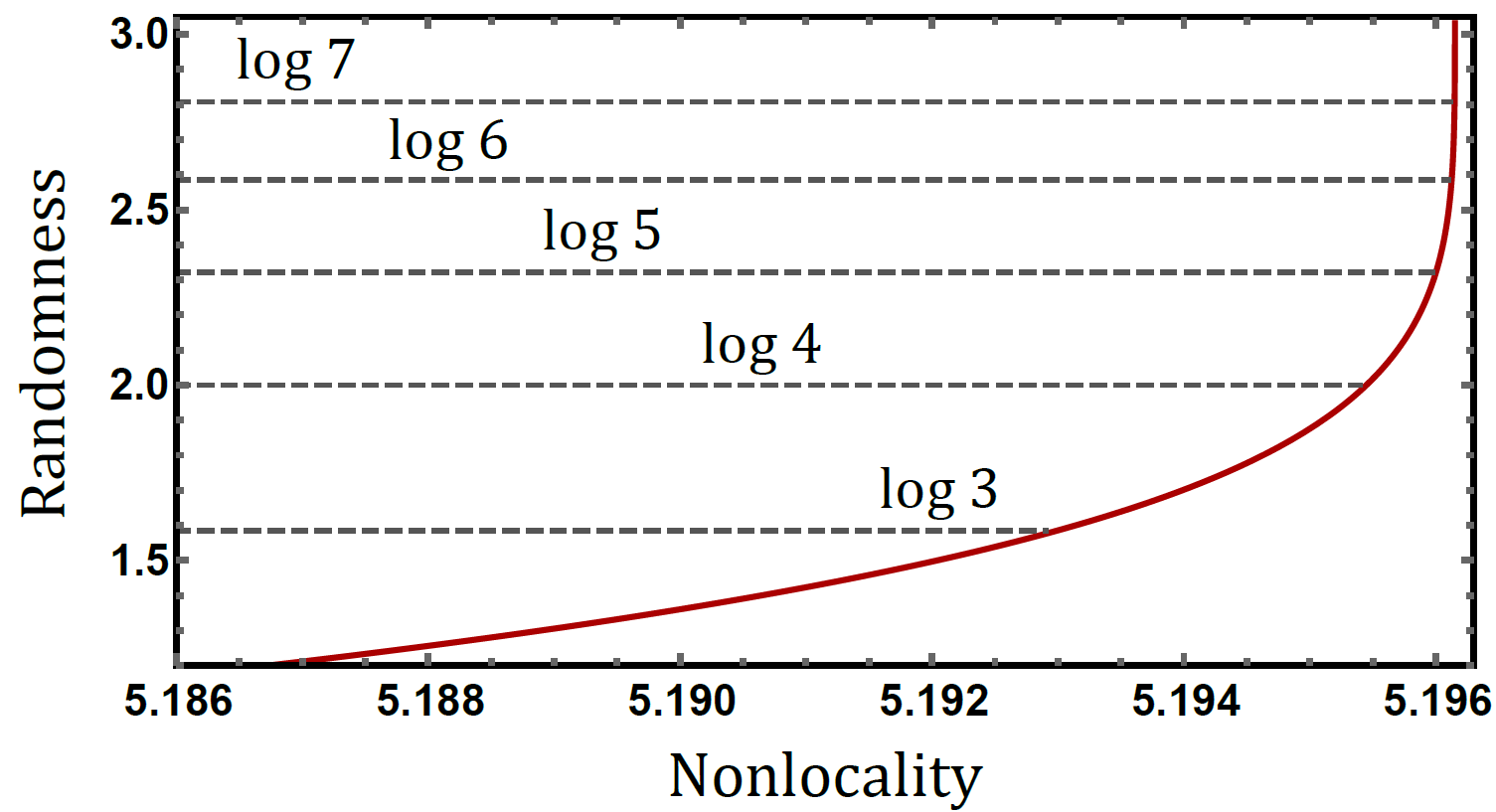}
	\caption{Lower bound on the amount of device-independent randomness versus the value of $\mathcal{T}_3^\text{SIC}$.}\label{FigRand3}
\end{figure}

%

\section{Conclusions}
MUBs and SICs are conceptually elegant, fundamentally important and practically useful features of quantum theory. We investigated their role in quantum nonlocality. For both MUBs and SICs (of any Hilbert space dimension) we presented families of Bell inequalities for which they produce the maximal quantum violations. Moreover, we showed that these maximal quantum violations certify natural operational notions of mutual unbiasedness and symmetric informational completeness. Then, we considered applications of both families of Bell inequalities in  practically relevant tasks. The Bell inequalities for MUBs turn out to be useful for the task of device-independent quantum key distribution and give the optimal key rate for measurements with $d$ outcomes. Moreover, for the case of qutrit systems we investigated the noise robustness of the protocol. For the Bell inequalities for SICs, we considered device-independent random number generation for qubits and qutrits based on SIC-POVMs. We showed (up to numerical precision) optimal randomness generation for qubit systems. For qutrit systems, we showed that more randomness can be generated than in any scheme using standard projective measurements and entanglement of up to dimension seven. These results were obtained using the RepLAB package, which helped to significantly reduce the complexity of the corresponding semidefinite programs by taking advantage of their symmetry.

This work opens many new research directions, so let us mention just a few of them. We showed that a maximal quantum violation of the Bell inequality for MUBs self-tests a maximally entangled state of local dimension $d$. In the case of the Bell inequality for SICs we have managed to certify the measurements of Bob, but we do not have a self-testing result for the state. If a self-test of the state is possible, what are the implications for the device-independent certification of the SIC-POVM setting? This may prove helpful towards solving another interesting question, namely that of proving optimal local randomness generation (i.e.~$2\log d$ bits) for any $d$ based on the Bell inequality for SIC-POVMs. Another avenue of exploration regards the concept of mutually unbiased measurements (MUMs). In this work, we have shown some of their basic properties with regard to MUBs as well as examples of how they are relevant in quantum information theory. However, a more systematic exploration of MUMs would be desirable. Similarly, a general exploration of operational SICs (OP-SICs) in quantum information theory, as well as their relation to SICs, would be of similar interest. Finally, we note that our noise-robust results for quantum key distribution and quantum random number generation may be relevant for experimental implementations.

\begin{acknowledgements}
We would like to thank Thais de Lima Silva and Nicolas Gisin for fruitful discussions. This work was supported by the Swiss National Science Foundation (Starting grant DIAQ, NCCR-QSIT). The project ``Robust certification of quantum devices'' is carried out within the HOMING programme of the Foundation for Polish Science co-financed by the European Union under the European Regional Development Fund. MF acknowledges support from the Polish NCN grant Sonata~UMO-2014/14/E/ST2/00020.
\end{acknowledgements}

\appendix
\onecolumngrid

\section{Bell inequalities for mutually unbiased bases}

In this appendix we fill in some details on the Bell inequalities for MUBs. We start by deriving the local and quantum bounds and proving the device-independent certification result stated in the main text. Then, we proceed to show that no stronger characterisation of Bob's measurement can be obtained from the maximal violation of our Bell inequality and, moreover, that the maximal violation is achieved by a single probability point.
\subsection{Proof of the local and quantum values (Theorem~\ref{thmBellForMUBs})}
\label{AppTheorem1}

We begin by proving the quantum bound of the Bell functional for MUBs ($\mathcal{S}_d^\text{MUB}$). We denote Alice's POVMs by $\{A_{x}^a\}$, Bob's POVMs by $\{P_b\}$ (for $y=1$)  and $\{Q_b\}$ (for $y=2$) and the shared state by $\ket{\psi}$ (without loss of generality we assume that the measurements are projective and the shared state is pure). For the moment, we focus on the first term contributing the the Bell functional, namely $\Bfunc{R}{MUB}$ defined in Eq.~\eqref{Md}. In a quantum model, we have
\begin{equation}\label{qMUB}
	\Bfunc{R}{MUB} =\sum_{x} \bracket{\psi}{A_x\otimes \left(P_{x_1}-Q_{x_2}\right)}{\psi},
\end{equation}
where the summation goes over $x = x_{1} x_{2} \in [d]^{2}$ and the observable $A_{x}$ is defined as $A_x\equiv A_x^1 - A_x^2$. In what follows, we will omit the tensor notation when Alice's or Bob's action is the identity. 

Applying the Cauchy--Schwarz inequality to every term in the summand of Eq.~\eqref{qMUB} gives
\begin{equation}\label{MUBcont}
	\Bfunc{R}{MUB}\leq \sum_x \sqrt{ \bracket{\psi}{A_x^1+A_x^2}{\psi}}\sqrt{\bracket{\psi}{\left(P_{x_1}-Q_{x_2}\right)^2}{\psi}},
\end{equation}
where we have used the fact that for projective measurements $(A_{x})^{2} = A_{x}^{1} + A_{x}^{2}$. In the next step we again use the Cauchy--Schwarz inequality but in a different form: for $s_i,r_i\geq 0$ it holds that $\sum_i \sqrt{s_i}\sqrt{r_i}\leq \sqrt{\sum_i s_i}\sqrt{\sum_i r_i}$ with equality if and only if $\forall i: s_i=k\cdot r_i$ for some proportionality constant $k$. This leads us to the bound	
\begin{equation}\label{holder}
	\Bfunc{R}{MUB}\leq \sqrt{\sum_x \bracket{\psi}{A_x^1+A_x^2}{\psi} }\sqrt{\sum_x \bracket{\psi}{\left(P_{x_1}-Q_{x_2}\right)^2}{\psi}}.
\end{equation}
We denote the sum under the first square-root by $t$. Also, we use projectivity and completeness to write
\begin{equation}
	\label{eq:sum-diff}
	\sum_x \left(P_{x_1}-Q_{x_2}\right)^2=\sum_x P_{x_1}+Q_{x_2}-\{P_{x_1},Q_{x_2}\}=2\left(d-1\right)\I.
\end{equation}
Then, we can return to the Bell functional $\Bfunc{S}{MUB}$ and bound it by
\begin{equation}\label{qBellfin}
	\Bfunc{S}{MUB}\leq \sqrt{2(d-1)t}-\frac{1}{2}\sqrt{\frac{d-1}{d}}t.
\end{equation}
To maximise the right-hand-side over $t\geq 0$ we differentiate with respect to $t$ to find a unique extremum at $t=2d$. Re-inserting this value into Eq.~\eqref{qBellfin} returns the quantum bound given in Eq.~\eqref{QMUB}.

To derive the classical bound on $\Bfunc{S}{MUB}$ given in Eq.~\eqref{CMUB} recall that it suffices to optimise over deterministic strategies. Moreover, once the strategy of Bob is fixed, finding the optimal strategy of Alice is easy. If Bob outputs $b = u_{1}$ for $y = 1$ and $b = u_{2}$ for $y = 2$, the Bell functional becomes
\begin{equation}
	\Bfunc{S}{MUB}=\sum_{x} \left(\delta_{x_1, u_{1}}-\delta_{x_2, u_{2}}\right)\bracket{\psi}{A_x^1-A_x^2}{\psi}-\frac{1}{2}\sqrt{\frac{d-1}{d}}\sum_{x}\bracket{\psi}{A_x^1+A_x^2}{\psi}.
\end{equation}
We define $R_\pm=\{x \in [d]^{2} |\delta_{x_1, u_{1}}-\delta_{x_2, u_{2}}=\pm 1\}$ and $R_0=[d]^2\setminus (R_+ \cup R_-)$. By expanding the above expression for $\Bfunc{S}{MUB}$ into the separate sums over $R_+$, $R_-$ and $R_0$ it becomes clear that the optimal choice of Alice is to choose $A_x^1=\I $ and $A_x^2=A_x^\perp=0$ (always output $a=1$) when $x\in R_+$, choose $A_x^2=\I $ and $A_x^1=A_x^\perp=0$ (always output $a=2$) when $x\in R_-$ and choose $A_x^\perp=\I $ and $A_x^1=A_x^2=0$ (always output $a=\perp$) when $x\in R_0$. Since $|R_\pm|=d-1$, this leads to the classical bound in Eq.~\eqref{CMUB}.

\subsection{Device-independent certification (Theorem~\ref{thmDIMUB})}\label{AppTheorem3}

In this section we show what can be deduced about the measurements of Bob and the shared state based only on observing the maximal violation. In the last part we show that the condition obtained for Bob's measurement is complete, i.e.~it cannot be strengthened in any way.

To simplify the notation we assume that the marginal states of Alice and Bob are full-rank (in any case no information can be deduced outside of the support of the state). Moreover, since we do not aim to certify the measurements of Alice we can without loss of generality assume that they are projective. For Bob, on the other hand, we do not make such an assumption and projectivity is rigorously deduced.

In the arguments below we assume that the state shared between Alice and Bob is pure. However, if Alice and Bob share a mixed state we simply purify it using an additional register and apply the arguments below to the purification. Since the purification register remains untouched throughout the argument, it is clear that the same conclusions hold if Alice and Bob share a mixed state.
\subsubsection{Measurements of Bob}
\label{app:meas-Bob}
In the previous section we have derived an upper bound on the quantum value of the Bell functional $\Bfunc{R}{MUB}$. Since the resulting bound is tight, our argument must be tight at every step. By examining each step of the argument we can deduce certain relations that must be satisfied by any quantum realisation that produces the maximal violation.

Eq.~\eqref{MUBcont} is obtained by applying the Cauchy--Schwarz inequality to the Bell functional $\Bfunc{R}{MUB}$ in Eq.~\eqref{qMUB}. Saturating the Cauchy--Schwarz inequality implies that for all $x$ we have
\begin{equation}\label{rcross}
	A_x\ket{\psi} = \mu_{x} \left(P_{x_1}-Q_{x_2}\right)\ket{\psi},
\end{equation}
for some complex number $\mu_{x}$. Left-multiplying both sides by $\bra{\psi}(P_{x_1}-Q_{x_2})$, one deduces that $\mu_{x}$ is real and non-negative. The non-negativity stems from the fact that we have previously used that $\bracket{\psi}{A_x\otimes \left(P_{x_1}-Q_{x_2}\right)}{\psi} \leq \lvert \bracket{\psi}{A_x\otimes \left(P_{x_1}-Q_{x_2}\right)}{\psi} \rvert$. Also, Eq.~\eqref{rcross} implies that
\begin{equation}\label{rrcors}
	\bracket{\psi}{(A_x)^2}{\psi} = \mu_{x}^{2} \bracket{\psi}{\left(P_{x_1}-Q_{x_2}\right)^2}{\psi}.
\end{equation}
The fact that the second application of the Cauchy--Schwarz inequality, which leads to Eq.~\eqref{holder}, is tight implies that $\bracket{\psi}{(A_x)^2}{\psi} = \bracket{\psi}{A_x^1+A_x^2}{\psi} = \nu \bracket{\psi}{\left(P_{x_1}-Q_{x_2}\right)^2}{\psi}$ for some constant $\nu > 0$. This implies that $\mu_{1} = \mu_{2} = \ldots = \mu_{d^2} \equiv \mu$. Finally, note that Eq.~\eqref{eq:sum-diff} holds if and only if all the measurements of Bob are projective and since we used it to derive our upper bound, we conclude that all the measurements of Bob must be projective.

Summing Eq.~\eqref{rrcors} over $x$ we obtain
\begin{equation}
	t=2 \mu^{2} (d-1),
\end{equation}
where $t\equiv \sum_x \bracket{\psi}{(A_x)^2}{\psi}$. Just after Eq.~\eqref{qBellfin} we found that the optimal value of $t$ is given by $t=2d$, which implies
\begin{equation}\label{r}
	\mu = \sqrt{\frac{d}{d-1}}.
\end{equation}
Thus, we have established the useful relation
\begin{equation}\label{cross}
	A_x\ket{\psi}=\sqrt{\frac{d}{d-1}}\left(P_{x_1}-Q_{x_2}\right)\ket{\psi}.
\end{equation}

Since the measurements of Alice are projective, the spectra of $A_x$ only contain the values $\{+1,-1,0\}$ and, hence, the observables must satisfy $(A_x)^3 = A_x$. Using this in Eq.~\eqref{cross} gives
\begin{equation}
	\sqrt{\frac{d}{d-1}}\left(P_{x_1}-Q_{x_2}\right)\ket{\psi}=\left(\frac{d}{d-1}\right)^{3/2}\left(P_{x_1}-Q_{x_2}\right)^3\ket{\psi}.
\end{equation}
Note that we have managed to eliminate the operators of Alice from the equation. Therefore, we can trace out Alice's system and subsequently right-multiply by the inverse of Bob's local state (which we assume to be full-rank). This leaves us with
\begin{equation}
	P_{x_1}-Q_{x_2}=\frac{d}{d-1}\left(P_{x_1}-Q_{x_2}\right)^3,
\end{equation}
which using projectivity can be simplified to
\begin{equation}\label{ttt}
	P_{x_1}-Q_{x_2}=d\left(P_{x_1}Q_{x_2}P_{x_1}-Q_{x_2}P_{x_1}Q_{x_2}\right).
\end{equation}
Summing over $x_1$ we obtain
\begin{equation}
	\I=d\sum_{x_1}P_{x_1}Q_{x_2}P_{x_1}.
\end{equation}
Since $\{P_{x_1}\}$ are orthogonal, this implies that
\begin{equation}\label{eq:sand1}
	P_{x_1} = d P_{x_1} Q_{x_2} P_{x_1},
\end{equation}
which is the desired relation. Analogously, summing Eq.~\eqref{ttt} over $x_2$ leads to
\begin{equation}\label{eq:sand2}
	Q_{x_2} = d Q_{x_2} P_{x_1} Q_{x_2}.
\end{equation}
\subsubsection{The shared state}
Here we show that if Alice and Bob observe the maximal quantum value of $\Bfunc{S}{MUB}$, then there exist local operations which allow them to extract a maximally entangled state of local dimension $d$. More specifically, if $\ket{\psi} \in \cH_{A} \otimes \cH_{B}$ is the state shared between Alice and Bob, then we explicitly construct isometries $V_{A} : \cH_{A} \to \bC^{d} \otimes \cH_{A}$ and $V_{B} : \cH_{B} \to \bC^{d} \otimes \cH_{B}$ such that
\begin{equation}
	( V_{A} \otimes V_{B} ) \ket{\psi} = \ket{\maxent{d}} \otimes \ket{\psi_{\textnormal{aux}}}.
\end{equation}
As usual these isometries are constructed out of the measurement operators: $\{ A_{x} \}_{x \in [d]^{2}}$ for Alice and $\{ P_{j} \}_{j = 1}^{d}$ and $\{ Q_{k} \}_{k = 1}^{d}$ for Bob.

Let us start with the construction on Bob's side. Bob starts by applying an isometry $R : \cH_{B} \to \bC^{d} \otimes \cH_{B}$ defined as
\begin{equation}
	R = \sum_{j} \ket{j} \otimes P_{j}.
\end{equation}
Note that all the summations in this section go over $[d]$. Then, he applies a unitary $S : \bC^{d} \otimes \cH_{B} \to \bC^{d} \otimes \cH_{B}$ defined as
\begin{equation}
	S = \sum_{k} \ketbraq{k} \otimes U_{k},
\end{equation}
where
\begin{equation}
	\label{eq:Uk}
	U_{k} = d \sum_{j} P_{j} Q_{1} P_{j + k}.
\end{equation}
Note that the integer $j + k$ in the subscript is taken modulo $d$ with the possible values in $\{1, 2, \ldots, d\}$. The fact that the operators $\{ U_{k} \}_{k = 1}^{d}$ are indeed unitaries is a direct consequence of Eqs.~\eqref{eq:sand1} and~\eqref{eq:sand2}. In fact, they correspond to a cyclic shift of the measurement operators $\{ P_{j} \}$:
\begin{equation}
	\label{eq:Uk-shift}
	U_{k} P_{j} U_{k}\hc = P_{j - k}.
\end{equation}
The combined extraction procedure on Bob's side reads:
\begin{equation}
	\label{eq:VB}
	V_{B} = SR = \sum_{j} \ket{j} \otimes U_{j} P_{j} = d \sum_{j} \ket{j} \otimes P_{d} Q_{1} P_{j}.
\end{equation}
To construct an isometry on Alice's side we first construct operators which on act on Alice's side analogous to how $P_{j}$ and $U_{j}$ act on Bob's side. Summing Eq.~\eqref{cross} over one of the indices implies that
\begin{align}
	P_{x_{1}} \ket{\psi} &= \frac{1}{d} \big( \I + \mu^{-1} \sum_{x_{2}} A_{x} \big) \ket{\psi},\\
	Q_{x_{2}} \ket{\psi} &= \frac{1}{d} \big( \I - \mu^{-1} \sum_{x_{1}} A_{x} \big) \ket{\psi}
\end{align}
and recall that $r = \sqrt{d/(d-1)}$. This motivates us to define the following operators on Alice's side:
\begin{align}
	\widetilde{P}_{x_{1}} &= \frac{1}{d} \big( \I + \mu^{-1} \sum_{x_{2}} A_{x} \big),\\
	\widetilde{Q}_{x_{2}} &= \frac{1}{d} \big( \I - \mu^{-1} \sum_{x_{1}} A_{x} \big).
\end{align}
The cross relations
\begin{align}
	\label{eq:relation-1}
	( \widetilde{P}_{x_{1}} \otimes \I ) \ket{\psi} = ( \I \otimes P_{x_{1}} ) \ket{\psi},\\
	( \widetilde{Q}_{x_{2}} \otimes \I ) \ket{\psi} = ( \I \otimes Q_{x_{2}} ) \ket{\psi}.
\end{align}
ensure that the new operators of Alice satisfy the same algebraic relations as the old operators of Bob. For instance to show projectivity note that
\begin{equation}
	( \widetilde{P}_{x_{1}} \otimes \I ) \ket{\psi} = ( \I \otimes P_{x_{1}} ) \ket{\psi} = ( \I \otimes P_{x_{1}}^{2} ) \ket{\psi} = ( \widetilde{P}_{x_{1}}^{2} \otimes \I ) \ket{\psi}.
\end{equation}
Tracing out the system of Bob gives
\begin{equation}
	\widetilde{P}_{x_{1}} \rho_{A} = \widetilde{P}_{x_{1}}^{2} \rho_{A}.
\end{equation}
Right-multiplying by $\rho_{A}^{-1}$ (recall the full-rank marginal assumption) leads to
\begin{equation}
	\widetilde{P}_{x_{1}} = \widetilde{P}_{x_{1}}^{2}.
\end{equation}
Summing over $x_{1}$ in Eq.~\eqref{eq:relation-1} gives
\begin{equation}
	( \sum_{x_{1}} \widetilde{P}_{x_{1}} \otimes \I ) \ket{\psi} = \ket{\psi}.
\end{equation}
Tracing over Bob's system and right-multiplying by $\rho_{A}^{-1}$ implies
\begin{equation}
	\sum_{x_{1}} \widetilde{P}_{x_{1}} = \I.
\end{equation}
To see that they satisfy relations analogous to Eqs.~\eqref{eq:sand1} note that
\begin{align}
	( \widetilde{P}_{x_{1}} \otimes \I ) \ket{\psi} &= ( \I \otimes P_{x_{1}} ) \ket{\psi} = d ( \I \otimes P_{x_{1}} Q_{x_{2}} P_{x_{1}} ) \ket{\psi} = d ( \widetilde{P}_{x_{1}} \otimes P_{x_{1}} Q_{x_{2}} ) \ket{\psi}\\
	&= d ( \widetilde{P}_{x_{1}} \widetilde{Q}_{x_{2}} \otimes P_{x_{1}} ) \ket{\psi} = d ( \widetilde{P}_{x_{1}} \widetilde{Q}_{x_{2}} \widetilde{P}_{x_{1}} \otimes \I ) \ket{\psi}
\end{align}
and therefore
\begin{equation}
	\widetilde{P}_{x_{1}} = d \widetilde{P}_{x_{1}} \widetilde{Q}_{x_{2}} \widetilde{P}_{x_{1}}.
\end{equation}
By symmetry they also satisfy
\begin{equation}
	\widetilde{Q}_{x_{2}} = d \widetilde{Q}_{x_{2}} \widetilde{P}_{x_{1}} \widetilde{Q}_{x_{2}}.
\end{equation}
This implies that
\begin{equation}
	\widetilde{U}_{k} = d \sum_{j} \widetilde{P}_{j} \widetilde{Q}_{1} \widetilde{P}_{j + k}
\end{equation}
are valid unitaries on $\cH_{A}$. Moreover, it is easy to check that
\begin{equation}
	\label{eq:U-move}
	( \widetilde{U}_{k} \otimes \I ) \ket{\psi} = ( \I \otimes U_{k}\hc ) \ket{\psi}.
\end{equation}
Now we are ready to define the local isometry on Alice's side. Again, it consists of two parts
\begin{align}
	\widetilde{R} &= \sum_{j} \ket{j} \otimes \widetilde{P}_{j},\\
	\widetilde{S} &= \sum_{k} \ketbraq{k} \otimes \widetilde{U}_{k}
\end{align}
and the combined extraction isometry on Alice's side reads
\begin{equation}
	V_{A} = \widetilde{S} \widetilde{R} = \sum_{j} \ket{j} \otimes \widetilde{U}_{j} \widetilde{P}_{j}.
\end{equation}
Applying the local isometries to the initial state gives
\begin{equation}
	\ket{\psi_{\textnormal{out}}} = ( V_{A} \otimes V_{B} ) \ket{\psi} = \sum_{jk} \ket{jk} \otimes ( \widetilde{U}_{j} \widetilde{P}_{j} \otimes U_{k} P_{k} ) \ket{\psi}.
\end{equation}
Since
\begin{equation}
	( \widetilde{P}_{j} \otimes P_{k} ) \ket{\psi} = ( \I \otimes P_{k} P_{j} ) \ket{\psi} = \delta_{jk} ( \I \otimes P_{j} ) \ket{\psi},
\end{equation}
the cross terms necessarily vanish:
\begin{equation}
	\ket{\psi_{\textnormal{out}}} = \sum_{j} \ket{jj} \otimes ( \widetilde{U}_{j} \otimes U_{j} P_{j} ) \ket{\psi}.
\end{equation}
Moreover, Eqs.~\eqref{eq:U-move} and \eqref{eq:Uk-shift} imply that
\begin{equation}
	( \widetilde{U_{j}} \otimes U_{j} P_{j} ) \ket{\psi} = ( \I \otimes U_{j} P_{j} U_{j}\hc ) \ket{\psi} = ( \I \otimes P_{d} ) \ket{\psi},
\end{equation}
which gives
\begin{equation}
	\label{eq:psiout-explicit}
	\ket{\psi_{\textnormal{out}}} = \ket{\maxent{d}} \otimes \sqrt{d} ( \I \otimes P_{d} ) \ket{\psi},
\end{equation}
where $\ket{\psi_{d}^\text{max}}$ is the standard maximally entangled state of local dimension $d$. Since $\ket{\psi_{\textnormal{out}}}$ must be a normalised state, we immediately deduce that
\begin{equation}
	\bracket{\psi}{ \I \otimes P_{d} }{\psi} = \frac{1}{d}.
\end{equation}
However, it is intuitively clear that our choice of $P_{d}$ is arbitrary. A slightly different definition of $U_{k}$ in Eq.~\eqref{eq:Uk} would lead to the same conclusion for an arbitrary $P_{x_{1}}$, while swapping the two measurements implies the same for all $Q_{x_{2}}$. Therefore, as a side result of this argument we conclude that the maximal violation necessarily implies that the marginal distributions on Bob are uniform, i.e.:
\begin{equation}
	\bracket{\psi}{ \I \otimes P_{x_{1}} }{\psi} = \bracket{\psi}{ \I \otimes Q_{x_{2}} }{\psi} = \frac{1}{d}
\end{equation}
for all $x_{1}, x_{2}$.
\subsubsection{The condition obtained for Bob's measurements is complete}
\label{app:Bobs-completeness}
In Appendix~\ref{app:meas-Bob} we have shown that if Bob's measurements are capable of producing the maximal violation of $\Bfunc{S}{MUB}$, they must satisfy
\begin{align}\label{rel}
	P_a=dP_aQ_bP_a \qquad \text{and} \qquad Q_b=dQ_bP_aQ_b
\end{align}
for all $a, b \in [d]$. To show that no stronger characterisation  of Bob's measurements is possible based only on the observed Bell violation, we show that any pair of measurements acting on a finite-dimensional Hilbert space which satisfy these relations can be used to produce the maximal violation of $\Bfunc{S}{MUB}$.

By summing the equations in Eq.~\eqref{rel} over $b$ and $a$ respectively, the completeness relation implies $P_a=P_a^2$ and $Q_b=Q_b^2$. Hence, Bob's measurements are projective. Moreover, from Ref.~\cite{Miguel} we know that Eq.~\eqref{rel} implies that all the measurement operators have equal traces, i.e.~for all $a, b$ we have
\begin{equation}
	\label{eq:equal-trace}
	\Tr\left(P_a\right)=\Tr\left(Q_b\right)=n.
\end{equation}
for some integer $n$. Denoting the dimension of the Hilbert space by $D$, the completeness relation gives
\begin{equation}\label{tr2}
	D=\Tr\left(\I\right)=\sum_{a=1}^d \Tr\left(P_a\right)=dn.
\end{equation}
Let us now define $K_{ab}=P_a-Q_b$. Expanding and applying relations~\eqref{rel} lead to
\begin{equation}
	K_{ab}=\frac{d}{d-1} K_{ab}^3.
\end{equation}
This implies that the possible eigenvalues of $K_{ab}$ belong to $\{0,\pm \sqrt{(d-1)/d}\}$. However, from Eq.~\eqref{eq:equal-trace} it follows that $\Tr\left(K_{ab}\right)=0$, which means that there are as many positive eigenvalues of $K_{ab}$ as there are negative ones. In order to find the number of such pairs of eigenvalues, we evaluate
\begin{equation}\label{tr}
	\Tr\left(K_{ab}^2\right)=\Tr\left(P_a+Q_b-\{P_a,Q_b\}\right)=2n-\Tr\left(\{P_a,Q_b\}\right)=2n-2\Tr\left(P_aQ_bP_a\right)=2n \, \frac{d-1}{d},
\end{equation}
where we have used that $\Tr(P_a Q_b)=\Tr(P_aQ_bP_a)$ and the relations \eqref{rel}. Hence, $K_{ab}$ has  $n$ pairs of non-zero eigenvalues.

Now we are ready to construct a quantum realisation which achieves the maximal quantum value of $\Bfunc{S}{MUB}$ using the measurements of Bob discussed above. We assume that the system of Alice is of the same dimension $D$, that Alice and Bob share the maximally entangled state $\ket{\maxent{D}}$ and define Alice's observables as $A_x=\sqrt{d/(d-1)}K_{x_1x_2}\tran$. Notably, the eigenvalues of $A_x$ come from $\{0,\pm 1\}$, so it is a valid observable. The Bell functional $\Bfunc{R}{MUB}$ then reads
\begin{equation}
	\Bfunc{R}{MUB}=\sum_x \bracket{\maxent{D}}{A_x\otimes K_{x_1x_2}}{\maxent{D}}=\sqrt{\frac{d}{d-1}}\sum_x \bracket{\maxent{D}}{ \I \otimes (K_{x_1x_2})^2}{\maxent{D}}=\frac{1}{D}\sqrt{\frac{d}{d-1}}\sum_x \Tr\left((K_{x_1x_2})^2\right)=2\sqrt{d(d-1)},
\end{equation}
where we have used that for any linear operator $O$ we have $O \otimes \I \ket{\maxent{D}} = \I \otimes O\tran \ket{\maxent{D}}$, the fact that the local state of Bob is $\I/D$, and the equations \eqref{tr} and \eqref{tr2}. Moreover, it is easy to check that
\begin{equation}
	\gamma_d\sum_x \bracket{\maxent{D}}{(A_x)^2\otimes \I}{\maxent{D}}=\gamma_d\sum_x \frac{d}{D(d-1)}\Tr\left((K_{x_1x_2})^2\right)=\sqrt{d(d-1)}.
\end{equation}
In conclusion, we arrive at
\begin{equation}
	\Bfunc{S}{MUB}=\Bfunc{R}{MUB}-\gamma_d\sum_x \bracket{\maxent{D}}{(A_x)^2\otimes \I}{\maxent{D}}=\sqrt{d(d-1)},
\end{equation}
which concludes the proof.
\subsection{Maximal quantum violations for MUBs imply a unique probability distribution}\label{AppPdist}
In this section we show that the relations derived in Appendix~\ref{AppTheorem3} are sufficient to reconstruct the entire probability distribution. This implies that the maximal violation is achieved by a unique probability point. In this section we only explicitly compute the probabilities which involve the first two outcomes of Alice. Clearly, the probabilities including the third outcome are determined by normalisation.

Recall that in Eq.~\eqref{cross} we established that
\begin{align}\label{e1}
	A_x\ket{\psi}=\left(A_x^1-A_x^2\right)\ket{\psi}= \mu \left(P_{x_1}-Q_{x_2}\right)\ket{\psi},
\end{align}
where $\mu = \sqrt{d / (d - 1)}$. This immediately implies that
\begin{equation}\label{e2}
	(A_x)^2\ket{\psi}=\left(A_x^1+A_x^2\right)\ket{\psi}= \mu^{2} \left(P_{x_1}-Q_{x_2}\right)^2\ket{\psi}.
\end{equation}
By taking the sum and difference of Eq.~\eqref{e1} and Eq.~\eqref{e2}, we obtain
\begin{equation}
	\begin{aligned}\label{rels}
		& A_x^1\ket{\psi}=\frac{\mu}{2}\left[\left(\mu+1\right)P_{x_1}+\left(\mu-1\right)Q_{x_2}-\mu\{P_{x_1},Q_{x_2}\}\right]\ket{\psi}\\
		& A_x^2\ket{\psi}=\frac{\mu}{2}\left[\left(\mu-1\right)P_{x_1}+\left(\mu+1\right)Q_{x_2}-\mu\{P_{x_1},Q_{x_2}\}\right]\ket{\psi}.
	\end{aligned}
\end{equation}
This allows us to write down the probabilities in terms of expectation values involving only Bob's operators:
\begin{align}
	&\bracket{\psi}{ A_x^1 \otimes P_{u} }{\psi} = \frac{\mu}{2} \left((\mu + 1) \delta_{x_1u} \ave{P_{x_1}} + \big[ \mu(1 - \delta_{x_1u}) - 1 \big] \ave{Q_{x_2} P_{u}} - \mu \ave{ P_{x_1} Q_{x_2} P_{u} } \right),\\
	&\bracket{\psi}{ A_x^{2} \otimes P_{u} }{\psi} = \frac{\mu}{2} \left( (\mu - 1) \delta_{x_1u} \ave{P_{x_1}} + \big[ \mu(1 - \delta_{x_1u}) + 1 \big] \ave{Q_{x_2} P_{u}} - \mu \ave{ P_{x_1} Q_{x_2} P_{u} } \right),\\
	&\bracket{\psi}{ A_x^{1} \otimes Q_{u} }{\psi} = \frac{\mu}{2} \left((\mu - 1) \delta_{x_2u} \ave{Q_{x_2}} + \big[ \mu (1 - \delta_{x_2u}) + 1 \big] \ave{P_{x_1} Q_{u}} - \mu \ave{ Q_{x_2} P_{x_1} Q_{u} } \right),\\
	&\bracket{\psi}{ A_x^{2} \otimes Q_{u} }{\psi} = \frac{\mu}{2} \left( (\mu + 1) \delta_{x_2u} \ave{Q_{x_2}} + \big[ \mu(1 - \delta_{x_2u}) - 1 \big] \ave{P_{x_1} Q_{u}} - \mu \ave{ Q_{x_2} P_{x_1} Q_{u} } \right),
\end{align}
where $\langle \cdot \rangle$ denotes the expectation value with respect to $\ket{\psi}$. In the previous section we have already showed that $\ave{ P_{u} } = \ave{ Q_{u} } = \frac{1}{d}$ for all $y$. To compute the remaining terms we take advantage of the extraction isometries proposed before. Intuitively, we manage to replace the expectation values of the original measurement operators on the original unknown states by the expectation values of the extracted measurement operators on the maximally entangled state. More specifically, we use the fact that since $V_{B}\hc V_{B} = \I$, for any linear operator $O$ on Bob's side we have
\begin{equation}
	\ave{O} = \Tr ( O \rho_{B} ) = \Tr ( V_{B}\hc V_{B} O V_{B}\hc V_{B} \rho_{B} ) = \Tr( V_{B} O V_{B}\hc \cdot V_{B} \rho_{B} V_{B}\hc ).
\end{equation}
In the previous section we have shown that
\begin{equation}
	V_{B} \rho_{B} V_{B}\hc = \I \otimes P_{d} \rho_{B} P_{d}.
\end{equation}
Then, Eq.~\eqref{eq:VB} implies that
\begin{equation}
	V_{B} P_{u} Q_{v} V_{B}\hc = d^{2} \sum_{k} \ketbra{u}{k} \otimes P_{d} Q_{1} P_{u} Q_{v} P_{k} Q_{1} P_{d},
\end{equation}
which leads to
\begin{equation}
	\ave{ P_{u} Q_{v} } = \Tr ( V_{B} P_{u} Q_{v} V_{B}\hc \cdot V_{B} \rho_{B} V_{B}\hc ) = \frac{1}{d} \ave{ P_{d} } = \frac{1}{d^{2}}.
\end{equation}
Clearly, we also have $\ave{ Q_{v} P_{u} } = \ave{ P_{u} Q_{v} }^{*} = 1/d^{2}$. Similarly, we have
\begin{equation}
	V_{B} P_{u} Q_{v} P_{w} V_{B}\hc = d^{2} \ketbra{u}{w} \otimes P_{d} Q_{1} P_{u} Q_{v} P_{w} Q_{1} P_{d},
\end{equation}
which leads to
\begin{equation}
	\ave{ P_{u} Q_{v} P_{w}} = \Tr ( V_{B} P_{u} Q_{v} P_{w} V_{B}\hc \cdot V_{B} \rho_{B} V_{B}\hc ) = \frac{\delta_{uw}}{d} \ave{ P_{d} } = \frac{\delta_{uw}}{d^{2}}.
\end{equation}
To compute the expectation values we have used a particular extraction. If we know consider an extraction procedure which swaps the roles of the first and second measurement of Bob, by symmetry we will conclude that
\begin{equation}
	\ave{ Q_{u} P_{v} Q_{w}} = \frac{\delta_{uw}}{d^{2}}.
\end{equation}
Having computed all the necessary terms we can simply write down the probabilities:
\begin{equation}
	\begin{aligned}
		\bracket{\psi}{ A_x^{1} \otimes P_{u} }{\psi} &=
		\begin{cases}
			\frac{1}{2d} \Big( 1 + \sqrt{ \frac{d - 1}{d} } \, \Big) &\nbox{if} x_1 = u,\\
			\frac{1}{2d (d - 1)} \Big( 1 - \sqrt{ \frac{d - 1}{d} } \, \Big) &\nbox{otherwise.}
		\end{cases}\\
		\bracket{\psi}{ A_x^{2} \otimes P_{u} }{\psi} &=
		\begin{cases}
			\frac{1}{2d} \Big( 1 - \sqrt{ \frac{d - 1}{d} } \, \Big) &\nbox{if} x_1 = u,\\
			\frac{1}{2d (d - 1)} \Big( 1 + \sqrt{ \frac{d - 1}{d} } \, \Big) &\nbox{otherwise.}
		\end{cases}
	\end{aligned}
\end{equation}
and
\begin{equation}
	\begin{aligned}
		\bracket{\psi}{ A_x^{1} \otimes Q_{u} }{\psi} &=
		\begin{cases}
			\frac{1}{2d} \Big( 1 - \sqrt{ \frac{d - 1}{d} } \, \Big) &\nbox{if} x_2 = u,\\
			\frac{1}{2d (d - 1)} \Big( 1 + \sqrt{ \frac{d - 1}{d} } \, \Big) &\nbox{otherwise.}
		\end{cases}\\
		\bracket{\psi}{ A_x^{2} \otimes Q_{u} }{\psi} &=
		\begin{cases}
			\frac{1}{2d} \Big( 1 + \sqrt{ \frac{d - 1}{d} } \, \Big) &\nbox{if} x_2 = u,\\
			\frac{1}{2d (d - 1)} \Big( 1 - \sqrt{ \frac{d - 1}{d} } \, \Big) &\nbox{otherwise.}
		\end{cases}
	\end{aligned}
\end{equation}
In particular, it is easy to check that
\begin{align}
	\bracket{\psi}{A_x^1}{\psi}=\bracket{\psi}{A_x^2}{\psi}=\frac{1}{d}.
\end{align}

\section{Mutually unbiased measurements}\label{AppMate}

In this appendix, we analyse the structure of mutually unbiased measurements (MUMs). As a reminder, we repeat the definition:
\begin{definition}
	Two $d$-outcome measurements $\{P_a\}_{a=1}^d$ and $\{Q_b\}_{b=1}^d$ acting on the Hilbert space $\cH$ are mutually unbiased if
	\begin{equation}\label{eq:MUM_app}
		P_{a}=dP_{a} Q_{b} P_{a} \qquad \text{and}  \qquad  Q_{b}=dQ_{b} P_{a} Q_{b},
	\end{equation}
	for all $a$ and $b$.
\end{definition}
In the following, we define three natural subclasses of MUMs, introduce technical tools for analysing their structures, and through low-outcome number examples we deduce how these subclasses relate to each other.

\subsection{Three relevant subclasses}
Let $P = \{ P_{a} \}_{a = 1}^{d}$ and $Q = \{ Q_{b} \}_{b = 1}^{d}$ be a pair of $d$-outcome MUMs acting on $\cH$.
\begin{definition}
	We say that $P$ and $Q$ are \textbf{mutually unbiased bases} (MUBs) if $\cH = \bC^{d}$ and the measurement operators are rank-one projectors $P_{a} = \ketbraq{u_{a}}$, $Q_{b} = \ketbraq{v_{b}}$. Note that the MUM conditions in Eq.~\eqref{eq:MUM_app} imply that ${ \abs{ \braket{ u_{a} }{ v_{b} } } } = 1/\sqrt{d}$.
\end{definition}
\begin{definition}
	We say that $P$ and $Q$ are a \textbf{direct sum of mutually unbiased bases} if $\cH = \bigoplus_{j} \bC^{d}$ and $P_{a} = \bigoplus_{j} P_{a}^{j}, Q_{b} = \bigoplus_{j} Q_{b}^{j}$, where for every $j$ the pair $P^{j}$ and $Q^{j}$ are mutually unbiased bases.
	In the following, we will denote this class by $\textnormal{MUB}_{\oplus}$.
\end{definition}
\begin{definition}
	We say that $P$ and $Q$ are \textbf{MUB-extractable} if there exists a completely positive unital map $\Lambda : \cB( \cH ) \to \cB( \bC^{d} )$ such that the measurements defined as
	\begin{equation}
		P_{a}' = \Lambda( P_{a} ) \nbox{and} Q_{b}' = \Lambda( Q_{b} )
	\end{equation}
	are mutually unbiased bases.
	In the following, we will denote this class by $\textnormal{MUB}_\textnormal{ext}$.
\end{definition}

It follows directly from the above definitions that
\begin{equation}\label{eq:MUMinclusions_def}
	\text{MUB} \subsetneq \text{MUB}_\oplus \subseteq \text{MUB}_\text{ext} \subseteq \text{MUM},
\end{equation}
where the first inclusion is trivially seen to be strict. In the remainder of this appendix, we show that in general all of the above inclusions are strict.

\subsection{Technical tools}
In this section we introduce some technical tools to analyse the structure of MUMs. First, following Ref.~\cite{Miguel}, we derive a canonical form of MUMs in which they are fully characterised by a collection of unitary operators. Then, we provide a necessary and sufficient condition for MUMs to be unitarily equivalent to a direct sum of MUBs in terms of these unitaries. Lastly, we derive a generic lemma on the extractability of arbitrary sets of Hermitian operators via completely positive unital maps. In the subsequent section, we apply these techniques to analyse the inclusions within the MUM subclasses for outcome numbers $d = 2, 3, 4, 5$.

In this appendix we use $\cH$ to denote a separable Hilbert space, $\cB(\cH)$ to denote the set of bounded operators acting on $\cH$ and $\cB_{H}(\cH)$ to denote the set of bounded Hermitian operators acting on $\cH$.

\subsubsection{Canonical form of mutually unbiased measurements}
\label{app:canonical-form}

Let us take a pair of MUMs $\{P_a\}_{a = 1}^d$ and $\{Q_b\}_{b = 1}^d$ on a separable Hilbert space $\cH$. That is, measurements that satisfy the relations in Eq.~\eqref{eq:MUM_app}, namely
\begin{equation}\label{sandwichapp}
	\forall a,b: \quad P_a=dP_a Q_{b} P_a  \quad  \text{and}  \quad  Q_{b}=dQ_{b} P_a Q_{b}.
\end{equation}
Following Ref.~\cite{Miguel} we provide a characterisation of such a measurement pair.

First, note that these conditions imply projectivity, which can be seen by summing over the middle term. Then, defining $O_{ab} = \sqrt{d} P_{a} Q_{b}$, it is easy to see that
\begin{align}
	O_{ab} O_{ab}\hc = P_{a},\\
	O_{ab}\hc O_{ab} = Q_{b}.
\end{align}
This means that all the projectors are isomorphic, i.e.~either all $P_{a}, Q_{b}$ are finite-rank (and then $\Tr P_{a} = \Tr Q_{b} = n$ for all $a, b$ for some fixed $n \in \mathbb{N}$) or none of them is. Let $\cH_{a}$ denote the subspace on which $P_{a}$ projects. The completeness relation $\sum_{a} P_{a} = \I$ implies that
\begin{equation}
	\cH \simeq \bigoplus_{a = 1}^{d} \cH_{a}.
\end{equation}
However, the fact that all these Hilbert spaces are isomorphic allows us to write
\begin{equation}
	\cH \simeq \cH' \otimes \bC^{d}
\end{equation}
for some other (potentially infinite-dimensional but still separable) Hilbert space $\cH'$. Then, the first measurement reads
\begin{equation}
	\label{eq:Pa-operator}
	P_{a} = \I \otimes \ketbraq{a},
\end{equation}
where $\{ \ket{a} \}_{a = 1}^{d}$ is an orthonormal basis on $\bC^{d}$. The second measurement can be written as
\begin{equation}
	Q_{b} = \frac{1}{d} \sum_{jk} X_{jk}^{b} \otimes \ketbra{j}{k}
\end{equation}
for some operators $X_{jk}^{b} \in \cB( \cH' )$. Since $Q_{b} = Q_{b}\hc$, we must have that $X_{jk}^{b} = [ X_{kj}^{b} ]\hc$. Moreover, it is easy to show that all the $X_{jk}^{b}$ are unitary. Note that
\begin{equation}
	d^{2} P_{j} Q_{b} P_{k} Q_{b} P_{j} = d P_{j} Q_{b} P_{j} = P_{j} = \I \otimes \ketbraq{j}.
\end{equation}
On the other hand, a direct calculation gives
\begin{equation}
	d^{2} P_{j} Q_{b} P_{k} Q_{b} P_{j} = X_{jk}^{b} X_{kj}^{b} \otimes \ketbraq{j}.
\end{equation}
Since $X_{kj}^{b} = [ X_{jk}^{b} ]\hc$, we obtain $X_{jk}^{b} [X_{jk}^{b}]\hc = \I$, and by symmetry $[X_{jk}^{b}]\hc X_{jk}^{b} = \I$.

The first MUM condition $P_{j} Q_{b} P_{j} = P_{j} / d$ implies that
\begin{equation}
	X_{jj}^{b} = \I
\end{equation}
for all $j$ and $b$. The second MUM condition has an interesting implication: $Q_{b} P_{j} Q_{b} = Q_{b} / d$ leads to
\begin{equation}
	X_{jk}^{b} = X_{ja}^{b} X_{ak}^{b}
\end{equation}
for all $j, k, a, b$. In other words, the entire projector $Q_{b}$ is determined by only $d$ of the $X_{jk}^{b}$ operators. For instance, we can write
\begin{equation}
	X_{jk}^{b} = X_{j1}^{b} X_{1k}^{b}
\end{equation}
and recall that $X_{11}^{b} = \I$.

The structure derived above allows us to introduce a unitary transformation on $\cH$. Let
\begin{equation}
	U \equiv \sum_{j} X_{1j}^{1} \otimes \ketbraq{j}.
\end{equation}
Clearly, $U P_{a} U\hc = P_{a}$, and
\begin{equation}
	Q_{b}' \equiv U Q_{b} U\hc = \frac{1}{d} \sum_{jk} X_{1j}^{1} X_{jk}^{b} X_{k1}^{1} \otimes \ketbra{j}{k}.
\end{equation}
It is easy to see that
\begin{equation}
	\label{eq:Qb'-operator}
	Q_{1}' = \I \otimes \ketbraq{v},
\end{equation}
where
\begin{equation}
	\ket{v} = \frac{1}{\sqrt{d}} \sum_{j = 1}^{d} \ket{j}
\end{equation}
is a uniform superposition of all the basis states. In other words, this unitary transformation fixes the form of $Q_{1}'$.

Now we characterise the remaining $Q_{b}'$. Let
\begin{equation}
	Y_{jk}^{b} \equiv X_{1j}^{1} X_{jk}^{b} X_{k1}^{1}.
\end{equation}
These are still unitary and it holds for these operators as well that $Y_{jj}^b = \I$ and $Y_{jk}^b = Y_{j1}^b Y_{1k}^b$. Let us denote $V_j^b \equiv Y_{j1}^b$, and therefore
\begin{equation}
	Q'_b = \frac1d \sum_{jk} V_j^b (V_k^b)\hc \otimes \ketbra{j}{k}.
\end{equation}
From the previous arguments we have that $V_j^1 = V_1^b = \I$ for all $j$, $b$.
The orthogonality constraint $Q'_b Q'_{b'} = 0$ for $b \neq b'$ implies
\begin{equation}
	\sum_j ( V_j^b )\hc V_j^{b'} = 0 ~~~ \forall b \neq b',
\end{equation}
whereas the completeness relation $\sum_b Q'_b = \I$ gives
\begin{equation}
	\sum_b V_j^b (V_k^b)\hc = \delta_{jk} d \, \I.
\end{equation}

The following two propositions, whose results appeared originally in Ref.~\cite{Miguel}, summarise the observations we have made so far.
\begin{proposition}\label{prop:sandwichchar}
	Let $\{P_a\}_{a = 1}^d$ and $\{Q_b\}_{b = 1}^d$ be two $d$-outcome mutually unbiased measurements on a separable Hilbert space $\cH$. Then, the Hilbert space $\cH$ is isomorphic to $\cH' \otimes \bC^d$, for some Hilbert space $\cH'$, and the measurement operators can be written as
	\begin{equation}\label{eq:sandwich_alt}
		\begin{split}
			& \left. P_{a} = \I \otimes \ketbraq{a}, \right. \\
			& \left. Q_b = \frac1d \sum_{jk} V_j^b (V_k^b)\hc \otimes \ketbra{j}{k}, \right. \\
			& \left. V_j^b (V_j^b)\hc = (V_j^b)\hc V_j^b = \I ~~~ \forall j,b, \right. \\
			& \left. V_j^1 = V_1^b = \I ~~~ \forall j, b, \right. \\
			& \left. \sum_j (V_j^b)\hc V_j^{b'} = 0 ~~~ \forall b \neq b', \right. \\
			& \left. \sum_b V_j^b (V_k^b)\hc = \delta_{jk} d \, \I ~~~ \forall j,k. \right.
		\end{split}
	\end{equation}
\end{proposition}
We will refer to the representation of $\{ P_{a} \}$ and $\{ Q_{b} \}$ in terms of $\{V_{j}^{b}\}$ as the \emph{canonical form}.

\subsubsection{Condition for equivalence to direct sum of MUBs}

In the following, we show a necessary and sufficient condition for a pair of MUMs to be unitarily equivalent to a direct sum of MUBs.
\begin{proposition}\label{prop:comm}
	For a pair of MUMs $\{P_a\}_{a = 1}^d$ and $\{Q_b\}_{b = 1}^d$ the following statements are equivalent:
	\begin{enumerate}
		\item If we represent $\{P_{a}\}$ and $\{Q_{b}\}$ in the canonical form, then all the $\{V_{j}^{b}\}$ matrices commute:
		\begin{equation}\label{eq:Vcomm}
			[V_j^b , V_{j'}^{b'}] = 0 ~~~ \forall j, j', b, b'.
		\end{equation}
		\item The measurements $\{P_{a}\}$ and $\{Q_{b}\}$ correspond to a direct sum of MUBs.
	\end{enumerate}
	%
	%
\end{proposition}

\begin{proof}
	The fact that $(2) \implies (1)$ is clear, so let us focus on $(1) \implies (2)$. From Eq.~\eqref{eq:Vcomm} it follows that all the $Y_{jk}^{b}$ operators commute, i.e.~that
	\begin{equation}
		[ Y_{jk}^{b}, Y_{j'k'}^{b'} ] = 0
	\end{equation}
	for all $j, j', k, k', b, b'$. This implies that one can find a basis on $\cH'$ denoted by $\{ \ket{e_{n}} \}_{n}$ such that
	\begin{equation}
		Y_{jk}^{b} = \sum_{n} \lambda_{n}^{jk, b} \ketbraq{e_{n}},
	\end{equation}
	with $| \lambda_{n}^{jk, b} | = 1$. Then
	\begin{equation}\label{eq:Qdiag}
		Q_{b} = \sum_{n} \ketbraq{e_{n}} \otimes T_{n, b},
	\end{equation}
	where
	\begin{equation}
		T_{n, b} \equiv \frac{1}{d} \sum_{jk} \lambda_{n}^{jk, b} \ketbra{j}{k}.
	\end{equation}
	Since $Q_{b}$ is Hermitian, so must be the operators $T_{n, b}$, and since $Q_{b}$ is projective, the operators $T_{n, b}$ are also projectors. In fact, one can compute the trace to verify that they are rank-one projectors. Since $Y_{jj}^{b} = \I$, we have that $\lambda_{n}^{jj, b} = 1$, and
	\begin{equation}
		\Tr T_{n, b} = \frac{1}{d} \sum_{j} \lambda_{n}^{jj, b} = 1.
	\end{equation}
	We also have that $ \bra{a} T_{n, b} \ket{a} = \frac1d$, and therefore $\{ \ketbraq{a} \}_{a = 1}^d$ and $\{ T_{n, b} \}_{b = 1}^d$ form a pair of MUBs in dimension $d$. By looking at Eqs.~\eqref{eq:sandwich_alt} and \eqref{eq:Qdiag} we immediately see that $\{P_a\}$ and $\{Q_b\}$ can be written as direct sums of MUBs.
\end{proof}
%
%
\subsubsection{Condition for non-extractability}
In this section we derive a condition guaranteeing that a pair of MUMs cannot be transformed into a pair of MUBs under the action of a completely positive unital map. This condition can be understood as a certificate that the particular pair of MUMs does not belong to the class MUB$_\text{ext}$ defined earlier. Let us start with a more general technical statement. In the following we use the fact that for a Hermitian bounded operator $M$ its spectrum, denoted by $\spec(M)$, is a bounded and closed subset of $\mathbb{R}$ and hence $\max \{ \spec(M) \}$ is well-defined.

\begin{lemma}\label{lem:CPunital}
	Let $\{A_k\}_{k=1}^n \subset \cB_{H}(\cH_A)$ be a set of bounded Hermitian operators on $\cH_A$, and $\{B_k\}_{k=1}^n \subset \cB_{H}(\cH_B)$ be a set of bounded Hermitian operators on $\cH_B$ such that $\dim \cH_B < \infty$. Then, if
	\begin{equation}\label{eq:condition}
		\dim \cH_B \cdot \max \Big\{ \spec \Big( \sum_k A_k\tran \otimes B_k \Big) \Big\} < \sum_k \Tr B_k^2,
	\end{equation}
	then there does not exist a completely positive unital map $\Lambda: \cB( \cH_A ) \to \cB( \cH_B )$ such that $\Lambda(A_k) = B_k$ for all $k = 1, \ldots, n$.
\end{lemma}

\begin{proof}
	The existence of the above unital map can be written as the semidefinite programming feasibility problem \cite{Boyd}
	\begin{equation}\label{eq:primal}
		\begin{split}
			& \left. C \in \cB_{H}( \cH_A \otimes \cH_B ) \right. \\
			& \left. C \ge 0 \right. \\
			& \left. \Tr_A C = \I_B \right. \\
			& \left. \Tr_A \left[ C \left( A_k\tran \otimes \I_B \right) \right] = B_k ~~~ \forall k, \right.
		\end{split}
	\end{equation}
	due to the Choi--Jamio\l kowski isomorphism \cite{Choi, Jamiolkowski}. Let us define Hermitian dual variables $X \in \cB_{H}( \cH_A \otimes \cH_B )$, $Y \in \cB_{H}( \cH_B )$ and $Z_k \in \cB_{H}( \cH_B )$ for the above constraints, and define the Lagrangian function
	\begin{equation}\label{eq:Lagrangian}
		\cL(X, Y, \{Z_k\} ) \equiv \sup_{ C \in \cB_{H}( \cH_A \otimes \cH_B ) } \bigg\{ \Tr(XC) + \Tr \left[ Y \left( \I_B - \Tr_A C \right) \right] + \sum_k \Tr\Big[ Z_k \big\{ B_k - \Tr_A \left[ C \left( A_k\tran \otimes \I_B \right) \right] \big\} \Big] \bigg\},
	\end{equation}
	which is guaranteed to be real. We also define the dual problem
	\begin{equation}\label{eq:dual}
		\begin{split}
			& \left. \cL(X, Y, \{Z_k\} ) < 0 \right. \\
			& \left. X \ge 0. \right.
		\end{split}
	\end{equation}
	The primal in Eq.~\eqref{eq:primal} and the dual are \emph{weak alternatives}, that is, they cannot be both feasible (there do not exist variables $C$, $X$, $Y$ and $Z_k$ satisfying all the constraints in both Eqs.~\eqref{eq:primal} and \eqref{eq:dual}), because in that case we would have
	\begin{equation}
		0 > \cL(X, Y, \{Z_k\} ) \ge \Tr(XC) + \Tr \left[ Y \left( \I_B - \Tr_A C \right) \right] + \sum_k \Tr\Big[ Z_k \big\{ B_k - \Tr_A \left[ C \left( A_k\tran \otimes \I_B \right) \right] \big\} \Big] \ge 0,
	\end{equation}
	which is a contradiction. Therefore, if we find a feasible point for the dual in Eq.~\eqref{eq:dual}, then the primal in Eq.~\eqref{eq:primal} is infeasible.
	
	Let us rewrite the Lagrangian in Eq.~\eqref{eq:Lagrangian} as
	\begin{equation}
		\begin{split}
			\cL(X, Y, \{Z_k\} ) & \left. = \sup_{ C \in \cB_{H}( \cH_A \otimes \cH_B ) } \Tr \bigg\{ C \Big[ X - \I_A \otimes Y - \sum_k \left(A_k\tran \otimes \I_B \right) \left( \I_A \otimes Z_k \right) \Big] \bigg\} + \Tr Y + \sum_k \Tr (Z_k B_k) \right. \\
			& \left. \equiv \sup_{ C \in \cB_{H}( \cH_A \otimes \cH_B ) } \Tr (C X') + \Tr Y + \sum_k \Tr (Z_k B_k), \right.
		\end{split}
	\end{equation}
	where we have defined $X' = X - \I_A \otimes Y - \sum_k \left(A_k\tran \otimes Z_k \right)$ and used the fact that the dual map of the partial trace is $\Tr_A^\ast(.) = \I_A \otimes (.)$. In order to have $\cL(X, Y, \{Z_k\} ) < +\infty$, we need to set $X' = 0$, because there is no restriction imposed on $C$ in the dual problem. Therefore, an equivalent formulation of the dual feasibility problem in Eq.~\eqref{eq:dual} is
	\begin{equation}\label{eq:dual_eqv}
		\begin{split}
			& \left. \Tr Y + \sum_k \Tr (Z_k B_k) < 0 \right. \\
			& \left. \I_A \otimes Y + \sum_k \left(A_k\tran \otimes Z_k \right) \ge 0. \right.
		\end{split}
	\end{equation}
	
	We choose the ansatz $Y = y \I_B$ and $Z_k = -z B_k$ with $y, z \in \mathbb{R}$, and substitute it into Eq.~\eqref{eq:dual_eqv}, which gives
	\begin{equation}\label{eq:dual_yz}
		\begin{split}
			& \left. y \cdot \dim \cH_B - z \cdot \sum_k \Tr B^2_k < 0 \right. \\
			& \left. y \I - z \cdot \sum_k \left(A_k\tran \otimes B_k \right) \ge 0. \right.
		\end{split}
	\end{equation}
	Satisfying the second constraint in general requires $y \ge 0$ (specifically if $\sum_k \left(A_k\tran \otimes B_k \right)$ is not full rank), and therefore we also need $z \ge 0$ in order to satisfy the first constraint. In order to get the lowest value in the first constraint, it is desirable to satisfy the second inequality with equality, which leads to
	\begin{equation}
		y = z \cdot \max \Big\{ \spec \Big( \sum_k A_k\tran \otimes B_k \Big) \Big\}.
	\end{equation}
	Plugging this into Eq.~\eqref{eq:dual_yz} gives
	\begin{equation}
		\begin{split}
			& \left. z \cdot \bigg[
			\dim \cH_B \cdot \max \Big\{ \spec \Big( \sum_k A_k\tran \otimes B_k \Big) \Big\} - \sum_k \Tr B_k^2 \bigg] < 0 \right. \\
			& \left. z \ge 0, \right.
		\end{split}
	\end{equation}
	which is feasible whenever
	\begin{equation}
		\mathrm{dim}\cH_B \cdot \max \Big\{ \spec \Big( \sum_k A_k\tran \otimes B_k \Big) \Big\} < \sum_k \Tr B_k^2.
	\end{equation}
\end{proof}
\begin{remark}\label{rem:CPunital}
	Consider two sets of $n$ projective measurements with $m$ outcomes, $\{ A^y_b \}_{b = 1, \ldots, m}^{y = 1, \ldots, n} \subset \cB_{H}(\cH_A)$ and $\{ B^y_b \}_{b = 1, \ldots, m}^{y = 1, \ldots, n} \subset \cB_{H}(\cH_B)$. Since for positive semidefinite operators $\max \{ \spec(M) \} = \norm{M}$, where $\norm{.}$ is the operator norm, in this case the criterion in Eq.~\eqref{eq:condition} reads
	\begin{equation}\label{eq:CPnormcondition}
		\norm[\Big]{ \sum_{y, b} (A_b^y)\tran \otimes B_b^y } < n.
	\end{equation}
	Using the triangle inequality and the fact that every $\sum_{b} (A_b^y)\tran \otimes B_b^y$ is a projection, we have that
	\begin{equation}
		\norm[\Big]{ \sum_{y, b} (A_b^y)\tran \otimes B_b^y } \le \sum_y \norm[\Big]{ \sum_{b} (A_b^y)\tran \otimes B_b^y } = n.
	\end{equation}
	The saturation of this inequality is equivalent to the existence of a state $\ket{\psi} \in \cH_A \otimes \cH_B$ such that
	\begin{equation}\label{eq:statecondition}
		\bra{\psi} \sum_{b} (A_b^y)\tran \otimes B_b^y \ket{\psi} = 1 ~~ \forall y.
	\end{equation}
	Therefore, Eq.~\eqref{eq:condition} for two sets of $n$ projective measurement is equivalent to the non-existence of a state $\ket{\psi}$ such as in Eq.~\eqref{eq:statecondition}.
\end{remark}

\subsection{Examples}

In this section, using the above techniques, we show that for outcome numbers $d = 2$ (3), every MUM pair is unitarily equivalent to a direct sum of MUBs in dimension $2$ (3). Therefore, for $d = 2$ and 3, $\text{MUB}_\oplus = \text{MUB}_\text{ext} = \text{MUM}$. However, we also show that this is not the case for $d = 4$ and 5, where we show explicit examples of MUM pairs that are not MUB-extractable. This shows that in general $\text{MUB}_\text{ext} \subsetneq \text{MUM}$.

\subsubsection{Outcome numbers 2 and 3}
\label{app:outcome23}

\begin{proposition}
	For $d = 2$ and $d = 3$, every pair of mutually unbiased measurements can be written as a direct sum of $d$-dimensional MUBs.
\end{proposition}

\begin{proof}
	From Eq.~\eqref{eq:sandwich_alt} it is clear that for every pair of MUMs we have that
	\begin{equation}\label{eq:zerosum}
		\begin{split}
			& \left. V^1_j = V^b_1 = \I ~~~ \forall b, j, \right. \\
			& \left. \sum_j V^b_j = 0 ~~~ \forall b \neq 1, \right. \\
			& \left. \sum_b V^b_j = 0 ~~~ \forall j \neq 1. \right.
		\end{split}
	\end{equation} For $d = 2$ this fixes all $V^b_j$, that is, $V^1_1 = V^1_2 = V^2_1 = \I$, whereas $V^2_2 = -\I$. Then all $V^b_j$ commute, and by Proposition \ref{prop:comm} the measurements are direct sums of 2-dimensional MUBs.
	
	For the $d = 3$ case, we have that $V^1_1 = V^1_2 = V^1_3 = V^2_1 = V^3_1 = \I$. Let us denote $V^2_2 = V$. Then from Eq.~\eqref{eq:zerosum} it follows that $V^2_3 = V^3_2 = - \I - V$ and $V^3_3 = V$, and again all the matrices $\{ V^b_j \}$ commute.
\end{proof}

\subsubsection{Outcome numbers 4 and 5}
\label{app:outcome45}

Let us now construct two 5-outcome MUMs $\{P_a\}_{a = 1}^5$ and $\{Q_b\}_{b = 1}^5$ in dimension 10, such that they are not MUB-extractable. We will define the operators $P_a$ and $Q_b$ formally on the space $\cB( \bC^{2} \otimes \bC^{5} )$, and we will denote the qubit Pauli operators by $X$, $Y$ and $Z$. We define the $P_a$ operators as
\begin{equation}\label{eq:P5}
	P_a = \I_2 \otimes \ketbraq{a},
\end{equation}
where $\{ \ket{a} \}_{a = 1}^5$ is the computational basis on $\bC^5$, and the first $Q$ operator as
\begin{equation}\label{eq:Q51}
	Q_1 = \I_2 \otimes \ketbraq{v}, ~~ \ket{v} = \frac{1}{ \sqrt{5} } \sum_{j = 1}^5 \ket{j}.
\end{equation}
For the next three $Q$ operators, we define unitaries that will transform $Q_1$ into $Q_2$, $Q_3$ and $Q_4$:
\begin{equation}
	U_b = \sum_{j = 1}^5 U^b_j \otimes \ketbraq{k}, ~~ b = 2, 3, 4,
\end{equation}
where
\begin{equation}
	\begin{split}
		& \left. U^2_1 = Z, ~~~ U^2_2 = X, ~~~ U^2_3 = - \frac12 X + \frac{ \sqrt{3} }{2} Y, ~~~ U^2_4 = - \frac12 X - \frac{ \sqrt{3} }{2} Y, ~~~ U^2_5 = - Z, \right. \\
		& \left. U^3_1 = Z, ~~~ U^3_2 = - \frac12 X - \frac{ \sqrt{3} }{2} Y, ~~~ U^3_3 = - Z, ~~~ U^3_4 = X, ~~~ U^3_5 = - \frac12 X + \frac{ \sqrt{3} }{2} Y, \right. \\
		& \left. U^4_1 = Z, ~~~ U^4_2 = - Z, ~~~ U^4_3 = - \frac12 X - \frac{ \sqrt{3} }{2} Y, ~~~ U^4_4 = - \frac12 X + \frac{ \sqrt{3} }{2} Y, ~~~ U^4_5 = X \right. \\
	\end{split}
\end{equation}
(note that these unitaries are not the same as the $V^b_j$ in Eq.~\eqref{eq:sandwich_alt}). Using these unitaries, we define
\begin{equation}\label{eq:Q5234}
	Q_b = U_b Q_1 U_b^\dagger, ~~ b = 2, 3, 4,
\end{equation}
and finally,
\begin{equation}\label{eq:Q55}
	Q_5 = \I_{10} - Q_1 - Q_2 - Q_3 - Q_4.
\end{equation}

\begin{proposition}\label{prop:d5}
	The measurements $\{P_a\}_{a = 1}^5$ and $\{Q_b\}_{b = 1}^5$ defined in Eqs.~\eqref{eq:P5}, \eqref{eq:Q51}, \eqref{eq:Q5234} and \eqref{eq:Q55} are mutually unbiased, but they are not MUB-extractable.
\end{proposition}

\begin{proof}
	
	It is straightforward to check that these measurements satisfy the relations in Eq.~\eqref{sandwichapp}, and are therefore MUMs (see the attached Mathematica file ``dim5.nb''). Now we will show, using Lemma \ref{lem:CPunital}, that there does not exist a completely positive unital map $\Lambda : \cB( \bC^{10} ) \to \cB( \bC^5 )$ such that $\Lambda( P_a ) = A_a$ and $\Lambda( Q_b ) = B_b$, where $A_a$ and $B_b$ are projectors onto a pair of MUBs in dimension~5.
	
	From Ref.~\cite{Brierley}, we know that up to a global unitary transformation and the reordering of the elements of $\{B_b\}$, every pair of MUBs in dimension 5 can be written as
	\begin{equation}
		A_a = \ketbraq{a}, ~~~ B_b = F_5 A_b F_5^\dagger,
	\end{equation}
	where $F_5$ is the Fourier matrix in dimension five, defined by its elements
	\begin{equation}
		\bra{a} F_5 \ket{b} = \frac{1}{ \sqrt{5} } \omega^{ (a - 1)(b - 1) }, ~~~ a, b = 1, \ldots, 5, ~~~ \omega = \mathrm{e}^{ \frac{ 2 \pi \mathrm{i} }{5} }.
	\end{equation}
	According to Remark \ref{rem:CPunital} and Eq.~\eqref{eq:statecondition} therein, a completely positive unital map $\Lambda : \cB( \bC^{10} ) \to \cB( \bC^5 )$ such that $\Lambda( P_a ) = A_a$ and $\Lambda( Q_b ) = B_b$ does not exist, if there does not exist a state $\ket{\psi} \in \bC^{10} \otimes \bC^5$, such that
	\begin{equation}
		\bra{\psi} \sum_a P_a\tran \otimes A_a \ket{\psi} = \bra{\psi} \sum_b Q_b\tran \otimes B_b \ket{\psi} = 1.
	\end{equation}
	Note that due to the freedom in permuting the elements of $\{B_b\}$, we need to check whether a state $\ket{\psi} \in \bC^{10} \otimes \bC^5$ exists such that
	\begin{equation}\label{eq:finalcondition}
		\bra{\psi} \sum_a P_a\tran \otimes A_a \ket{\psi} = \bra{\psi} \sum_b Q_b\tran \otimes B_{ \sigma(b) } \ket{\psi} = 1
	\end{equation}
	for every permutation $\sigma \in S_5$ on the set $\{1, 2, 3, 4, 5\}$.
	
	In order to rule out the existence of such a state, first notice that the operator $\sum_a P_a\tran \otimes A_a =: \cP$ is a projection. Therefore, if a state $\ket{\psi}$ such as in Eq.~\eqref{eq:finalcondition} exists, then we have that $\cP \ket{\psi} = \ket{\psi}$, and therefore
	\begin{equation}
		\bra{\psi} \cP \sum_b Q_b\tran \otimes B_{ \sigma(b) } \cP \ket{\psi} = 1.
	\end{equation}
	However, using the fact that for any operator $M$, $\norm{M}_\infty \le \norm{M}_2$, where $\norm{\cdot}_{p}$ is the Schatten $p$-norm, we have that
	\begin{equation}
		\begin{split}
			\bra{\psi} \cP \sum_b Q_b\tran \otimes B_{ \sigma(b) } \cP \ket{\psi} & \left. \le \norm[\Big]{ \cP \sum_b Q_b\tran \otimes B_{ \sigma(b) } \cP }_\infty \le \norm[\Big]{ \cP \sum_b Q_b\tran \otimes B_{ \sigma(b) } \cP }_2 \right. \\
			& \left. = \sqrt{ \Tr \left( \cP \sum_b Q_b\tran \otimes B_{ \sigma(b) } \cP \sum_{b'} Q_{b'}\tran \otimes B_{ \sigma(b') } \cP \right) } < 1 ~~~ \forall \sigma \in S_5, \right.
		\end{split}
	\end{equation}
	which is straightforward to verify (see the attached Mathematica file ``dim5.nb'').
	
\end{proof}

Let us also provide an explicit example of a pair of 4-outcome MUMs $\{P_a\}_{a = 1}^4$ and $\{Q_b\}_{b = 1}^4$ in dimension 8, such that they are not MUB-extractable. Similarly to the previous example, we write
\begin{equation}\label{eq:P4Q41}
	\begin{split}
		& \left. P_a = \I_2 \otimes \ketbraq{a}, \right. \\
		& \left. Q_1 = \I_2 \otimes \ketbraq{v}, ~~ \ket{v} = \frac12 \sum_{j = 1}^4 \ket{j}, \right.
	\end{split}
\end{equation}
and we define $U_2$ that will transform $Q_1$ into $Q_2$:
\begin{equation}
	U_2 = \sum_{j = 1}^4 U^2_j \otimes \ketbraq{k},
\end{equation}
where
\begin{equation}
	U^2_1 = \frac{1}{ \sqrt{3} } (X + Y + Z), ~~~ U^2_2 = \frac{1}{ \sqrt{3} } (X - Y - Z), U^2_3 = \frac{1}{ \sqrt{3} } (- X + Y - Z), U^2_4 = \frac{1}{ \sqrt{3} } (- X - Y + Z).
\end{equation}
Using these unitaries, we define
\begin{equation}\label{eq:Q42}
	Q_2 = U_2 Q_1 U_2^\dagger.
\end{equation}
Next, let us define
\begin{equation}\label{eq:Q43}
	Q_3 = \frac14
	\begin{pmatrix}
		1 & 0 & 0 & -1 & -1 & 0 & 0 & 1 \\
		0 & 1 & 1 & 0 & 0 & -1 & -1 & 0 \\
		0 & 1 & 1 & 0 & 0 & -1 & -1 & 0 \\
		-1 & 0 & 0 & 1 & 1 & 0 & 0 & -1 \\
		-1 & 0 & 0 & 1 & 1 & 0 & 0 & -1 \\
		0 & -1 & -1 & 0 & 0 & 1 & 1 & 0 \\
		0 & -1 & -1 & 0 & 0 & 1 & 1 & 0 \\
		1 & 0 & 0 & -1 & -1 & 0 & 0 & 1 \\
	\end{pmatrix},
\end{equation}
and finally,
\begin{equation}\label{eq:Q44}
	Q_4 = \I_{8} - Q_1 - Q_2 - Q_3.
\end{equation}

\begin{proposition}\label{prop:d4}
	The measurements $\{P_a\}_{a = 1}^4$ and $\{Q_b\}_{b = 1}^4$ defined in Eqs.~\eqref{eq:P4Q41}, \eqref{eq:Q42}, \eqref{eq:Q43} and \eqref{eq:Q44} are mutually unbiased, but they are not MUB-extractable.
\end{proposition}

\begin{proof}
	It is straightforward to check that these measurements satisfy the relations in Eq.~\eqref{sandwichapp}, and are therefore mutually unbiased (see the attached Mathematica file ``dim4.nb''). Now we will show, using Lemma \ref{lem:CPunital}, that there does not exist a completely positive unital map $\Lambda : \cB( \bC^{8} ) \to \cB( \bC^4 )$ such that $\Lambda( P_a ) = A_a$ and $\Lambda( Q_b ) = B_b$, where $A_a$ and $B_b$ are projectors onto a pair of MUBs in dimension 4.
	
	From Ref.~\cite{Brierley}, we know that up to a global unitary transformation and the reordering of the elements of $\{B_b\}$, every pair of MUBs in dimension 4 can be written as
	\begin{equation}
		A_a = \ketbraq{a}, ~~~ B_b(x) = F_4(x) A_b F_4(x)^\dagger, ~~~ x \in [0, 2\pi),
	\end{equation}
	where $F_4(x)$ is the one-parameter family of complex Hadamard matrices in dimension 4, defined as 
	\begin{equation}
		F_4(x) = \frac12
		\begin{pmatrix}
			1 & 1 & 1 & 1 \\
			1 & 1 & -1 & -1 \\
			1 & -1 & \mathrm{i} \mathrm{e}^{ \mathrm{i} x } & - \mathrm{i} \mathrm{e}^{ \mathrm{i} x } \\
			1 & -1 & - \mathrm{i} \mathrm{e}^{ \mathrm{i} x } & \mathrm{i} \mathrm{e}^{ \mathrm{i} x }
		\end{pmatrix}
		~~~ x \in [0, 2\pi).
	\end{equation}
	According to Remark \ref{rem:CPunital} and Eq.~\eqref{eq:CPnormcondition} therein, a completely positive unital map $\Lambda : \cB( \bC^{8} ) \to \cB( \bC^4 )$ such that $\Lambda( P_a ) = A_a$ and $\Lambda( Q_b ) = B_{ \sigma(b) }(x)$ does not exist, if
	\begin{equation}\label{eq:assop}
		\norm[\Big]{ \sum_a P_a\tran \otimes A_a + \sum_b Q_b\tran \otimes B_{ \sigma(b) }(x) }_\infty < 2,
	\end{equation}
	and we need to check this condition for all $\sigma \in S_4$ permutations and all values of $x \in [0, 2\pi)$. This can be done numerically up to machine precision by computing the largest eigenvalue of the above operator. Fig.~\ref{fig:ass} contains the largest eigenvalue, also maximised over $\sigma \in S_4$, for 10000 different values of $x \in [0, 2\pi)$ (also see the attached Matlab files ``dim4\_plot.m'' and ``dim4\_example\_plot'' (written for Octave) to generate the plot). It is apparent that the norm is always strictly smaller than 2, and hence the numerical evidence is convincing.
	
	\begin{figure}[h!]
		\centering
		\includegraphics[width=0.8\columnwidth]{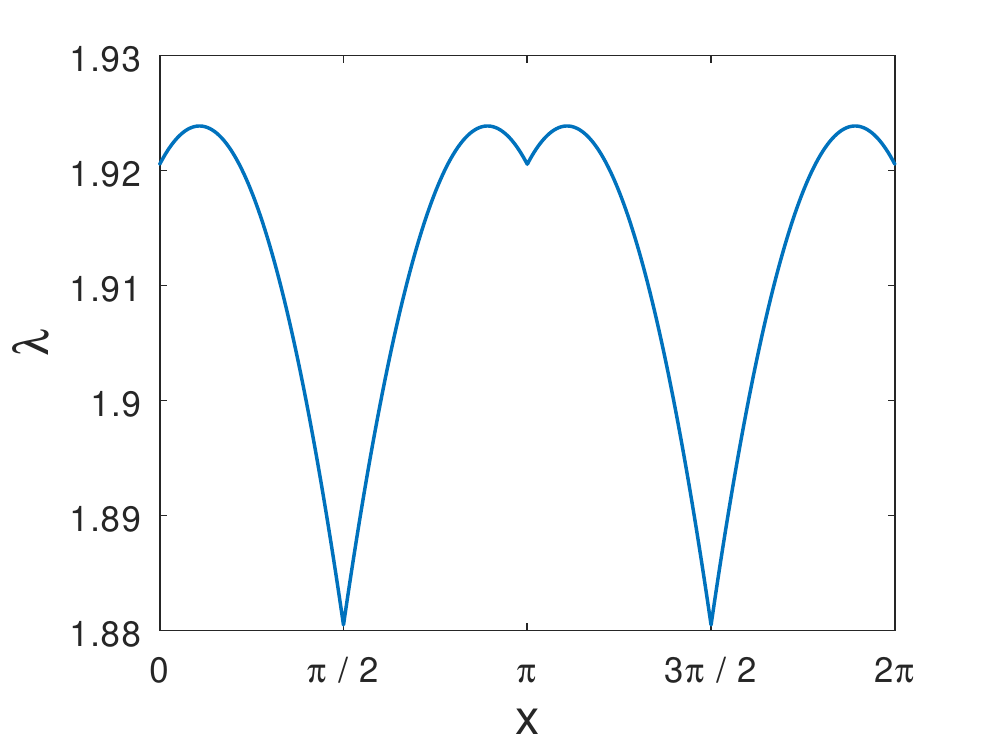}
		\caption{Maximal eigenvalue of the operator in Eq.~\eqref{eq:assop} for 10000 different values of the MUB parameter $x \in [0, 2\pi)$, each maximised over $\sigma \in S_4$.}\label{fig:ass}
	\end{figure} 
	
\end{proof}

From the above examples, it is clear that in general the set MUM is strictly larger than the set MUB$_\text{ext}$, which is in turn strictly larger than the set MUB.
Finally, let us consider a direct sum of a pair of MUBs with a pair of MUMs that cannot be written as a direct sum of MUBs. This gives a pair of measurements that is not a direct sum of MUBs, but that can be mapped to MUBs via a completely positive unital map. This shows that the set MUB$_\text{ext}$ is strictly larger than the set MUB$_\oplus$.

All the above considerations lead to the following classification of MUMs (also see Fig.~\ref{fig:class}): 

\begin{equation}
	\text{MUB} \subsetneq \text{MUB}_\oplus \subsetneq \text{MUB}_\text{ext} \subsetneq \text{MUM},
\end{equation}
\begin{figure}[h!]
	\centering
	\includegraphics[width=0.6\columnwidth]{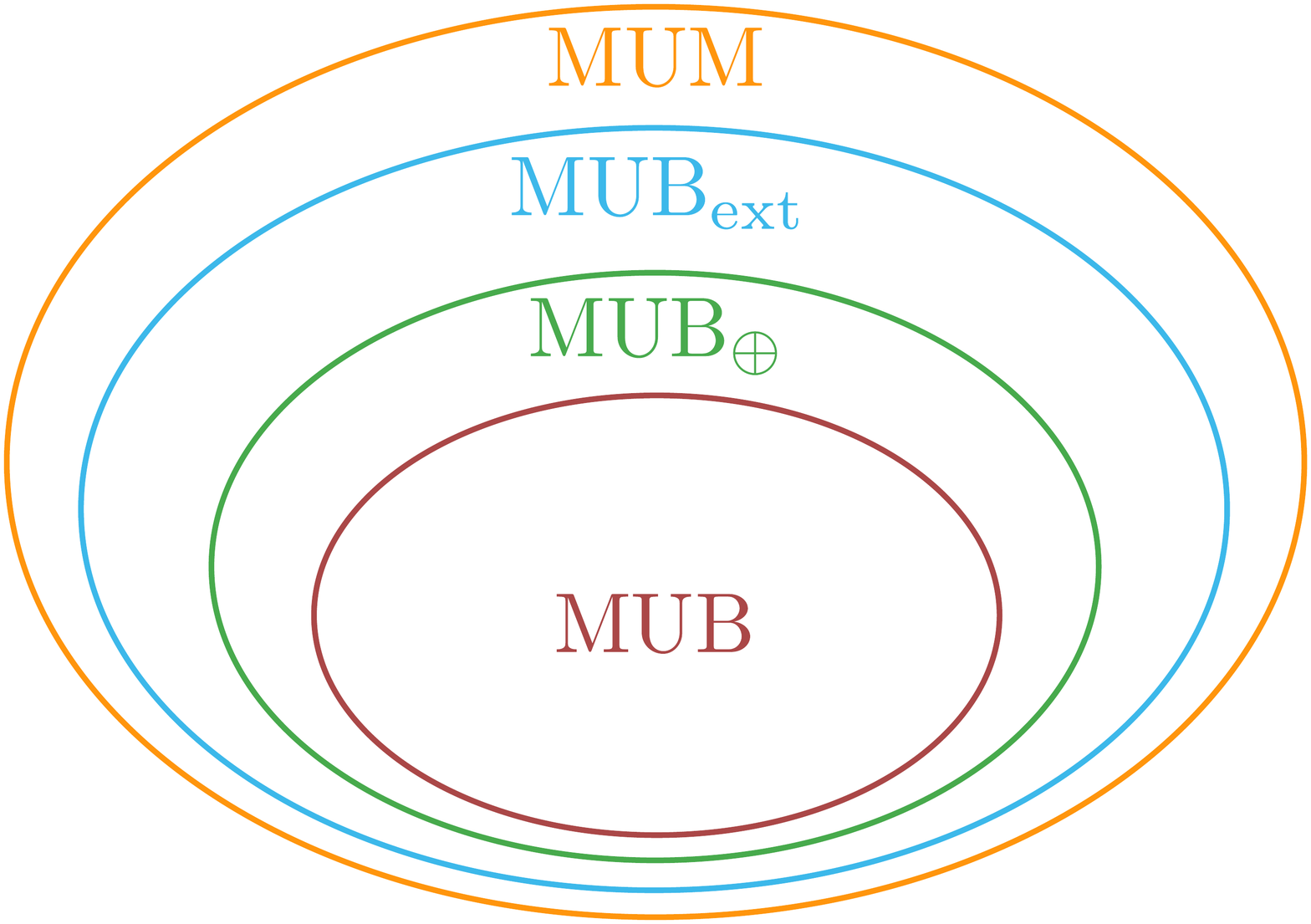}
	\caption{Classification of mutually unbiased measurements. Note that these classes partially collapse ($\textnormal{MUB}_{\oplus} = \textnormal{MUM}$) for outcome numbers 2 and 3, but they are strictly different for outcome numbers 4 and 5.}\label{fig:class}
\end{figure}
where in general all the inclusions are strict.

\subsection{Incompatibility robustness}\label{AppIncomp}
In Ref.~\cite{Designolle} five measures of incompatibility robustness have been considered. In particular these measures have been evaluated for MUBs in arbitrary dimension $d$. In this appendix we show that for measurements acting on a finite-dimensional Hilbert space the algebraic relations given in Eq.~\eqref{sandwichapp}  are sufficient to guarantee precisely the same value as for MUBs for all five measures. This implies that for the generalised incompatibility robustness MUMs are among the most incompatible pairs of measurements with $d$-outcomes.

Explaining the concept of incompatibility robustness is beyond the scope of this work, so we restrict ourselves to giving explicit constructions and arguments necessary to prove our claims. We are extensively using the language and notation introduced in Ref.~\cite{Designolle}. To prove that a pair of measurements gives a certain value of incompatibility robustness we must provide an explicit construction of a joint measurement and argue that no higher value of incompatibility robustness is possible. In the rest of this appendix we first specify a construction and finally argue that it saturates the upper bounds derived in Ref.~\cite{Designolle}. In a nutshell, the reason why the constructions proposed for MUBs work for measurements satisfying the algebraic relations given in Eq.~\eqref{sandwichapp} is the fact that in these constructions at any point we only consider a single measurement operator from the first measurement and a single operator from the second measurement. We have shown before that in this case operators satisfying the algebraic relations given in Eq.~\eqref{sandwichapp} are indistinguishable from MUBs.

What turns out to be crucial is that the MUM conditions completely determine the spectrum of the operator $P_{a} + Q_{b}$. In Appendix~\ref{app:canonical-form} we showed that if a pair of $d$-outcome MUMs exists in a Hilbert space $\cH$, then we must have $\cH \simeq \cH' \otimes \bC^{d}$ for some other Hilbert space $\cH'$ and there exists a unitary $U : \cH \to \cH' \otimes \bC^{d}$ such that
\begin{equation}
	U P_{a} U\hc = \I \otimes \ketbraq{a} \nbox{and} U Q_{1} U\hc = \I \otimes \ketbraq{v},
\end{equation}
where $\{ \ket{j} \}_{j = 1}^{d}$ is an orthonormal basis on $\bC^{d}$ and $\ket{v} = \frac{1}{\sqrt{d}} \sum_{j = 1}^{d} \ket{j}$. Clearly, computing the spectrum of $P_{a} + Q_{1}$ reduces to computing the spectrum of a rank-two operator acting on $\bC^{d}$. Finally, note that the choice of $Q_{1}$ as the projector that should take a particularly simple form after applying the unitary was arbitrary and we can easily find a unitary which achieves the same goal for any $Q_{b}$. From this we conclude that
\begin{equation}
	\label{Pa+Qb-spec}
	\spec( P_{a} + Q_{b} ) =
	\begin{cases}
		\Big\{ 1 + \frac{1}{\sqrt{2}}, 1 - \frac{1}{\sqrt{2}} \Big\} \quad \mbox{for} \quad d = 2,\\
		\Big\{ 1 + \frac{1}{\sqrt{d}}, 1 - \frac{1}{\sqrt{d}}, 0 \Big\} \quad \mbox{for} \quad d \geq 3.
	\end{cases}
\end{equation}

\textbf{Construction:} Let $\{ P_{a} \}_{a = 1}^{d}$ and $\{ Q_{b} \}_{b = 1}^{d}$ be a pair of $d$-outcome measurements acting on $\bC^{n}$ which satisfies the algebraic relations given in Eq.~\eqref{sandwichapp} (in particular, they must be projective). Consider the joint measurement given by operators $\{ G_{ab} \}_{a, b = 1}^{d}$ defined as
\begin{equation}
	G_{ab} = \frac{1}{2 (d + \sqrt{d})} \big( P_{a} + Q_{b} + \sqrt{d} \{P_{a}, Q_{b}\} \big).
\end{equation}
To verify that these are positive semidefinite note that we can explicitly compute their spectrum. Since $\{ P_{a}, Q_{b} \} = ( P_{a} + Q_{b} )^{2} - ( P_{a} + Q_{b} )$ and we already know the spectrum of $P_{a} + Q_{b}$ we immediately conclude that
\begin{equation}
	\label{eq:spec-Gab}
	\spec( G_{ab} ) = \Big\{ \frac{1}{d}, 0 \Big\}.
\end{equation}
It is easy to verify that they are also normalised, i.e.~$\sum_{ab} G_{ab} = \I$. Computing the marginals gives
\begin{equation}
	\sum_{b} G_{ab} = \frac{1}{2 (d + \sqrt{d})}\left( d P_{a} + \I + 2 \sqrt{d} P_{a} \right) = \frac{1}{2} \Big( 1 + \frac{1}{\sqrt{d} + 1} \Big) P_{a} + \frac{ \I }{ 2 (d + \sqrt{d}) }.
\end{equation}
This implies that the depolarising, random and probabilistic incompatibility robustness $\etad, \etar, \etap$ of these two measurements satisfy
\begin{equation}
	\etad, \etar, \etap \geq \frac{1}{2} \Big( 1 + \frac{1}{\sqrt{d} + 1} \Big).
\end{equation}
Moreover, using the fact that $\I \geq P_{a}$ we obtain
\begin{equation}
	\sum_{b} G_{ab} \geq \frac{ d + 2\sqrt{d} + 1}{ 2(d + \sqrt{d}) } P_{a} = \frac{1}{2} \Big( 1 + \frac{1}{ \sqrt{d} } \Big) P_{a},
\end{equation}
which implies that the generalised incompatibility robustness $\etag$ satisfies
\begin{equation}
	\etag \geq \frac{1}{2} \Big( 1 + \frac{1}{ \sqrt{d} } \Big).
\end{equation}
For the jointly-measurable incompatibility robustness we have to provide an explicit construction of a subnormalised jointly measurable noise. For $d = 2$ we choose
\begin{equation}
	\tilde{H}_{ab} = \frac{ 3 - 2 \sqrt{2} }{2} ( \I - 2 G_{ab} ),
\end{equation}
whose positivity follows immediately from Eq.~\eqref{eq:spec-Gab}. This immediately implies that
\begin{equation}
	\etajm \geq 2 ( \sqrt{2} - 1 ).
\end{equation}
For $d \geq 3$ let
\begin{equation}
	\tilde{H}_{ab} = \frac{1 - \gamma_{d}}{d (d - 2)} \Big[ \I + \frac{d}{d-1} ( P_{a} + Q_{b} ) ( P_{a} + Q_{b} - 2 \I ) \Big],
\end{equation}
where $\gamma_{d} = \frac{1}{2}( 1 + \frac{1}{\sqrt{d}} )$. To see that $\tilde{H}_{ab} \geq 0$ it is sufficient to observe that its spectrum can be computed from the spectrum of $P_{a} + Q_{b}$. A direct computation gives
\begin{equation}
	\sum_{b} \tilde{H}_{ab} = \frac{1 - \gamma_{d}}{d - 1} ( \I - P_{a} ),
\end{equation}
which implies that
\begin{equation}
	\etajm \geq \frac{1}{2} \Big( 1 + \frac{1}{ \sqrt{d} } \Big).
\end{equation}

\textbf{Upper bounds:} The upper bounds we use are collected in Table I in terms of quantities specified in Eqs.~(18) and (19) of Ref.~\cite{Designolle}, so let us first compute these quantities. Since the measurements are projective and since all the measurement operators must have the same trace we immediately deduce that $f = 2$ and $g^{d} = g^{r} = g^{p} = \frac{2}{d}$. The last two quantitities $\lambda$ and $\mu$ can be read off directly from the spectra given in Eq.~\eqref{Pa+Qb-spec}: 
\begin{align}
	\lambda &= \max_{a, b} \lmax ( P_{a} + Q_{b} ) = 1 + \frac{1}{\sqrt{d}},\\
	\gjm &= \min_{a, b} \lmin ( P_{a} + Q_{b} ) =
	\begin{cases}
		1 - \frac{1}{\sqrt{2}} \quad \mbox{for} \quad d = 2,\\
		0 \quad \mbox{for} \quad d \geq 3,
	\end{cases}
\end{align}
where $\lmax(M)$ and $\lmin(M)$ denote the maximal and minimal eigenvalues of $M$, respectively. Plugging these quantities into the upper bounds stated in Table I of Ref.~\cite{Designolle} shows that the lower bounds derived above are in fact tight.

Finally, let us show that the MUB value for the generalised incompatibility robustness coincides with the lowest value achievable by any pair of $d$-outcome measurements.

\begin{proposition}
	Given an arbitrary pair of POVMs, $\{ A_a \}_{a = 1}^{n_A}$ and $\{ B_b \}_{b = 1}^{n_B}$ acting on a separable Hilbert space, their generalised incompatibility robustness satisfies
	\begin{equation}
		\etag \ge \frac12 \left( 1 + \frac{ \sqrt{n_A} + \sqrt{n_B} + 2 }{ n_A + n_B + \sqrt{n_A} + \sqrt{n_B} } \right).
	\end{equation}
\end{proposition}

\begin{proof}
	First notice that $\etag$ is monotonic under pre-processing of (i.e.~applying a completely positive unital map on) the measurement pair \cite{Designolle}. Therefore, for any pair of POVMs $\{ A_a \}_{a = 1}^{n_A}$ and $\{ B_b \}_{b = 1}^{n_B}$, their generalised robustness is lower bounded by the robustness of the pair of projective measurements $\{ P_a \}_{a = 1}^{n_A}$ and $\{ Q_b \}_{b = 1}^{n_B}$ obtained by the Naimark dilation of the original POVMs. Let us construct a joint measurement for the projective measurements,
	\begin{equation}
		\begin{split}
			G_{ab}  & \left. = \frac{1}{ 2( n_A + n_B + \sqrt{n_A} + \sqrt{n_B} ) } \left[ P_a + Q_b + (\sqrt{n_A} + \sqrt{n_B}) \{ P_a, Q_b \} + n_A P_a Q_b P_a + n_B Q_b P_a Q_b \right] \right. \\
			& \left. = \frac{1}{ 2( n_A + n_B + \sqrt{n_A} + \sqrt{n_B} ) } \left[ (P_a + \sqrt{n_B} P_a Q_b)\hc (P_a + \sqrt{n_B} P_a Q_b) + (Q_b + \sqrt{n_A} Q_b P_a)\hc (Q_b + \sqrt{n_A} Q_b P_a) \right] \ge 0. \right.
		\end{split}
	\end{equation}
	It is easy to see that $\sum_{ab} G_{ab} = \I$, and performing a partial sum over $b$ gives
	\begin{equation}
		\begin{split}
			\sum_b G_{ab} & \left. = \frac{1}{ 2( n_A + n_B + \sqrt{n_A} + \sqrt{n_B} ) } \left( n_B P_a + \I + 2( \sqrt{n_A} + \sqrt{n_B} ) P_a + n_A P_a + n_B \sum_b Q_b P_a Q_b \right) \right. \\
			& \left. \ge \frac{1}{ 2( n_A + n_B + \sqrt{n_A} + \sqrt{n_B} ) } \left[ n_A + n_B + 2( \sqrt{n_A} + \sqrt{n_B} ) + 2 \right] P_a = \frac12 \left( 1 + \frac{ \sqrt{n_A} + \sqrt{n_B} + 2 }{ n_A + n_B + \sqrt{n_A} + \sqrt{n_B} } \right) P_a,
			\right.
		\end{split}
	\end{equation}
	where the inequality follows from $\I \ge P_a$ and $n_B \sum_b Q_b P_a Q_b \ge P_a$. The latter results from the fact that for all vectors $\ket{ \psi }$
	\begin{equation}
		\begin{split}
			\bra{ \psi } n_B \sum_b Q_b P_a Q_b \ket{ \psi } & \left. = n_B \sum_b \bra{ \psi } Q_b P_a P_a Q_b \ket{\psi} = n_B \sum_b \norm[\big]{ P_a Q_b \ket{ \psi } }^2 = \sum_{b'} 1^2 \sum_b \norm[\big]{ P_a Q_b \ket{ \psi } }^2 \right. \\
			& \left. \ge \abs[\Big]{ \sum_b 1 \cdot \norm{ P_a Q_b \ket{ \psi } } }^2 \ge \abs[\Big]{ \norm{ \sum_b P_a Q_b \ket{ \psi } } }^2 = \norm{ P_a \ket{ \psi } }^2 = \bra{ \psi } P_a \ket{ \psi },
			\right.
		\end{split}
	\end{equation}
	where we used the Cauchy--Schwarz and the triangle inequalities. Similarly, it holds that
	\begin{equation}
		\sum_a G_{ab} \ge \frac12 \left( 1 + \frac{ \sqrt{n_A} + \sqrt{n_B} + 2 }{ n_A + n_B + \sqrt{n_A} + \sqrt{n_B} } \right) Q_b,
	\end{equation}
	which concludes the proof.
\end{proof}

\begin{corollary}
	An arbitrary pair of $d$-outcome MUMs is among the most incompatible $d$-outcome measurement pairs under the generalised incompatibility robustness measure.
\end{corollary}


\section{Bell inequalities for SICs}
In this appendix, we elaborate on details concerning the Bell inequalities for SICs. We derive the quantum and no-signaling bounds for the simplified Bell inequality for SICs and then proceed to derive the local and quantum bounds for the (unconstrained) Bell inequality for SICs. Then, we prove that a maximal violation of this inequality certifies that Bob's measurements are OP-SICs.

\subsection{Quantum correlations and no-signaling correlations in marginally constrained Bell scenarios for SICs}\label{AppTheorem4}

In this section we prove tight bounds on the maximal quantum and no-signaling violations of the Bell functional $\Bfunc{R}{SIC}$ defined in Eq.~\eqref{sicBell} under the marginal constraints given in Eq.~\eqref{marg}. Before we proceed let us prove the following technical result.
\begin{proposition}
	\label{prop:SDP}
	Let $\{ B_{u} \}_{u = 1}^{n}$ be a set of $n$ projectors (for $n \geq 2$) acting on a Hilbert space $\cH$ and let $\ket{\psi} \in \cH$ be a normalised vector. Then, if
	\begin{equation}
		\frac{1}{n} \sum_{u = 1}^{n} \bracket{\psi}{ B_{u} }{\psi} = q
	\end{equation}
	for some $q \in [0, 1]$, then
	\begin{equation}
		\sum_{u \neq v} \bracket{\psi}{ \{ B_{u}, B_{v} \} }{\psi} \geq f(q, n) \equiv q n ( qn - 1 ).
	\end{equation}
\end{proposition}
\begin{proof}
	First note that since all the quantities appearing in the problem are invariant under permutations of $\{ B_{u} \}_{u = 1}^{n}$, for the purpose of deriving bounds we can restrict our analysis to the uniform case in which for every $u$ we have
	\begin{equation}
		\bracket{\psi}{B_{u}}{\psi} = q
	\end{equation}
	and for every pair $u \neq v$ we have
	\begin{equation}
		\bracket{\psi}{ \{ B_{u}, B_{v} \} }{\psi} = 2 t
	\end{equation}
	for some real number $t$. Now consider a moment matrix $\Gamma$ defined as $\Gamma_{ij}=\bracket{\psi}{\mathcal{S}_i\mathcal{S}_j}{\psi}$ where the list of monomials reads $\mathcal{S}=\{\I,B_1,B_2,\ldots, B_{n}\}$. Clearly, $\Gamma$ is positive semidefinite but it might have complex entries, so it is more convenient to consider its symmetrised version $\Gamma_{\textnormal{sym}} = ( \Gamma + \Gamma\tran )/2$, which reads
	\begin{equation}
		\Gamma_{\textnormal{sym}} = \begin{pmatrix}
			1 & q & q & q & \ldots & q\\
			q & q & t & t & \ldots & t\\
			q & t & q & t & \ldots & t\\
			q & t & t & q& \ldots & t \\
			\vdots & \vdots & \vdots & \ddots & \ddots & \vdots\\
			q & t & t & t & \ldots & q
		\end{pmatrix}.
	\end{equation}
	To take advantage of the positivity of $\Gamma_{\textnormal{sym}}$ we take advantage of the Schur-complement condition. Consider a real Hermitian block matrix
	\begin{equation}
		X = \begin{pmatrix}
			A & B\\
			B\tran & C
		\end{pmatrix},
	\end{equation}
	where $A > 0$ is positive definite. Then, the Schur-complement condition states that $X \geq 0$ if and only if $C - B\tran A^{-1} B \geq 0$. In our case we take the first row/column against the remaining $n$ rows/columns, which means that $\Gamma_{\textnormal{sym}} \geq 0$ is equivalent to
	\begin{equation}
		C - B\tran A^{-1} B = n ( t - q^{2} ) \ketbra{v}{v} + (q - t) \I,
	\end{equation}
	where $\ket{v} = \frac{1}{\sqrt{n}} \sum_{j = 1}^{n} \ket{j}$  is the uniform vector. The positivity of this operator implies that
	\begin{equation}
		n ( t - q^{2} ) + (q - t) \geq 0,
	\end{equation}
	and therefore
	\begin{equation}
		t \geq \frac{ q ( qn - 1 ) }{n - 1}.
	\end{equation}
	Performing the summation over $u \neq v$ leads to the final statement of the proposition.
\end{proof}

\subsubsection{The quantum bound}
We can without loss of generality assume that the optimal shared state is pure and denote it by $\ket{\psi}$. Similarly, we can without loss of generality take the optimal measurements of both Alice and Bob to be projective and denote the complete sets of projectors by $\{A_{x}^a\}$ and $\{B_y^b\}$ respectively.  Hence, in quantum theory, we have
\begin{equation}
	\mathcal{R}_d^\text{SIC}=\sum_{x_{1} < x_{2}}\bracket{\psi}{A_{x}\otimes \left(B_{x_{1}}-B_{x_{2}}\right)}{\psi},
\end{equation}
where we have introduced the observable $A_{x}=A_{x}^1-A_{x}^2$ for Alice and recall that $x = x_{1} x_{2}$. For simplicity let us define the vectors
\begin{align}
	& \ket{\alpha_{x}}=A_{x}\otimes\I \ket{\psi}\\
	& \ket{\beta_{x}}=\I \otimes\left(B_{x_{1}}-B_{x_{2}}\right) \ket{\psi}.
\end{align}
This allows us to write the Bell functional in a simple form which obeys a straightforward upper bound via the Cauchy--Schwarz inequality
\begin{equation}\label{step}
	\mathcal{R}_d^\text{SIC}=\sum_{x_{1} < x_{2}} \braket{\alpha_{x}}{\beta_{x}}\leq \sum_{x_{1} < x_{2}} |\braket{\alpha_{x}}{\beta_{x}}|\leq  \sum_{x_{1} < x_{2}} \sqrt{\braket{\alpha_{x}}{\alpha_{x}}} \sqrt{\braket{\beta_{x}}{\beta_{x}}}.
\end{equation}
To proceed further, we note that the marginal constraints \eqref{marg} in quantum theory read
\begin{align}\label{qmarg1}
	\forall x: \qquad &\bracket{\psi}{\left(A_{x}^1+A_{x}^2\right)\otimes \I}{\psi}=\frac{2}{d},\\
	\label{qmarg2}
	\forall y: \qquad &\bracket{\psi}{\I\otimes B_{y}}{\psi}=\frac{1}{d}.
\end{align}
In what follows, we will for simplicity write marginal expressions  $\bracket{\psi}{A\otimes \I }{\psi}=\bracket{\psi}{A}{\psi}$ and similarly for marginal expressions on Bob's side. 

One straightforwardly sees that for projective measurements the inner product $\braket{\alpha_{x}}{\alpha_{x}}$ is identical to the left-hand side of the marginal constraint~\eqref{qmarg1}, whereas the inner product $\braket{\beta_{x}}{\beta_{x}}$ features two terms which are determined by the marginal constraint~\eqref{qmarg2}. We therefore have
\begin{align}
	&\braket{\alpha_{x}}{\alpha_{x}}=\frac{2}{d},\\
	& \braket{\beta_{x}}{\beta_{x}}=\frac{2}{d}-\bracket{\psi}{\{B_{x_{1}},B_{x_{2}}\}}{\psi}.
\end{align} 
Inserting this into Eq.~\eqref{step} we obtain
\begin{equation}\label{step1}
	\mathcal{R}_d^\text{SIC}\leq \frac{\sqrt{2}}{d} \sum_{x_{1} < x_{2}}\sqrt{2 -d \bracket{\psi}{\{B_{x_{1}},B_{x_{2}}\}}{\psi}}.
\end{equation}
Concavity of the square root states that for all $q_i \geq 0$ we have $\frac{1}{n} \sum_{i=1}^{n} \sqrt{q_i} \leq \sqrt{ \frac{1}{n} \sum_{i=1}^n q_i}$ with equality if and only if all $q_i$ are equal. Applying this to Eq.~\eqref{step1} we obtain
\begin{equation}\label{stepbound}
	\mathcal{R}_d^\text{SIC} \leq \frac{\sqrt{2}}{d} \sqrt{2 N^2  - N d \sum_{x_{1} < x_{2}}\bracket{\psi}{\{B_{x_{1}},B_{x_{2}}\}}{\psi}},
\end{equation}
where $N = \binom{d^2}{2}$ is the number of terms in the sum. To bound the sum under the square root we use Proposition~\ref{prop:SDP} which implies that
\begin{equation}
	\sum_{x_{1} < x_{2}} \bracket{\psi}{\{B_{x_{1}},B_{x_{2}}\}}{\psi} \geq f( 1/d, d^{2} ) = d ( d - 1 ).
\end{equation}
Plugging this bound into Eq.~\eqref{stepbound} implies
\begin{equation}
	\Bfunc{R}{SIC} \maxval{Q} d (d - 1) \sqrt{ d ( d + 1 ) },
\end{equation}
which is precisely the result stated in the main text.
\subsubsection{The no-signaling bound}
In order to find the no-signaling bound of $\Bfunc{R}{SIC}$, subject to the constraints \eqref{marg}, we first derive its algebraically maximal value and then saturate this value with an explicit no-signaling model.  Consider a specific choice of inputs $x$ and $y$ such that $y=x_{1}$. Since Alice has three possible outcomes and Bob has two possible outcomes, we have  six variables $p_{ab}=p\left(a,b|x,y=x_{1}\right)$. They must obey positivity, normalisation and the marginal constraints of Alice and Bob, i.e.~
\begin{align}\nonumber
	& p_{ab}\geq 0,\\\nonumber
	& \sum_{a,b}p_{ab}=1,\\\nonumber
	& p_{11}+p_{1\perp}+p_{21}+p_{2\perp}=\frac{2}{d},\\
	&  p_{11}+p_{21}+p_{\perp 1}=\frac{1}{d}.
\end{align}
Note that the probabilities of winning and losing correspond to $p_{11}$ and $p_{21}$, respectively. The maximal contribution to $\Bfunc{R}{SIC}$ is made by maximising $p_{11}-p_{21}$ under the above constraints. This means $p_{11}=1/d$ and $p_{21}=0$. An analogous argument for the inputs $x$ and $y = x_{2}$ leads to $p_{11}=0$ and $p_{21}=1/d$.  Thus,
\begin{equation}
	\Bfunc{R}{SIC}=\sum_{x_{1} < x_{2}}\left(\frac{1}{d}+\frac{1}{d}\right)=d\left(d^2-1\right).
\end{equation}
To see that this algebraic bound can be saturated with a no-signaling model, let us again consider the event in which the inputs are $x$ and $y=x_{1}$. Then we choose $p=(p_{ab})_{ab}$ as
\begin{equation}
	p=\frac{1}{d}
	\begin{pmatrix}
		1 & 0\\
		0 & 1\\
		0 & d-2
	\end{pmatrix}.
\end{equation}
When the inputs are  $x$ and $y=x_{2}$ we instead choose 
\begin{equation}
	p=\frac{1}{d}
	\begin{pmatrix}
		0 & 1\\
		1 & 0\\
		0 & d-2
	\end{pmatrix}.
\end{equation}
For the input combinations that do not contribute to $\Bfunc{R}{SIC}$, i.e.~when Alice receives $x$ and Bob receives $y\notin \{x_1,x_2\}$, we choose
\begin{equation}
	p=\frac{1}{2d}
	\begin{pmatrix}
		1 & 1\\
		1 & 1\\
		0 & 2(d-2)
	\end{pmatrix}.
\end{equation}
It is straightforward to check that these distributions satisfy the marginal constraints~\eqref{marg}, the no-signaling constraints and saturate the algebraic bound.
\subsection{Bell inequalities for SICs (proof of Theorem~\ref{sictheorem})}\label{AppTheorem5}
We first prove the quantum bound and then proceed to prove the classical bound.

\subsubsection{Quantum bound}
\label{app:S-SIC-q-val}

We can recycle our earlier maximisation of $\Bfunc{R}{SIC}$, but this time without the marginal constraints \eqref{marg}. Following the calculation in Eq.~\eqref{step} one arrives at 
\begin{equation}
	\label{eq:SIC-CS-ineq}
	\Bfunc{R}{SIC}\leq \sum_{x_{1} < x_{2}} \sqrt{\bracket{\psi}{A_{x}^1+A_{x}^2}{\psi}}  \sqrt{\bracket{\psi}{B_{x_{1}}+B_{x_{2}}}{\psi}-\bracket{\psi}{\{B_{x_{1}},B_{x_{2}}\}}{\psi}}.
\end{equation}
Using the Cauchy--Schwarz inequality, we obtain the upper bound
\begin{equation}\label{stepstep}
	\Bfunc{R}{SIC}\leq \sqrt{s}\Big(r-\sum_{x_{1} < x_{2}}\bracket{\psi}{\{B_{x_{1}},B_{x_{2}}\}}{\psi}\Big)^{1/2},
\end{equation}
where we have defined
\begin{equation}
	\label{sr}
	\begin{aligned}
		& s \equiv \sum_{x_{1} < x_{2}}\bracket{\psi}{A_{x}^1+A_{x}^2}{\psi},\\
		& r \equiv \sum_{x_{1} < x_{2}}\bracket{\psi}{B_{x_{1}}+B_{x_{2}}}{\psi}.
	\end{aligned}
\end{equation}
An upper bound on the second factor in Eq.~\eqref{stepstep} can be obtained using Proposition~\ref{prop:SDP}. Since
\begin{equation}
	r = \sum_{x_{1} < x_{2}}\bracket{\psi}{\left(B_{x_{1}}+B_{x_{2}}\right)}{\psi} = ( d^{2} - 1 ) \sum_{x_{1}} \bracket{\psi}{ B_{x_{1}} }{\psi},
\end{equation}
we arrive at
\begin{equation}
	\Bfunc{R}{SIC} \leq \sqrt{s} \sqrt{ r - f \bigg( \frac{r}{d^{2}(d^{2} - 1) }, d^{2} \bigg) } = \frac{ \sqrt{ s r \big[ d^{2} (d^{2} - 1) - r \big] } }{d^{2} - 1}.
\end{equation}
Returning to our Bell functional of interest, namely $\Bfunc{S}{SIC}$,  we have
\begin{equation}\label{witnessbound}
	\Bfunc{S}{SIC}\leq \frac{ \sqrt{ s r \big[ d^{2} (d^{2} - 1) - r \big] } }{d^{2} - 1} - \alpha_d s-\frac{\beta_d}{d^2-1} r \equiv F(r, s),
\end{equation}
where the coefficients are specified by
\begin{align}
	& \alpha_d=\frac{1-\delta_{d,2}}{2}\sqrt{\frac{d}{d+1}} & \beta_d=\frac{d-2}{2}\sqrt{d(d+1)}.
\end{align}
Our goal now is to find the maximise $F(r, s)$ over $r \in [0, d^{2} ( d^{2} - 1 )]$ and $s \in [0, d^{2} ( d^{2} - 1 )/2]$. For $d = 2$ the marginal terms vanish and we immediately see that the optimal values are given by $r_{\textnormal{opt}} = s_{\textnormal{opt}} = 6$ and $F(r_{\textnormal{opt}}, s_{\textnormal{opt}}) = 2 \sqrt{6}$. For $d \geq 3$ we relax the problem and maximise over $r, s \geq 0$. Standard differentiation techniques yield
\begin{equation}
	\label{eq:sopt-ropt}
	\left\{\frac{\partial F}{\partial r}\stackrel{!}{=}0 \quad\text{and} \quad  \frac{\partial F}{\partial s}\stackrel{!}{=}0\right\} \Rightarrow r_\text{opt}=s_\text{opt}=d\left(d^2-1\right),
\end{equation}
which falls inside the optimisation region. It is easy to check that
\begin{equation}
	F( r_{\textnormal{opt}}, s_{\textnormal{opt}} ) = \frac{d}{2} \sqrt{d\left(d+1\right)}.
\end{equation}
To check that no higher value can be achieved at the boundary note that whenever either $r = 0$ or $s = 0$ we necessarily have $F(r, s) \leq 0$. Plugging this value into Eq.~\eqref{witnessbound} combined with the bound derived for $d = 2$ gives
\begin{equation}
	\label{eq:S-SIC-q-val}
	\Bfunc{S}{SIC} \maxval{Q} \frac{d+2\delta_{d,2}}{2}\sqrt{d\left(d+1\right)}.
\end{equation}
As stated in the main text this bound can be saturated if there exists a set of SIC projectors in dimension $d$.
\subsubsection{The local bound}
Here we derive the local bound of $\Bfunc{S}{SIC}$. To this end, it is sufficient to consider all deterministic strategies of Alice and Bob. We focus on the deterministic strategy of Bob. We can represent his strategy as a list $\vec{b}\in\{1,\perp\}^{d^2}$, where the $y$-th element encodes the output associated to measurement setting $y$. Hence, Bob has $2^{(d^2)}$ possible deterministic strategies. However, due to the symmetries of the Bell functional, the maximal local value of $\Bfunc{S}{SIC}$ achievable for a given strategy of Bob is the same as that achievable when his strategy $\vec{b}$ is permuted into the strategy $\vec{b}'=(1,\ldots,1,\perp,\ldots,\perp)$. In other words, only the number of ``$1$s`` in $\vec{b}$ is relevant for the maximal value of $\Bfunc{S}{SIC}$. This can be seen from the fact that the Bell functional is invariant under a permutation of Bob's $d^2$ inputs (provided analogous permutations of Alice's input tuple $x$ and output $a$). Thus, we need only to consider strategies of Bob in which the $N$ first elements of $\vec{b}$ are ones, and the remaining $d^2-N$ elements are $\perp$. Finding the maximum over $N \in \{0, 1, \ldots, d^{2} \}$ gives the tight local bound.

This allows us to write the Bell functional in a local model as 
\begin{equation}
	\Bfunc{S}{SIC}=\sum_{x_{1} < x_{2}}\left(\bracket{\psi}{A_{x}}{\psi}\left[\delta(x_1\leq N)-\delta(x_2\leq N)\right]-\frac{\beta_d}{d^2-1}\left[\delta(x_1\leq N)+\delta(x_2\leq N)\right]-\alpha_d\bracket{\psi}{A_{x}^1+A_{x}^2}{\psi}\right),
\end{equation}
where $\delta(j\leq k)$ is equal to one if $j\leq k$ and otherwise equal to zero. Let us now define the following sets
\begin{equation}
	T_k^{(N)}\equiv \{(x_1,x_2)\in \text{Pairs}(d^2)| \delta(x_1\leq N)+\delta(x_2\leq N)=k\},
\end{equation}
for $k\in\{0,1,2\}$ and $N\in\{0, 1, \ldots,d^2\}$. Clearly, the sets $\{T_0^{(N)},T_1^{(N)},T_2^{(N)}\}$ constitute a partitioning of the set $\text{Pairs}(d^2)$. Hence
\begin{equation}
	\Bfunc{S}{SIC}=\sum_{k=0}^{2} \sum_{(x_1,x_2)\in T_k^{(N)}}\left(\bracket{\psi}{A_{x}}{\psi}\left[\delta(x_1\leq N)-\delta(x_2\leq N)\right]-\frac{k\beta_d}{d^2-1}-\alpha_d\bracket{\psi}{A_{x}^1+A_{x}^2}{\psi}\right).
\end{equation}
Notice that if $(x_1,x_2)\in T_1^{(N)}$, then it must hold that $\delta(x_1\leq N)-\delta(x_2\leq N)=1$. Also, when $(x_1,x_2)\in T_0^{(N)}$ or $(x_1,x_2)\in T_2^{(N)}$, then $\delta(x_1\leq N)-\delta(x_2\leq N)=0$. Thus, we can further simplify the above expression to
\begin{equation}
	\Bfunc{S}{SIC}=\sum_{k=0}^{2} \sum_{(x_1,x_2)\in T_k^{(N)}}\left(\bracket{\psi}{A_{x}}{\psi}\delta_{k,1}-\frac{k\beta_d}{d^2-1}-\alpha_d\bracket{\psi}{A_{x}^1+A_{x}^2}{\psi}\right).
\end{equation}
Our goal now is to optimise over the measurement operators of Alice, which only involves the first and third term, and let us consider the three values of $k$ separately. If $k = 0$ or $k = 2$ this contribution is never positive, so it is optimal to choose $A^1_{x}=A^2_{x}=0$ and $A^{\perp}_{x}=\I$. If $k = 1$ it is beneficial to choose $A^1_{x}=\I$ and $A^2_{x}=A^{\perp}_{x}=0$ because $1 - \alpha_{d} > 0$. Plugging in the optimal choice of measurements of Alice gives
\begin{equation}
	\Bfunc{S}{SIC}=|T_1^{(N)}|\left(1-\alpha_d-\frac{\beta_d}{d^2-1}\right)-|T_2^{(N)}|\frac{2\beta_d}{d^2-1}.
\end{equation}
The size of $T_1^{(N)}$ and $T_2^{(N)}$ are easily determined as follows. The number of ways of choosing a positive integer no larger than $N$ is $N$, and the number of ways of choosing a positive integer larger than $N$ but no larger than $d^2$ is $d^2-N$. Hence, $|T_1^{(N)}|=N(d^2-N)$. Similarly, the number of ways of choosing two distinct positive integers no larger than $N$ is $N(N-1)/2$. Hence, $|T_2^{(N)}|=N(N-1)/2$. This leaves us with
\begin{equation}\label{cl}
	\begin{aligned}
		\Bfunc{S}{SIC} &= - N^{2} ( 1 - \alpha_{d} ) + N \big[ d^{2} ( 1 - \alpha_{d} ) - \beta_{d} \big]\\
		&= \frac{N}{2\sqrt{d+1}}\left(d\sqrt{d}-2d^2\sqrt{d}+\sqrt{d+1}\left(2d^2-2N\right)+\sqrt{d}\left(N+2\right)\right),
	\end{aligned}
\end{equation}
where we have taken $d\geq 3$. Now we must only find the optimal value of $N\in\{0,1,\ldots, d^2\}$. To this end, we differentiate the right-hand side with respect to $N$ and find that the maximum is achieved for
\begin{equation}
	N = \frac{\sqrt{d}\left(2+d-2d^2+2d\sqrt{d(d+1)}\right)}{4\sqrt{d+1}-2\sqrt{d}}.
\end{equation}
In general, this is evidently not an integer. Hence, to find the optimal integer choice of $N$, we first show that $N\leq d$. Simple manipulations allow us to write the statement $N \leq d$ as
\begin{equation}
	2d^2-3d-2\geq 2(d-2)\sqrt{d(d+1)}.
\end{equation}
Squaring both sides and simplifying reduces this to $(d-2)^2\geq 0$, which is evidently true. Next, we show that $N\geq d-1$. Again, simple manipulations allow us to write the statement $N\geq d-1$ as
\begin{equation}
	2d^2-3d\leq (2d^2-4d+4)\sqrt{\frac{d+1}{d}}.
\end{equation}
Squaring both sides, multiplying by $d$ and simplifying leads to
\begin{equation}
	0\leq 7d^3-16d+16,
\end{equation}
which is true for all $d \geq 0$. Thus we conclude that $d-1\leq N\leq d$. To show that $N = d$ constitutes a better solution that $N = d - 1$ we show that $\Bfunc{S}{SIC}(N=d)-\Bfunc{S}{SIC}(N=d-1)\geq 0$. Thanks to Eq.~\eqref{cl} this reduces to showing that
\begin{equation}
	3d^3+2d^2-13d+4\geq 0,
\end{equation}
which is true for all $d \geq 2$. Thus, we conclude that the optimal choice is $N=d$. Inserting $N=d$ in Eq.~\eqref{cl} returns the local bound in Eq.~\eqref{SIClocal}. 

Notably, the case of $d=2$ can be obtained by analogously following the above procedure. However, the corresponding Bell scenario is of sufficiently small scale so that the classical bound is arguably even easier obtained by brute-force consideration of all deterministic strategies.

\subsection{Device-independent certification for SICs (proof of Theorem~\ref{thmDIsic})}\label{AppDIsic}
In this appendix we provide a proof of Theorem~\ref{thmDIsic}, which essentially reduces to deriving explicit conditions under which the argument presented in the first part of this appendix is tight and combining them to produce the desired form. Note that the proof technique is quite similar to the one used in Appendix~\ref{app:meas-Bob}.

In Appendix~\ref{app:S-SIC-q-val} we have shown how to derive a tight bound on the maximal quantum value of the Bell functional $\Bfunc{S}{SIC}$ (see Eq.~\eqref{eq:S-SIC-q-val}). In the first step we applied the Cauchy--Schwarz inequality to obtain Eq.~\eqref{eq:SIC-CS-ineq}. This can only be tight if the vectors are aligned, i.e.~for every $x$ we have
\begin{equation}
	\label{eq:cross-mux}
	A_{x} \ket{\psi} = \mu_{x} (B_{x_{1}} - B_{x_{2}}) \ket{\psi}
\end{equation}
for some positive real number $\mu_{x}$. Moreover, here we have used the fact that for all $y$ we have $\bracket{\psi}{B_{y}^{2}}{\psi} \leq \bracket{\psi}{B_{y}}{\psi}$. If the marginal state of Bob is full-rank, then this bound is tight only when all the measurement operators $B_{y}$ are projective. The next use of Cauchy--Schwarz inequality, which leads to Eq.~\eqref{stepstep}, can only give a tight result if for all $x$ we have
\begin{equation}
	\label{eq:cross-v}
	\bracket{\psi}{ A_{x}^1+A_{x}^2 }{\psi} = \nu \big( \bracket{\psi}{ B_{x_{1}}+B_{x_{2}} }{\psi}-\bracket{\psi}{\{B_{x_{1}},B_{x_{2}}\}}{\psi} \big)
\end{equation}
for some fixed real positive number $\nu$. In fact, this number can be computed by summing over $x$. Then, if we recall that the maximal value can only be achieved when the quantities $r$ and $s$ defined in Eq.~\eqref{sr} take the values given in Eq.~\eqref{eq:sopt-ropt} we immediately conclude that $\nu = (d + 1)/d$. However, it is clear by comparing Eq.~\eqref{eq:cross-mux} and Eq.~\eqref{eq:cross-v} that $\mu_{x}^{2} = \nu$, which implies that for all $x$ we have
\begin{equation}
	\label{eq:cross-relation}
	A_{x} \ket{\psi} = \sqrt{ \frac{d + 1}{d} } (B_{x_{1}} - B_{x_{2}}) \ket{\psi}.
\end{equation}
Since we assume that the measurements of Alice are projective, we have $(A_{x})^{3} = A_{x}$. Applying this to Eq.~\eqref{eq:cross-relation}, tracing out Alice's system and right-multiplying by the marginal of Bob, we conclude that for all $x_{1} < x_{2}$ we have
\begin{equation}
	\label{eq:operator-relation}
	B_{x_{1}} - B_{x_{2}} - (d + 1) (B_{x_{1}} B_{x_{2}} B_{x_{1}} - B_{x_{2}} B_{x_{1}} B_{x_{2}} ) = 0.
\end{equation}
In the following we will use $x_{1}$ and $x_{2}$ as summation indices for double sums over $x_{1} < x_{2}$ and $y$ for single sums over $[d^{2}]$.
The optimality condition derived in Eq.~\eqref{eq:sopt-ropt} implies that
\begin{equation}
	\sum_{y} \bracket{\psi}{ B_{y} }{\psi} = \frac{1}{d^{2} - 1} \sum_{x_{1} < x_{2}} \bracket{\psi}{ B_{x_{1}}+B_{x_{2}} }{\psi} = d
\end{equation}
and
\begin{equation}
	\sum_{x_{1} < x_{2}}\bracket{\psi}{\{ B_{x_{1}}, B_{x_{2}} \}}{\psi} = d (d - 1).
\end{equation}
Define
\begin{equation}
	K \equiv \sum_{y} B_{y}
\end{equation}
and note that
\begin{equation}
	\Tr (K \rho_{B}) = d.
\end{equation}
Moreover, since
\begin{equation}
	K^{2} = \sum_{y} B_{y} + \sum_{x_{1} < x_{2}} \{ B_{x_{1}}, B_{x_{2}} \},
\end{equation}
we have
\begin{equation}
	\Tr (K^{2} \rho_{B}) = d + d (d - 1) = d^{2}.
\end{equation}
This means that the Cauchy--Schwarz
\begin{equation}
	\abs{ \Tr (X\hc Y) }^{2} \leq \Tr( X\hc X ) \Tr( Y\hc Y )
\end{equation}
for $X = K \rho_{B}^{1/2}$ and $Y = \rho_{B}^{1/2}$ holds with equality, which implies that
\begin{equation}
	K \rho_{B}^{1/2} = z \rho_{B}^{1/2}
\end{equation}
for some complex number $z$ such that $\abs{z} = d$. Right-multiplying by $\rho_{B}^{-1/2}$ (again we take advantage of the full-rank assumption) leads to $K = z \, \I$. Now it is clear that $z$ must in fact be real and positive and so $z = d$ and $K = d \, \I$. This is the first result of Theorem~\ref{thmDIsic}.

To prove the second part we go back to Eq.~\eqref{eq:operator-relation}. Since the operators $\{ B_{y} \}_{y = 1}^{d^{2}}$ are projectors, we can use Jordan's lemma which states that any pair of projectors can be simultaneously block-diagonalised such that the blocks are of size either $1 \times 1$ or $2 \times 2$. Moreover, it is known that the non-trivial blocks (up to a unitary rotation) form a 1-parameter family. By applying the relation given in Eq.~\eqref{eq:operator-relation} to a particular pair of projectors $B_{x_{1}}$ and $B_{x_{2}}$ for $x_{1} < x_{2}$, we conclude that the Hilbert space of Bob decomposes as $\cH_{B} \simeq \cH_{B_{1}} \oplus \cH_{B_{2}} \oplus \cH_{B_{3}}$ and the operators read
\begin{align}
	\label{eq:Bu}
	B_{x_{1}} &= ( P \otimes \I ) \oplus 0 \oplus \I,\\
	\label{eq:Bv}
	B_{x_{2}} &= ( Q \otimes \I ) \oplus 0 \oplus \I,
\end{align}
where
\begin{align}
	P &= \frac{1}{2} ( \I + \cos \theta \, Z + \sin \theta \, X ),\\
	Q &= \frac{1}{2} ( \I + \cos \theta \, Z - \sin \theta \, X )
\end{align}
for $\cos \theta = 1/\sqrt{d + 1}$, $\sin \theta = \sqrt{d/(d + 1)}$ and $X$ and $Z$ are the $2 \times 2$ Pauli matrices. Note that at this point we know nothing about the dimensions of $\cH_{B_{1}}, \cH_{B_{2}}$ and $\cH_{B_{3}}$, but we clearly see that $B_{x_{1}}$ and $B_{x_{2}}$ are isomorphic (in particular, either $\Tr B_{x_{1}}$ and $\Tr B_{x_{2}}$ are both finite and equal to each other or they are both infinite). Moreover, note that for any pair $x_{1} < x_{2}$ the space decomposes into three subspaces, but these subspaces depend on the specific pair, i.e.~the subspace $\cH_{B_{1}}$ for $(x_{1}, x_{2}) = (0, 1)$ is different from the subspace $\cH_{B_{1}}$ for $(x_{1}, x_{2}) = (0, 2)$.

The final goal is to show that for all $x_{1} < x_{2}$ we have $\dim \cH_{B_{3}} = 0$, i.e.~that the projectors $B_{x_{1}}$ and $B_{x_{2}}$ are in fact of the form
\begin{equation}
	\label{eq:Bx1-Bx2-final}
	\begin{aligned}
		B_{x_{1}} &= ( P \otimes \I ) \oplus 0,\\
		B_{x_{2}} &= ( Q \otimes \I ) \oplus 0.
	\end{aligned}
\end{equation}
This immediately implies that
\begin{equation}
	B_{x_{1}} = (d + 1) B_{x_{1}} B_{x_{2}} B_{x_{1}} \nbox{and} B_{x_{2}} = (d + 1) B_{x_{2}} B_{x_{1}} B_{x_{2}},
\end{equation}
which is precisely the second result of Theorem~\ref{thmDIsic}. Let us first show that this result holds when the expectation values of the projectors $B_{y}$ and the anticommutators $\{ B_{x_{1}}, B_{x_{2}} \}$ are distributed uniformly. In the second step we show that this conclusion holds even without the uniformity assumption.

Let us for now assume that for all $y$ we have
\begin{equation}
	\Tr( B_{y} \rho_{B} ) = \frac{1}{d}
\end{equation}
and for all $x_{1} < x_{2}$ we have
\begin{equation}
	\Tr( \{ B_{x_{1}}, B_{x_{2}} \} \rho_{B} ) = \frac{2}{d (d + 1)}.
\end{equation}
Then, we immediately see that
\begin{equation}
	\Tr \big[ ( B_{x_{1}} - B_{x_{2}} )^{2} \rho_{B} \big] = \Tr( B_{x_{1}} + B_{x_{2}} - \{ B_{x_{1}}, B_{x_{2}} \} \rho_{B} ) = \frac{2}{d + 1}.
\end{equation}
Moreover, a direct calculation shows
\begin{equation}
	( B_{x_{1}} - B_{x_{2}} )^{2} = (\sin \theta X)^{2} \otimes \I \oplus 0 \oplus 0 = \frac{d}{d + 1} \I \otimes \I \oplus 0 \oplus 0.
\end{equation}
Let $\Pi$ be a projector on $\cH_{B_{1}}$:
\begin{equation}
	\Pi = \I \otimes \I \oplus 0 \oplus 0.
\end{equation}
Then, we have
\begin{equation}
	\label{eq:Pi-weight}
	\Tr( \Pi \rho_{B} ) = \frac{2}{d}.
\end{equation}
On the other hand the fact that $\Tr( B_{x_{1}} \rho_{B} ) = \Tr( B_{x_{2}} \rho_{B} )$ implies that
\begin{equation}
	\Tr( X \otimes \I \oplus 0 \oplus 0 \, \rho_{B} ) = 0.
\end{equation}
This allows us to rewrite $\Tr( B_{y} \rho_{B} ) = 1/d$ as
\begin{equation}
	\label{eq:Z+0+Id}
	\Tr \Big( \frac{1}{2} \cos \theta \, Z \otimes \I \oplus 0 \oplus \I \, \rho_{B} \Big) = 0.
\end{equation}
Let us now show that
\begin{equation}
	\label{eq:Z-vanishes}
	\Tr( Z \otimes \I \oplus 0 \oplus 0 \, \rho_{B} ) = 0.
\end{equation}
Define
\begin{equation}
	\sigma_{AB} \equiv \frac{d}{2} ( \I \otimes \Pi ) \ketbraq{\psi} ( \I \otimes \Pi ),
\end{equation}
which is a normalised state because $\Tr \sigma_{AB} = \frac{d}{2} \Tr( \Pi \rho_{B} ) = 1$. Now we take advantage of the fact that if we have two Hermitian operators $X, Y$ satisfying $E^{2} = F^{2} = \I$ and $\{E, F\} = 0$, then on any normalised state $\tau$ the expectation values must satisfy (see Ref.~\cite{Kaniewski} for an elementary proof)
\begin{equation}
	\Tr( E \tau )^{2} + \Tr( F \tau )^{2} \leq 1.
\end{equation}
In our case we set
\begin{align}
	E &\equiv A_{x} \otimes \big[ ( X \otimes \I ) \oplus 0 \oplus 0 \big],\\
	F &\equiv \I \otimes \big[ ( Z \otimes \I ) \oplus 0 \oplus 0 \big]
\end{align}
and $\tau = \sigma_{AB}$. To verify that $\Tr( E \tau ) = 1$ note that
\begin{align}
	\Tr \Big( A_{x} &\otimes \big[ ( X \otimes \I ) \oplus 0 \oplus 0 \big] \, \sigma_{AB} \Big) = \frac{d}{2} \bracket{\psi}{ A_{x} \otimes \big[ ( X \otimes \I ) \oplus 0 \oplus 0 \big] }{\psi}\\
	&= \frac{d}{2 \sin \theta} \bracket{\psi}{ A_{x} \otimes ( B_{x_{1}} - B_{x_{2}} ) }{\psi} = \frac{d}{2 \sin \theta} \sqrt{ \frac{d + 1}{d} } \bracket{\psi}{ \I \otimes ( B_{x_{1}} - B_{x_{2}} )^{2} }{\psi} = 1,
\end{align}
where we have used the fact that $\sin \theta = \sqrt{d/(d + 1)}$. This immediately implies that $\Tr( F \tau ) = 0$ and therefore Eq.~\eqref{eq:Z-vanishes} holds. Combining Eqs.~\eqref{eq:Z+0+Id} and~\eqref{eq:Z-vanishes} gives
\begin{equation}
	\Tr( 0 \oplus 0 \oplus \I \, \rho_{B} ) = 0.
\end{equation}
Since we are assuming that $\rho_{B}$ is full-rank, this implies that $\dim \cH_{B_{3}} = 0$.

In the last part we argue that this conclusion holds even without the uniformity assumption. To do so note that $\dim \cH_{B_{3}} = 0$ is equivalent to the operator $B_{x_{1}} + B_{x_{2}}$ not having an eigenvalue of $2$. We show that every quantum realisation can be transformed into a symmetrised realisation which satisfies the uniformity condition. Moreover, as the symmetrised operators inherit their spectra from the original ones, we can deduce that the original measurement operators of Bob do not have an eigenvalue of $2$.

Suppose we are given a quantum realisation on $\cH_{A} \otimes \cH_{B}$ given by the state $\rho_{AB}$ (with $\rho_{B}$ being full-rank) and the operators $A_{x}, B_{y}$, which achieves the maximal violation. For $\sigma$ being a permutation of $[d^{2}]$ consider the measurement operators given by
\begin{align}
	B^{\sigma}_{y} &\equiv B_{\sigma(y)},\\
	A^{\sigma}_{x} &\equiv
	\begin{cases}
		A_{\sigma(x_{1}) \sigma(x_{2})} &\nbox{if} \sigma(x_{1}) < \sigma(x_{2}),\\
		- A_{\sigma(x_{2}) \sigma(x_{1})} &\nbox{otherwise.}
	\end{cases}
\end{align}
It is straighforward to check that these operators also achieve the maximal violation on the state $\rho_{AB}$. Now by taking a convex combination over all permutations, we construct a symmetrised quantum realisation. We label the permutations of $[d^{2}]$ by $\sigma_{j}$, where $j \in \{1, 2, \ldots, N \}$ and $N = (d^{2}) !$. The symmetrised realisation acts on $\bC^{N} \otimes \bC^{N} \otimes \cH_{A} \otimes \cH_{B}$, where we have added two $N$-dimensional registers, one for each party, to serve as classical shared randomness. The symmetrised realisation reads
\begin{align}
	\rho_{AB}' &\equiv \frac{1}{N} \sum_{j = 1}^{N} \ketbraq{j} \otimes \ketbraq{j} \otimes \rho_{AB},\\
	A_{x}' &\equiv \sum_{j = 1}^{N} \ketbraq{j} \otimes A_{x}^{\sigma},\\
	B_{y}' &\equiv \sum_{j = 1}^{N} \ketbraq{j} \otimes B_{y}^{\sigma}.
\end{align}
The symmetrised realisation still satisfies the condition that the marginal of Bob is full-rank and, moreover, all the expectation values are uniform. Therefore, the operators $B_{1}'$ and $B_{2}'$ must be of the form given in Eq.~\eqref{eq:Bx1-Bx2-final}, which in particular implies that the eigenvalue of $2$ does not belong to the spectrum of $B_{1}' + B_{2}'$. On the other hand, it is easy to see that
\begin{equation}
	\spec( B_{1}' + B_{2}' ) = \bigcup_{ x_{1} < x_{2} } \spec ( B_{x_{1}} + B_{x_{2}} ),
\end{equation}
which implies that the original operators $B_{x_{1}}$ and $B_{x_{2}}$ must also be of the form given in Eq.~\eqref{eq:Bx1-Bx2-final}.


\begin{thebibliography}{99}
	
\bibitem{KS}	
S. Kochen, and E. P. Specker, The Problem of Hidden Variables in Quantum Mechanics, Indiana University Mathematics Journal \textbf{17}, 59 (1967).	

\bibitem{Bell}
J. S. Bell, On the Einstein Podolsky Rosen Paradox, Physics. \textbf{1}, 195 (1964).

\bibitem{Qcrypt}
N. Gisin, G. Ribordy, W. Tittel, and H. Zbinden,
Quantum cryptography,
\href{https://doi.org/10.1103/RevModPhys.74.145}{Rev. Mod. Phys. \textbf{74}, 145 (2002).}


\bibitem{Nielsen}
M. Nielsen and I. L. Chaung,
Quantum Computation and Quantum Information,
Cambridge University Press.


\bibitem{Qmeas}
P. Busch, P. Lahti, J-P. Pellonp\"a\"a and K. Ylinen.
Quantum Measurement,
Springer.


	
\bibitem{Schwinger}
J. Schwinger,
Unitary operator bases, 
\href{https://doi.org/10.1073/pnas.46.4.570}{PNAS \textbf{46}, 570-579 (1960).}

\bibitem{Zauner}
G. Zauner,
Quantendesigns, Grundz\"uge einer nichtkommutativen Designtheorie,
Doctoral thesis at University of Vienna (1999).

\bibitem{Renes}
J. M. Renes, R. Blume-Kohout, A. J. Scott, C. M. Caves,
Symmetric Informationally Complete Quantum Measurements,
\href{https://doi.org/10.1063/1.1737053}{J. Math. Phys. \textbf{45}, 2171 (2004).}




\bibitem{Wooters1}
W. K. Wooters,
Quantum Measurements and Finite Geometry,
\href{https://doi.org/10.1007/s10701-005-9008-x}{ Found Phys \textbf{36}, 112  (2006).}


\bibitem{Grassl1}
M. Grassl,
On SIC-POVMs and MUBs in Dimension 6,
\href{https://arxiv.org/abs/quant-ph/0406175}{arXiv:quant-ph/0406175}


\bibitem{Beneduci}
R. Beneduci, T. J. Bullock, P. Busch, C. Carmeli, T. Heinosaari, and A. Toigo,
Operational link between mutually unbiased bases and symmetric informationally complete positive operator-valued measures,
\href{https://doi.org/10.1103/PhysRevA.88.032312}{Phys. Rev. A \textbf{88}, 032312 (2013).}

\bibitem{Bengtsson1}
I. Bengtsson,
From SICs and MUBs to Eddington,
\href{https://doi.org/10.1088/1742-6596/254/1/012007}{J. Phys. Conf. Ser. \textbf{254}  012007 (2010).}

\bibitem{BengtssonCabello}
I. Bengtsson, K. Blanchfield, and A. Cabello,
A Kochen-Specker inequality from a SIC,
\href{https://doi.org/10.1016/j.physleta.2011.12.011}{Phys. Lett. A \textbf{376},  374 (2012).}


\bibitem{Rastegin}
A. E. Rastegin,
Uncertainty relations for MUBs and SIC-POVMs in terms of generalized entropies,
\href{https://doi.org/10.1140/epjd/e2013-40453-2}{Eur. Phys. J. D  \textbf{67}, 269 (2013).}








\bibitem{MUBreview}
T. Durt, B-G. Englert, I. Bengtsson and K. \.Zyczkowski,
On mutually unbiased bases,
\href{https://doi.org/10.1142/S0219749910006502}{Int. J. Quantum Information \text{8}, 535 (2010).}

\bibitem{MaassenUffink}
H. Maassen and J. B. M. Uffink,
Generalized entropic uncertainty relations,
\href{https://doi.org/10.1103/PhysRevLett.60.1103}{Phys. Rev. Lett. \textbf{60}, 1103 (1988).}



\bibitem{BB84}
C. H. Bennett and G. Brassard,
Quantum cryptography: Public key distribution and coin tossing,
In Proceedings of IEEE International Conference on Computers, Systems and Signal Processing, volume 175, page 8. New York, 1984.

\bibitem{E91}
A. K. Ekert,
Quantum cryptography based on Bell's theorem,
\href{https://doi.org/10.1103/PhysRevLett.67.661}{Phys. Rev. Lett. \textbf{67}, 661 (1991).}

\bibitem{6state}
D. Bru\ss,
Optimal Eavesdropping in Quantum Cryptography with Six States,
\href{https://doi.org/10.1103/PhysRevLett.81.3018}{Phys. Rev. Lett. \textbf{81}, 3018 (1998).}



\bibitem{Cerf02}
N. J. Cerf, M. Bourennane, A. Karlsson and N. Gisin,
Security of quantum key distribution using d-level systems,
\href{https://doi.org/10.1103/PhysRevLett.88.127902}{Phys. Rev. Lett. \textbf{88}, 127902 (2002).}

\bibitem{SARG}
V. Scarani, A. Ac\'in, G. Ribordy, and N. Gisin,
Quantum Cryptography Protocols Robust against Photon Number Splitting Attacks for Weak Laser Pulse Implementations,
\href{https://doi.org/10.1103/PhysRevLett.92.057901}{Phys. Rev. Lett. \textbf{92}, 057901 (2004).}


\bibitem{QSS0}
M. Hillery, V. Bu\v{z}ek, and A. Berthiaume
Quantum secret sharing,
\href{https://doi.org/10.1103/PhysRevA.59.1829}{Phys. Rev. A \textbf{59}, 1829 (1999).}


\bibitem{QSS1}
I-Ching Yu, Feng-Li Lin, and Ching-Yu Huang,
Quantum secret sharing with multilevel mutually (un)biased bases,
\href{https://doi.org/10.1103/PhysRevA.78.012344}{Phys. Rev. A \textbf{78}, 012344 (2008)}

\bibitem{QSS2}
A. Tavakoli, I. Herbauts, M. \.Zukowski, and M. Bourennane,
Secret sharing with a single d-level quantum system,
\href{https://doi.org/10.1103/PhysRevA.92.030302}{Phys. Rev. A \textbf{92}, 030302(R) (2015).}







\bibitem{Wooters}
W. K. Wootters and B. D. Fields,
Optimal state-determination by mutually unbiased measurements,
\href{https://doi.org/10.1016/0003-4916(89)90322-9}{Ann. Phys. \textbf{191}, 363 (1989).}

\bibitem{Adamson}
R. B. A. Adamson and A. M. Steinberg,
Improving Quantum State Estimation with Mutually Unbiased Bases,
\href{https://doi.org/10.1103/PhysRevLett.105.030406}{Phys. Rev. Lett. \textbf{105}, 030406 (2010).}

\bibitem{Ambainis}
A. Ambainis, A. Nayak, A. Ta-Shma, U. Vazirani, 
Dense quantum coding and a lower bound for 1-way quantum automata,
Proceedings of the 31st Annual ACM Symposium on Theory of Computing (STOC'99), pp. 376-383, 1999.


\bibitem{Tavakoli15}
A. Tavakoli, A. Hameedi, B. Marques, and M. Bourennane,
Quantum Random Access Codes Using Single d-Level Systems,
\href{https://doi.org/10.1103/PhysRevLett.114.170502}{Phys. Rev. Lett. \textbf{114}, 170502 (2015)}

\bibitem{Aguilar}
E. A. Aguilar, J. J. Borka\l{}a, P. Mironowicz, and M. Paw\l{}owski,
Connections between Mutually Unbiased Bases and Quantum Random Access Codes,
\href{https://doi.org/10.1103/PhysRevLett.121.050501}{Phys. Rev. Lett. \textbf{121}, 050501 (2018).}


\bibitem{Tavakoli18}
A. Tavakoli, J. Kaniewski, T. V\'ertesi, D. Rosset, and N. Brunner,
Self-testing quantum states and measurements in the prepare-and-measure scenario,
\href{https://doi.org/10.1103/PhysRevA.98.062307}{Phys. Rev. A \textbf{98}, 062307 (2018).}

\bibitem{Farkas19}
M. Farkas and J. Kaniewski,
Self-testing mutually unbiased bases in the prepare-and-measure scenario,
\href{https://doi.org/10.1103/PhysRevA.99.032316}{Phys. Rev. A \textbf{99}, 032316 (2019).}


\bibitem{Gottesman}
D. Gottesman,
Class of quantum error-correcting codes saturating the quantum Hamming bound,
\href{https://doi.org/10.1103/PhysRevA.54.1862}{Phys. Rev. A \textbf{54}, 1862 (1996).}

\bibitem{Calderbank}
A. R. Calderbank, E. M. Rains, P. W. Shor, and N. J. A. Sloane,
Quantum Error Correction and Orthogonal Geometry,
\href{https://doi.org/10.1103/PhysRevLett.78.405}{Phys. Rev. Lett. \textbf{78}, 405 (1997).}


\bibitem{Spengler}
C. Spengler, M. Huber, S. Brierley, T. Adaktylos, and B. C. Hiesmayr,
Entanglement detection via mutually unbiased bases,
\href{https://doi.org/10.1103/PhysRevA.86.022311}{Phys. Rev. A \textbf{86}, 022311 (2012).}

%
%
%







\bibitem{Caves}
C. M. Caves, C. A. Fuchs, and R. Schack, 
Unknown quantum states: The Quantum de Finetti representation,
\href{https://doi.org/10.1063/1.1494475}{J. Math. Phys. \textbf{43}, 4537 (2002).}





\bibitem{Medendorp}
Z. E. D. Medendorp, F. A. Torres-Ruiz, L. K. Shalm, G. N. M. Tabia, C. A. Fuchs, and A. M. Steinberg,
Experimental characterization of qutrits using symmetric informationally complete positive operator-valued measurements,
\href{https://doi.org/10.1103/PhysRevA.83.051801}{Phys. Rev. A \textbf{83}, 051801(R) (2011).}

\bibitem{Pimenta}
W. M. Pimenta, B. Marques, T. O. Maciel, R. O. Vianna, A. Delgado, C. Saavedra, and S. P\'adua,
Minimum tomography of two entangled qutrits using local measurements of one-qutrit symmetric informationally complete positive operator-valued measure,
\href{https://doi.org/10.1103/PhysRevA.88.012112}{Phys. Rev. A \textbf{88}, 012112 (2013).}

\bibitem{Bent}
N. Bent, H. Qassim, A. A. Tahir, D. Sych, G. Leuchs, L. L. S\'anchez-Soto, E. Karimi, and R. W. Boyd,
Experimental Realization of Quantum Tomography of Photonic Qudits via Symmetric Informationally Complete Positive Operator-Valued Measures,
\href{https://doi.org/10.1103/PhysRevX.5.041006}{Phys. Rev. X \textbf{5}, 041006 (2015).}



\bibitem{Renes2}
J. M. Renes,
Equiangular Spherical Codes in Quantum Cryptography,
Quant. Inf. Comput. \textbf{5}, 080 (2005).

\bibitem{Singapore}
B-G. Englert, D. Kaszlikowski, H. K. Ng, W. K. Chua, J. \v{R}eh\'a\v{c}ek, and Janet Anders,
Efficient and Robust Quantum Key Distribution With Minimal State Tomography,
\href{https://arxiv.org/abs/quant-ph/0412075v4}{arXiv:quant-ph/0412075v4}

\bibitem{Bouchard}
F. Bouchard, K. Heshami, D. England, R. Fickler, R. W. Boyd, B-G. Englert, L. L. S\'anchez-Soto, and E. Karimi,
Experimental investigation of high-dimensional quantum key distribution protocols with twisted photons,
\href{https://doi.org/10.22331/q-2018-12-04-111}{Quantum \textbf{2}, 111 (2018).}



\bibitem{Acin}
A. Ac\'in, S. Pironio, T. V\'ertesi, and P. Wittek,
Optimal randomness certification from one entangled bit,
\href{https://doi.org/10.1103/PhysRevA.93.040102}{Phys. Rev. A \textbf{93}, 040102(R) (2016).}


\bibitem{TavakoliSIC}
A. Tavakoli, D. Rosset, and M-O. Renou,
Enabling Computation of Correlation Bounds for Finite-Dimensional Quantum Systems via Symmetrization,
\href{https://doi.org/10.1103/PhysRevLett.122.070501}{Phys. Rev. Lett. \textbf{122}, 070501 (2019)}


\bibitem{Piotr}
P. Mironowicz and M. Paw\l{}owski,
Experimentally feasible semi-device-independent certification of four-outcome positive-operator-valued measurements,
\href{https://doi.org/10.1103/PhysRevA.100.030301}{Phys. Rev. A \textbf{100}, 030301(R) (2019).}

\bibitem{TavakoliNonProj}
A. Tavakoli, M. Smania, T. V\'ertesi, N. Brunner, and M. Bourennane,
Self-testing non-projective quantum measurements in prepare-and-measure experiments,
\href{https://arxiv.org/abs/1811.12712}{arXiv:1811.12712}

\bibitem{Massi}
M. Smania, P. Mironowicz, M. Nawareg, M. Paw\l{}owski, A. Cabello, and M. Bourennane,
Experimental device-independent certification of a symmetric, informationally complete, positive operator-valued measure,
\href{https://arxiv.org/abs/1811.12851}{arXiv:1811.12851}


\bibitem{Shang}
J. Shang, A. Asadian, H. Zhu, and O. G\"uhne,
Enhanced entanglement criterion via symmetric informationally complete measurements,
\href{https://doi.org/10.1103/PhysRevA.98.022309}{Phys. Rev. A \textbf{98}, 022309 (2018).}

\bibitem{Bae}
J. Bae, B. C. Hiesmayr, and D. McNulty,
Linking entanglement detection and state tomography via quantum 2-designs,
\href{https://doi.org/10.1088/1367-2630/aaf8cf}{New J. Phys. \textbf{21} 013012 (2019).}



\bibitem{FuchsRev}
C. A. Fuchs and R. Schack,
Quantum-Bayesian coherence,
\href{https://doi.org/10.1103/RevModPhys.85.1693}{Rev. Mod. Phys. \textbf{85}, 1693 (2013).}



\bibitem{Appelby1}
D. M. Appleby, C. A. Fuchs, and H. Zhu,
Group theoretic, lie algebraic and Jordan algebraic formulations of the sic existence problem,
Quantum Inf. Comput. \textbf{15}, 61 (2015).


\bibitem{Appelby}
M. Appleby, S. Flammia, G. McConnell, and J. Yard,
SICs and Algebraic Number Theory,
\href{https://doi.org/10.1007/s10701-017-0090-7}{ Found Phys \textbf{47}, 1042 (2017).}


\bibitem{ScottGrassl}
A. J. Scott, and M. Grassl,
SIC-POVMs: A new computer study,
\href{https://doi.org/10.1063/1.3374022}{J. Math. Phys. \textbf{51}, 042203 (2010)}



\bibitem{Scott}
A. J. Scott,
SICs: Extending the list of solutions,
\href{https://arxiv.org/abs/1703.03993}{arXiv:1703.03993}




\bibitem{FuchsReview}
C. A. Fuchs, M. C. Hoang, and B. C. Stacey,
The SIC Question: History and State of Play,
\href{https://doi.org/10.3390/axioms6030021}{Axioms \textbf{21}, 6 (2017).}


\bibitem{Brunner}
N. Brunner, D. Cavalcanti, S. Pironio, V. Scarani, and S. Wehner,
Bell nonlocality,
\href{https://doi.org/10.1103/RevModPhys.86.419}{Rev. Mod. Phys. \textbf{86}, 419 (2014).}


\bibitem{Supic}
I. \v{S}upi\'c, and J. Bowles,
Self-testing of quantum systems: a review,
\href{https://arxiv.org/abs/1904.10042}{arXiv:1904.10042}
















\bibitem{Vidick}
D. Ostrev and T. Vidick,
The structure of nearly-optimal quantum strategies for the CHSH (n) XOR games,
Quantum Information \& Computation \textbf{16}, 1191 2016.


\bibitem{Coladangelo}
A. Coladangelo, K. T. Goh and V. Scarani,
All pure bipartite entangled states can be self-tested,
\href{https://doi.org/10.1038/ncomms15485 (2017)}{Nature Communications \textbf{8}, 15485 (2017).}


\bibitem{Jed}
J. Kaniewski, I. \v{S}upi\'c, J. Tura, F. Baccari, A. Salavrakos, and R. Augusiak,
Maximal nonlocality from maximal entanglement and mutually unbiased bases, and self-testing of two-qutrit quantum systems,
\href{https://doi.org/10.22331/q-2019-10-24-198}{Quantum \textbf{3}, 108 (2019).}

\bibitem{Sarkar}
S. Sarkar, D. Saha, J. Kaniewski, and R. Augusiak,
Self-testing quantum systems of arbitrary local dimension with minimal number of measurements,
\href{https://arxiv.org/abs/1909.12722}{arXiv:1909.12722}




\bibitem{Romero}
M. J. Kewming, S. Shrapnel, A. G. White, and J. Romero,
Hiding Ignorance Using High Dimensions,
\href{arXiv:1903.09487}{https://arxiv.org/abs/1903.09487}

\bibitem{Hensen}
B. Hensen et. al.,
Loophole-free Bell inequality violation using electron spins separated by 1.3 kilometres,
\href{https://doi.org/10.1038/nature15759}{Nature \textbf{526}, 682 (2015).}

\bibitem{Shalm}
L. K. Shalm, 
Strong Loophole-Free Test of Local Realism,
\href{https://doi.org/10.1103/PhysRevLett.115.250402}{Phys. Rev. Lett. \textbf{115}, 250402 (2015).}

\bibitem{Giustina}
M. Giustina, et. al., 
Significant-Loophole-Free Test of Bell's Theorem with Entangled Photons,
\href{https://doi.org/10.1103/PhysRevLett.115.250401}{Phys. Rev. Lett. \textbf{115}, 250401 (2015).}


\bibitem{Harald}
W. Rosenfeld, D. Burchardt, R. Garthoff, K. Redeker, N. Ortegel, M. Rau, and H. Weinfurter,
Event-Ready Bell Test Using Entangled Atoms Simultaneously Closing Detection and Locality Loopholes,
\href{https://doi.org/10.1103/PhysRevLett.119.010402}{Phys. Rev. Lett. \textbf{119}, 010402 (2017).}




\bibitem{ji08a}
S.-W. Ji, J. Lee, J. Lim, K. Nagata, and H.-W. Lee,
Multisetting Bell inequality for qudits,
\href{https://doi.org/10.1103/PhysRevA.78.052103}{Phys. Rev. A \textbf{78}, 052103 (2008).}

\bibitem{liang09a}
Y.-C. Liang, C.-W. Lim, and D.-L. Deng,
Reexamination of a multisetting Bell inequality for qudits,
\href{https://doi.org/10.1103/PhysRevA.80.052116}{Phys. Rev. A \textbf{80}, 052116 (2009).}

\bibitem{lim10a}
J. Lim, J. Ryu, S. Yoo, C. Lee, J. Bang, and J. Lee,
Genuinely high-dimensional nonlocality optimized by complementary measurements,
\href{https://doi.org/10.1088/1367-2630/12/10/103012}{New J. Phys. \textbf{12}, 103012 (2010).}

\bibitem{Bechmann}
H. Bechmann-Pasquinucci and N. Gisin, 
Bell inequality for quNits with binary measurements,
Quantum Inf. Comput. \textbf{3}, 157 (2003).


\bibitem{CHSH}
J. F. Clauser, M. A. Horne, A. Shimony, and R. A. Holt,
Proposed Experiment to Test Local Hidden-Variable Theories,
\href{https://doi.org/10.1103/PhysRevLett.23.880}{Phys. Rev. Lett. \textbf{23}, 880 (1969).}





\bibitem{Brierley}
S. Brierley, S. Weigert, and I. Bengtsson,
All Mutually Unbiased Bases in Dimensions Two to Five,
Quantum Info. \& Comp.  \textbf{10}, 0803 (2010).

\bibitem{Jeba}
C. Jebarathinam, J-C. Hung, S-L. Chen, and Y-C. Liang,
Maximal violation of a broad class of Bell inequalities and its implication on self-testing,
\href{https://doi.org/10.1103/PhysRevResearch.1.033073}{Phys. Rev. Research \textbf{1}, 033073 (2019).}

\bibitem{Kaniewski2}
J. Kaniewski,
A weak form of self-testing,
\href{https://arxiv.org/abs/1910.00706}{arXiv:1910.00706}

\bibitem{Tasca}
D. S. Tasca, P. S{\'a}nchez, S. P. Walborn, and {\L}. Rudnicki,
Mutual unbiasedness in coarse-grained continuous variables,
\href{https://doi.org/10.1103/PhysRevLett.120.040403}{Phys. Rev. Lett. \textbf{120}, 040403 (2018).}

\bibitem{Krishna}
M. Krishna and K. R. Parthasarathy,
An Entropic Uncertainty Principle for Quantum Measurements,
\href{https://www.jstor.org/stable/25051432?seq=1#metadata_info_tab_contents}{Indian Journal of Statistics \textbf{64}, (3) 842 (2002).}

\bibitem{HMZ16}
T. Heinosaari, T. Miyadera, and M. Ziman,
An invitation to quantum incompatibility,
\href{https://doi.org/10.1088/1751-8113/49/12/123001}{Journal of Physics A: Mathematical
and Theoretical \textbf{49}, (12) 123001 (2016).}

\bibitem{Haapasalo}
E. Haapasalo,
Robustness of incompatibility for quantum devices,
\href{https://doi.org/10.1088/1751-8113/48/25/255303}{Journal of Physics A: Mathematical
and Theoretical \textbf{48}, (25) 255303 (2015).}

\bibitem{QKD}
N. Gisin, G. Ribordy, W. Tittel, and H. Zbinden,
Quantum cryptography,
\href{https://doi.org/10.1103/RevModPhys.74.145}{Rev. Mod. Phys. \textbf{74}, 145 (2002).}




\bibitem{Masanes}
L. Masanes, S. Pironio, and A. Ac\'in, 
Secure device-independent quantum key distribution with causally independent measurement devices,
\href{https://doi.org/10.1038/ncomms1244}{Nature Communications  \textbf{2}, 238 (2011).}

\bibitem{Franz}
T. Franz, F. Furrer, and R. F. Werner,
Extremal quantum correlations and cryptographic security,
\href{https://doi.org/10.1103/PhysRevLett.106.250502}{Phys. Rev. Lett. \textbf{106}, 250502 (2011).}


	


\bibitem{NPA}	
M. Navascu\'es, S. Pironio, and A. Ac\'in,
Bounding the Set of Quantum Correlations,
\href{https://doi.org/10.1103/PhysRevLett.98.010401}{Phys. Rev. Lett. \textbf{98}, 010401 (2007).}

\bibitem{Rosset}
D. Rosset,
SymDPoly: symmetry-adapted moment relaxations for noncommutative polynomial optimization,
\href{https://arxiv.org/abs/1808.09598}{arXiv:1808.09598}


\bibitem{Cai}
Y. Cai, J-D. Bancal, J. Romero, and V. Scarani,
A new device-independent dimension witness and its experimental implementation
\href{https://doi.org/10.1088\%2F1751-8113\%2F49\%2F30\%2F305301}{Journal of Physics A: Mathematical and Theoretical \textbf{49}, 305301 (2016)}


\bibitem{sedumi}
J. F. Sturm,
Using SeDuMi 1.02, A Matlab toolbox for optimization over symmetric cones,
\href{https://doi.org/10.1080/10556789908805766}{Optimization Methods and Software \textbf{11}, 625 (1999).}


\bibitem{ElegantBell}
N. Gisin,
Bell inequalities: many questions, a few answers,
\href{https://arxiv.org/abs/quant-ph/0702021}{arXiv:quant-ph/0702021}


\bibitem{Plato}
A. Tavakoli and N. Gisin,
Platonic solids and fundamental tests of quantum mechanics,
(in preparation).


\bibitem{Pironio2}
S. Pironio, A. Ac\'in, S. Massar, A. Boyer de la Giroday, D. N. Matsukevich, P. Maunz, S. Olmschenk, D. Hayes, L. Luo, T. A. Manning, and C. Monroe 
Random numbers certified by Bell's theorem,
\href{https://doi.org/10.1038/nature09008}{Nature \textbf{464}, 1021 (2010).}



\bibitem{Ariano}
G. M. D'Ariano, P. L. Presti, and P. Perinotti,
Classical randomness in quantum measurements,
\href{https://doi.org/10.1088/0305-4470/38/26/010}{J. Phys. A: Math. Gen. \textbf{38}, 5979 (2005).}

\bibitem{replabWebsite}
\href{https://replab.github.io/replab/}{https://replab.github.io}

\bibitem{replab}
D. Rosset, F. Montealegre-Mora, J-D. Bancal,
RepLAB: a computational/numerical approach to representation theory,
\href{https://arxiv.org/abs/1911.09154}{arXiv:1911.09154.}

\bibitem{sdpa}
\href{http://sdpa.sourceforge.net}{http://sdpa.sourceforge.net}




\bibitem{Miguel}
M. Navascu{\'e}s, S. Pironio, and A. Ac{\'i}n,
SDP relaxations for non-commutative polynomial optimization,
\href{https://doi.org/10.1007/978-1-4614-0769-0\_21}{Handbook on semidefinite, conic and polynomial optimization. International series in operations research \& management science \textbf{166}, 601 (2012).}

\bibitem{PPT}
M.~Horodecki, P.~Horodecki, and R.~Horodecki,
Separability of mixed states: necessary and sufficient conditions,
\href{https://doi.org/10.1016/S0375-9601(96)00706-2}{Physics Letters A. \textbf{223}, 1--8 (1996)}

\bibitem{Boyd}
S.~Boyd and L.~Vandenberghe,
Convex optimization,
Cambridge university press, 2004

\bibitem{Choi}
M.~D.~Choi, Completely positive linear maps on complex matrices,
Linear algebra and its applications, \textbf{10}(3), 285-290 (1975)

\bibitem{Jamiolkowski}
A.~Jamio\l kowski, Linear transformations which preserve trace and positive semidefiniteness of operators,
Reports on Mathematical Physics, \textbf{3}(4), 275-278 (1972)


\bibitem{Designolle}
S. Designolle, M. Farkas, and J. Kaniewski,
Incompatibility robustness of quantum measurements: a unified framework,
\href{https://doi.org/10.1088/1367-2630/ab5020}{New J. Phys. \textbf{21}, 113053 (2019).}

\bibitem{Kaniewski}
J. Kaniewski, M. Tomamichel, and S. Wehner,
Entropic uncertainty from effective anticommutators,
\href{https://doi.org/10.1103/PhysRevA.90.012332}{Phys. Rev. A \textbf{90}, 012332 (2014).}

\end{thebibliography}
\end{document}